\documentclass[runningheads,a4paper]{article}

\usepackage{amsmath,amsthm}
\usepackage{framed}
\usepackage[utf8x]{inputenc}
\usepackage[dvipsnames]{xcolor}
\usepackage{xspace}
\usepackage{graphicx}
\usepackage{xcolor,cancel}
\usepackage{url}
\usepackage{rotating}
\usepackage{subfigure}
\usepackage{amsfonts} 
\usepackage{amssymb}

\newcommand{\etaspec}{\bar\eta}

\newcommand{\RR}{\mathbb{R}}
\newcommand{\norm}[1]{\| #1 \|}
\newcommand{\Id}{{\mathrm{Id}}}
\newcommand{\prox}{\mathrm{prox}}
\newcommand{\Ii}{\mathcal{I}}
\newcommand{\x}{x}
\newcommand{\xt}{\tilde x}
\newcommand{\p}{p}
\newcommand{\Ee}{\mathcal{E}}

\DeclareMathOperator{\argmin}{argmin}
\DeclareMathOperator{\supp}{supp}
\DeclareMathOperator{\sign}{sign}
\DeclareMathOperator{\Span}{Span}
\DeclareMathOperator{\diverg}{div}

\newcommand{\uargmin}[1]{\underset{#1}{\argmin}\;}
\newcommand{\snorm}[1]{\left\|#1\right\|}
\newcommand{\normb}[1]{\Big|\!\Big| #1 \Big |\!\Big|}
\newcommand{\enscond}[2]{ \left\{ #1 \;:\; #2 \right\} }
\newcommand{\dotp}[2]{\left\langle #1,\,#2 \right\rangle}
\newcommand{\abs}[1]{\lvert #1 \rvert}
\newcommand{\choice}[1]{ %
	\left\{ %
		\begin{array}{ll} #1 \end{array} %
	\right. }
	
\newcommand{\ifq}{{\text{if} \quad}}
\newcommand{\otherwise}{{\text{otherwise}}}
	
\newtheorem{defn}{Definition}
\newtheorem{prop}{Proposition}
\newcommand{\qandq}{ \quad \text{and} \quad }
\newcommand{\ie}{{\em i.e.,~}}
\newcommand{\eg}{{\em e.g.,~}}

\usepackage{graphicx}
\begin{document}

\title{Block based refitting in $\ell_{12}$ sparse regularisation\\
}

\date{}


\author{Charles-Alban Deledalle\\
  Department  of  Electrical  and  Computer  Engineering\\   University  of  California\\   San  Diego,   La  Jolla,   USA
  \and
  Nicolas Papadakis\\
  CNRS, Univ. Bordeaux, IMB, UMR 5251,\\ F-33400 Talence, France
  \and
  Joseph Salmon\\
  IMAG, Univ.   Montpellier,    CNRS,   Montpellier,    France
  \and
  Samuel Vaiter\\
  CNRS \& IMB, UMR 5584\\ Univ. Bourgogne Franche-Comté, F-21000 Dijon, France
}


\maketitle

\begin{abstract}
  In many linear regression problems, including ill-posed inverse problems in image restoration,
  the data exhibit some sparse structures that can be used to regularize the inversion.
  To this end, a classical path is to use $\ell_{12}$ block based regularization.
  While efficient at retrieving
  the inherent sparsity patterns of the data -- the support -- the estimated
  solutions are known to suffer from a systematical bias.
  We propose a general framework for removing this artifact
  by refitting the solution towards the data
  while preserving key features of its structure such as the support.
  This is done through the use of refitting block penalties that
  only act on the support of the estimated solution. Based on an analysis of
  related works in the literature, we introduce a new penalty that is well
  suited for refitting purposes.
  We also present a new algorithm to obtain the refitted
  solution along with the original (biased) solution for any convex refitting
  block penalty. Experiments illustrate the good behavior of the proposed
  block penalty for refitting solutions of Total Variation and Total Generalized
  Variation models.

\end{abstract}

\section{Introduction}\label{sec:intro}

We consider linear inverse problems of the form $y = \Phi x + w$,
where $y \in \RR^p$ is an observed degraded image, $x \in \RR^n$ the unknown clean image, $\Phi: \RR^{n} \to \RR^p$
a linear operator and $w \in \RR^p$ a noise component, typically a zero-mean
white Gaussian random vector with standard deviation $\sigma > 0$.
To reduce the effect of noise and the potential ill-conditioning of $\Phi$,
we consider a regularized least squares problem with a sparse analysis
regularization term based on an $\ell_{12}$ block pernalty of the form
\begin{align}\label{isotv}
  \hat{x} \in
  \uargmin{x \in \RR^n} \tfrac12 \norm{\Phi x - y }_2^2 + \lambda \norm{\Gamma x}_{1,2}\enspace,
\end{align}
where $\lambda > 0$ is a regularization parameter, $\Gamma : \RR^{n} \to \RR^{m \times b}$ is a linear analysis operator
mapping an image over $m$ blocks of size $b$, and for $z \in \RR^{m \times b}$
\begin{align}
  \norm{z}_{1,2}=\sum_{i=1}^m \norm{z_i}_2=\sum_{i=1}^m \Big( \sum_{j=1}^b z_{i,j}^2 \Big)^{1/2}~,
\end{align}
with $z_i=(z_{i,j})_{j=1}^b\in\RR ^b$.
Note that the terminology {\it sparse analysis} is used here by
opposition to {\it sparse synthesis} models as discussed in the seminal work of
\cite{elad2007analysis}.
The first term in \eqref{isotv} is a data fidelity term enforcing
$x$ to be close to $y$ through $\Phi$, while the second term enforces the so-called
group sparsity on $x$ (sometimes refered to as joint sparsity, block sparsity or structured sparsity)
capturing the organization of the data as encoded by $\Gamma$, see for instance
\cite{bach2012optimization}.

\subsection{Related examples}

A typical example is the Lasso (Least absolute shrinkage and selection operator) \cite{tibshirani1996regression}.
The Lasso is a statistical procedure used for variable selection and
relying on regularized linear least square regression as
expressed in eq.~\eqref{isotv} in which $\Gamma = \Id$, $m = n$ and $b = 1$.
The Lasso is known to promote sparse solutions, \ie such that $\hat x_k = 0$ for
most indices $1 \leq k \leq n$.
Since blocks are of size $b=1$, the regularization term boils down to the classical
$\ell_1$ sparsity term that is unstructured as no interactions between the elements of $\hat x$ are considered.
In this paper, we will focus instead on cases of block penalties
where $b > 1$.  The group Lasso \cite{Yuan_Lin06,Lin_Zhang06} is one of them,
for which $\Gamma$ is designed to reorganize the elements $\hat x_k$
into groups $\hat z_i = (\Gamma \hat x)_i$ supposedly meaningful
according to some prior knowledge on the data.
The group Lasso is known to promote block sparse solutions,
\ie such that $\hat z_i = 0_b$ for most of the groups $1 \leq i \leq m$.
Note that elements within a non-zero group $\hat z_i \ne 0_b$ are not required to be sparse.

Regarding image restoration applications,
the authors of \cite{peyre2011group} use an $\ell_{12}$ regularization term
where $\Gamma$ extracts $m=n$ overlapping blocks of wavelet coefficients of size $b$ or where $\Gamma$ use a dyadic decomposition of the wavelet coefficients into blocks of variable size but non-overlapping~\cite{peyre2011dyadic}.
Such strategy was also used in audio processing~\cite{yu2008audio} for denoising.
Another example that we will investigate here is the one of the total-variation (TV) \cite{rudin1992nonlinear}.
We can distinguish two different forms of TV models.
Anisotropic total-variation (TVaniso) \cite{esedoglu2004decomposition}, considers
$\Gamma$ the operator which concatenates the vertical and horizontal
components of the discrete gradients into a vector of size $m=2n$, hence $b=1$.
Isotropic total-variation (TViso) considers instead
$\Gamma = \nabla$ being the operator which extracts $m=n$
discrete image gradient vectors of size $b = 2$.
Unlike TVaniso, TViso jointly enforces vertical and horizontal components of
the gradient to be simultaneously zero.
Since TVaniso does not take into account interactions between both directions,
it over favors vertical and horizontal structures while TViso behaves similarly
in all directions, hence their name \cite{esedoglu2004decomposition}.
Both models promote sparsity of the discrete gradient field of the image,
and, as a result, their solutions are piece-wise constant.
A major difference is that
TVaniso favors constant regions that are rectangular-like shaped and separated
by sharp edges. On the other hand,  TViso favors constant regions that are rounded-like shaped, as it has been shown by
 studying the properties of TV reconstructions in terms of extreme points of the level sets of the regularisation functional  \cite{chambolle2010introduction,bredies2020sparsity,boyer2019representer} .
In practice, when considering classical discretizations,  constant areas of TViso solutions  are separated by fast but gradual transitions. These smooth transitions only appear in  particular directions that depend on the considered discretization scheme on the regular image grid (see \cite{Caillaud} for a detailed analysis). These numerical artifacts can be solved with more advanced schemes \cite{condat2017discrete,chambolle2020crouzeix,Caillaud}.

Nevertheless, as TV is designed to promote piece-wise constant solutions, it is known to produce {\it staircasing} artifacts
that are all the more harmful as the images contain shaded objects
\cite{dobson1996recovery,chan2000high}. To reduce this effect,
the authors of \cite{liu2015image} suggested using an $\ell_{12}$ block sparsity
term not only by grouping vertical and horizontal components of the gradient,
but by grouping neighboring gradients in overlapping patches.
An alternative to reduce staircasing, that we will also investigate here, is
the second order Total Generalized Variation (TGV) model \cite{bredies2010total}
that promotes piece-wise affine solutions. As we will see,
TGV is another example of models that falls into this type of least squares problems regularized
with a sparse analysis term based on $\ell_{12}$ block penalties.
Sparsity in that case encodes that sought images are composed of
few shaded regions with few variations of slopes and separated by edges.

\subsection{Support of the solution}

Solutions of the sparse analysis regularization
model in \eqref{isotv} are known to be sparse
\cite{elad2007analysis,nam2013cosparse}, \ie such that
$(\Gamma x)_i = 0_b$ for most blocks $1 \leq i \leq m$.
It results that a key notion, central to all of these estimators, is
the one of support, \ie the set of non-zero blocks in $\Gamma \hat x$, defined as
\begin{align}
  \hat{\Ii} = \supp(\Gamma \hat x) &= \enscond{1 \leq i \leq m}{(\Gamma \hat x)_i \ne 0_b}~.
\end{align}
For the group Lasso, the support is typically used to identify groups of
covariates (columns of $\Phi$) being explanatory variables
for the dependent variable $y$ (\ie significantly correlated with $y$).
For TViso, the support is the set of pixel indices where
transitions occur in the restored image. In general, the support plays
an important role as it captures the intrinsic structural information underlying the data.
While being biased, in practice,
the estimate $\hat{x}$ obtained by sparse analysis regularization \eqref{isotv}
recover quite correctly the support $\supp(\Gamma x)$ of the underlying sparse signal $x$.
Under some additional assumptions, support recovery is
even proven to be exact as proved in \cite{vaiter2012robust} for $b=1$ (anisotropic case) and \cite{vaiter2017model} for $b \geq 1$.

\subsection{Bias of the solution}

Though the support $\hat \Ii$ of $\hat{x}$ can be very close to the one
of the sought image $x$,
the estimated amplitudes $\hat{x}_i$ suffers from a systematical bias.
When $\Phi = \Id$, the Lasso corresponds to the soft-thresholding (ST) operator
\begin{align}
  \hat{x}_i = {\rm ST}(y_i, \lambda) = \max(y_i - \lambda \sign y_i, 0)
\end{align}
for which all non-zero elements of the solutions are shrinked towards $0$ by a shift $\pm \lambda$
resulting to under- and over-estimated values.
With TViso, this bias is reflected by a loss of contrast in the image
since the amplitudes of some regions are regressed towards
the mean of the image \cite{strong2003edge,vaiter2013local,vaiter2017degrees}.
In TGV, not only a loss of contrast results from this bias,
but we observe that the slopes in areas of transitions
are often mis-estimated~\cite{kongskov2019directional}.

\subsection{Boosting approaches}

Given the artifacts induced by the $\ell_{12}$ sparse regularization,
many approaches have been developed to re-enhance the quality of
the solutions, {\it e.g.}, to reduce the loss of contrast and staircasing of TViso.
We refer to these approaches as boosting.
Most of them consist in solving \eqref{isotv}
iteratively based on the residue $\Phi \hat{x} - y$, or a related quantity,
obtained during the previous iterations. Among them, the well-known Bregman
iterations \cite{Osher} is often considered to recover part of
the loss of contrast for TViso. Other related procedures
are twicing \cite{Tukey77}, boosting with the $\ell_2$ loss
\cite{Buhlmann_Yu03},
unsharp residual iteration \cite{charest2008iterative},
SAIF-boosting \cite{milanfar2013tour,talebi2013saif},
ideal spectral filtering in the analysis sense \cite{gilboa2014total}
and SOS-boosting \cite{romano2015boosting}. While these approaches reduce
the bias in the estimated amplitudes, the support $\hat \Ii$ of the
original solution is not guaranteed to be preserved in the boosted solution,
even though this one may correspond to the support of the sought image $x$.

\subsection{Projection on the support}

Given the key role of the support of solutions of \eqref{isotv},
we believe that it is of main importance that a re-enhanced solution $\tilde x$ preserves it, \ie
such that $\supp(\Gamma x) \subseteq \hat{\Ii}$. For this reason, we focus
on re-fitting strategies that, unlike boosting, reduce the bias while
preserving the support of the original solution.
In the Lasso ($\Gamma=\Id$, $m=n$ and $b=1$) \cite{tibshirani1996regression},
a well known re-fitting scheme consists in performing {\it a posteriori}
a least-square re-estimation of the non-zero coefficients of the solution.
This post re-fitting technique became popular under various names in the
statistical literature:
Hybrid Lasso \cite{efron2004least}, Lasso-Gauss \cite{Rigollet_Tsybakov11},
OLS post-Lasso \cite{Belloni_Chernozhukov13},
Debiased Lasso (see \cite{Lederer13,Belloni_Chernozhukov13} for extensive
details on the subject).
Such approaches consists in approximating $y$ through $\Phi$
by an image sharing the same support as $\hat{x}$:
\begin{equation}\label{proj_supp}
  \tilde{x}^{\supp} \in \uargmin{x; \; \supp(\Gamma x) \subseteq \hat{\Ii}}
  \tfrac12 \norm{\Phi x - y}_2^2 \enspace,
\end{equation}
where $\hat{\Ii} = \supp(\Gamma \hat{x})$.
While this strategy works well for blocks of size $b=1$,
\eg for the Lasso or TVaniso,
it suffers from an excessive increase of variance whenever $b\geq 2$
(see \cite{deledalle2016clear} for illustrations on TViso).
This is due to the fact that solutions do not only present sharp edges, but may involve larger regions of gradual transitions where the debiasing may be too strong by re-introducing too much noise.
To cope with this issue, additional features of $\hat x$ than its support
must be also preserved by a refitting procedure.

\subsection{Advanced refitting strategies}
For the Lasso, it has been observed that  a pointwise preservation of the sign of $\hat x_i$ onto the support improves the numerical performances of the refitting \cite{chzhen2019}.
For $b=2$ and TViso like models, the joint projection on the  support with conservation of the  direction (or orientation) of $(\Gamma \hat x)_i$ has  been proposed in  \cite{brinkmann2016bias}.
Extension to second order regularization such as TGV \cite{bredies2010total} are   investigated in \cite{burger2019convergence} in the context of partially order spaces and approximate operators $\Phi$.
In a parallel line of research, it has been proposed in \cite{weiss2019contrast} to respect the inclusion of the level lines of $\hat x$ in the refitting by solving an isotonic regression problem.
All these models are constrained to respect exactly  the orientation $(\Gamma \hat x)_i$ of the biased solution on elements of the support, \ie when $i \in \hat \Ii$.
In \cite{deledalle2016clear,pierre2017luminance}, an alternative approach, based on the
preservation of covariant information between $\hat x$ and $y$,
aims only at preserving the orientation $(\Gamma \hat x)_i$ to some extent.
While also respecting the support of $\hat x$, this gives more flexibility for the refitted solution
to correct $\hat x$ and adapt to the data content $y$. This model is nevertheless insensitive to the
direction and it involves a quadratic penalty that tends to promote over smoothed refittings.

\subsection{Outline and contributions}%

In Section \ref{sec:refitting}, we present a general framework for refitting solutions promoted by $\ell_{12}$ sparse regularization \eqref{isotv}  that extends a preliminary version of this work \cite{deledalle2019refitting}.
Our variational refitting method relies on the use of block penalty  functions that act on the support of the biased solution $\hat x$.
We introduce the Soft-penalized Direction model (SD), while discussing suitable properties a refitting block penalty should satisfy.

In Section \ref{sec:practice}, we propose stable algorithms to compute our refitting strategy for any convex refitting block penalty.

We show in Section \ref{sec:related} how our model relates and inherits the advantages of other methods such as Bregman iterations \cite{Osher} or de-biasing approaches \cite{brinkmann2016bias,deledalle2016clear}.

Experiments in Section \ref{sec:results} exhibit  the practical benefits for the SD refitting for imaging problems involving TViso, a variant of TViso for color images and TGV based regularization.

With respect to the short paper \cite{deledalle2019refitting}, the contributions are as follows. We propose new  block penalties and a deep analysis of their properties. While providing technical information on the implementation of particular block penalties, we also detail  how the process can be generalized to arbitrary convex  functions.  The re-fitting scheme has been extended to another optimization algorithm, and a TGV based regularization model  is finally proposed (\ref{sec:tgv})  and experimented (Fig. \ref{fig:tgv}).

\section{Refitting with block penalties}\label{sec:refitting}
The refitting procedure of a biased solution $\hat x$ of \eqref{isotv} is expressed
in the following general framework
\begin{align}\label{general_refit}
  &
  \tilde{x}^{\phi} \in \!\! \uargmin{x; \; \supp(\Gamma x) \subseteq \hat{\Ii}} \!\!\!
  \frac{\norm{\Phi x - y}_2^2}{2}  +
  \sum_{i \in \hat{\Ii}} \phi((\Gamma x)_i, (\Gamma \hat x)_i)\enspace,
\end{align}
where $\phi : \RR^b \times \RR^b \to \RR$ is a block penalty
($b \geq 1$ is the size of the blocks)
promoting $\Gamma x$ to share information with $\Gamma \hat{x}$ in some sense to be specified.
To compute global optimum of the refitting model \eqref{general_refit},
we only consider in this paper
refitting block penalties such that $z \mapsto \phi(z, \hat{z})$
is convex.

To refer to some features of the vector $\Gamma \hat{x}$, let us first define properly the notions of
relative orientation, direction and projection between two vectors.
\begin{defn}\label{def:cos_and_proj}
  Let $z$ and $\hat z$ be vectors in $\RR^b$, we define
  \begin{align}
    \cos(z, \hat{z})
    =
    \dotp{\tfrac{z}{\norm{z}_2}}{\tfrac{\hat z}{\norm{\hat z}_2}}=\tfrac{1}{\norm{z}_2\norm{\hat z}_2}\sum_{j=1}^b z_j\hat z_j \enspace,
    \\
    \qandq P_{\hat z}(z)= \dotp{z}{\tfrac{\hat{z}}{\norm{\hat{z}}_2}}\tfrac{\hat{z}}{\norm{\hat{z}}_2}=\tfrac{\norm{z}_2}{\norm{\hat{z}}_2} \cos(z, \hat{z})\hat{z}
    \enspace,
  \end{align}
  where $P_{\hat z}(z)$ is the orthogonal projection of $z$ onto $\Span(\hat z)$ (\ie the orientation axis of $\hat z$).
  We  say that
  $z$ and $\hat z$ share the same orientation (resp.~direction), if $\abs{\cos(z, \hat{z})} = 1$ (resp.~$\cos(z, \hat{z}) = 1$).
  We also consider that $ \cos(z, \hat{z})=1$ in case of null vectors $z=0_b$ and/or $\hat z=0_b$.
\end{defn}
Thanks to Definition~\ref{def:cos_and_proj}, we can now introduce our refitting block penalty
designed to preserve the desired features of $\hat{z}=\Gamma \hat x$ in a simple way.
We call our block penalty the {\em Soft-penalized Direction} (SD) penalty which reads as
\begin{align}\label{pen:SD}
  \phi^{}_{\mathrm{SD}}(z, \hat{z}) =
  \lambda \norm{z}_2(1 - \cos (z, \hat{z}))\enspace.
\end{align}
We also introduce five other alternatives,
the Hard constrained Orientation (HO) penalty
\begin{align}
  \label{pen:HO}
  \phi^{}_{\mathrm{HO}}(z, \hat{z})
  &={\iota_{\{z \in \RR^b : \;\abs{\cos(z, \hat{z})}= 1\}}}(z) \enspace,
\end{align}
where $\iota_{\mathcal{C}}$ is the $0$/$+\infty$ indicator function of a set $\mathcal{C}$,
the Hard-constrained Direction (HD) penalty
\begin{align}
  \label{pen:HD}
  \phi_{\mathrm{HD}}(z, \hat{z})
  &={\iota_{\{z \in \RR^b : \; \cos(z, \hat{z}) = 1\}}}(z) \enspace,
\end{align}
the Quadratic penalized Orientation (QO) penalty
\begin{align}
  \label{pen:QO}
  \phi_{\mathrm{QO}}(z, \hat z)
  &
  =
  \tfrac{
    \lambda {\norm{z}^2_2}}{2\norm{\hat{z}}_2
  }{
    (1 - \cos^2(z, \hat{z}))
  }\enspace,
\end{align}
the Quadratic penalized  Direction (QD) penalty
\begin{align}
  \label{pen:QD}
  \phi_{\mathrm{QD}}(z, \hat z)
  & = \choice{ \frac{\lambda}{2}
    \frac{\norm{z}^2_2}{\norm{\hat{z}}_2}(1 -
    \cos^2 (z, \hat{z})) & \text{if } \cos(z, \hat{z}) \geq 0\\
    \frac{\lambda}{2}
    \frac{\norm{z}^2_2}{\norm{\hat{z}}_2}&\otherwise}
\end{align}
and the Soft-constrained Orientation (SO) penalty
\begin{align}
  \label{pen:SO}
  \phi_{\mathrm{SO}}(z, \hat z)
  &
  =
  \lambda\norm{z}_2 \sqrt{1 - \cos^2 (z, \hat z)}\enspace.
\end{align}

We will see in Section \ref{sec:related} that HD, HO and QO lead us to retrieve existing refitting models
known respectively in the litterature as
ICB (Infimal Convolution betweem Bregman distances) debiasing \cite{brinkmann2016bias},
Bregman debiaising \cite{brinkmann2016bias},
and CLEAR (Covariant LEAst square Refitting) \cite{deledalle2016clear}.

\subsection{Desired properties of refitting block penalties}
We now introduce properties a block penalty $\phi$ should satisfy for refitting purposes,
for any $\hat{z}$:
\begin{itemize}\setlength{\itemindent}{.1in}
\item [(P1)] $\phi$ is convex, non negative and $\phi(z, \hat{z}) = 0$, if $\cos(z, \hat z) = 1$ or $\norm{z}_2=0$,
\item [(P2)] $\phi(z', \hat{z}) \geq \phi(z'', \hat{z})$
  if $\norm{z'}_2 = \norm{z''}_2$ and $\cos (z'', \hat{z}) \geq \cos (z', \hat{z})$,
\item [(P3)] $z \mapsto \phi(z, \hat{z})$ is continuous,
\item [(P4)] $
  \phi(z, \hat{z}) \leq
  C \norm{z}_2, \text{ for } C > 0.
  $
\end{itemize}
Property (P1) stipulates that no configuration can be more favorable than $z$ and $\hat z$ having the same direction.
Hence, the direction of the refitted solution should be encouraged to follow the one of the biased solution.
Property (P2) imposes that for a fixed amplitude, the penalty should be increasing w.r.t. the angle formed with $\hat z$.
Property (P3) enforces refitting that can continuously adapt to the data and be robust to small perturbations.
Property (P4) claims that a refitting block penalty should not penalize more some configurations
than the original penalty $\norm{.}_{1,2}$, at least up to some multiplicative factor $C > 0$.


\begin{prop}\label{prop:block_properties}
Properties of block penalties lead to the following implications.
\begin{itemize}
\item[(a)] (P1)   $\Rightarrow \norm{ \Phi \tilde{x}^{\phi} - y }^2_2\leq \norm{\Phi \hat x - y}^2_2$.
\item[(b)] (P1)  $\Rightarrow\phi$ is non decreasing with respect to $\norm{z}_2$ for a fixed angle $(z,\hat z)$,
\item[(c)] (P2) $\Rightarrow$ $\phi$ is symmetric with respect to the orientation axis induced by $\hat z$,
\item[(d)] (P1) + (P2)  $\Rightarrow \phi(z', \hat{z}) \geq \phi(z, \hat{z})$,
  if $\norm{z'}_2 \geq  \norm{z}_2$ and $\cos (z, \hat{z}) =\cos (z', \hat{z})$,
\item[(e)] (P1) + (P3) $\Rightarrow\phi(z,\hat z)\to 0$ when  $\cos (z, \hat{z})\to 1$.
\end{itemize}
\end{prop}

\begin{proof}

\noindent
  $\bullet$ (a) As $\tilde x^\phi$ is solution of \eqref{general_refit}, the relation  is obtained by observing that  $\phi(\hat z,\hat z)=0$ with (P1).
  \medskip

  \noindent
  $\bullet$ (b) Looking at a ray $[0,z)$ for any vector $z$, this is a consequence of the convexity of $\phi$ and the fact that $\phi(z,\hat z)=0$ for $\norm{z}_2=0$.
  \medskip

  \noindent
  $\bullet$ (c) Since $\cos(z,\hat z)=\cos(-z,\hat z)$, one can combine  $(z',z'')=(z,-z)$ and  $(z',z'')=(-z,z)$  in (P2) to obtain (c).
  \medskip

  \noindent
  $\bullet$ (d) Direct consequence of points (b) and (c).
  \medskip

  \noindent
  $\bullet$ (e) This is due to the  continuity of the function $\phi(.,\hat z)$ that is $0$ on the ray $[0,\hat z)$ from (P1).
\end{proof}
\subsection{Properties of considered block penalties}

The properties of the previously introduced refitting block penalties are synthesized in Table \ref{tab:properties}.
The proposed SD model is the only one satisfying  all the desired properties.
Figure \ref{fig:BP} gives an illustration of the evolution of the different
refitting block penalties as a function of its arguments (first column). Il also exhibits the influence of the angle between $z$ and $\hat z$ (second column) and the norm of $z$ (third column) on the penalization value. From this figure,
SD is shown to be a continuous penalization
that increases continuously with respect to the absolute angle between $z$ and $\hat z$.
In addition of satisfying the desired properties, one can also observe from this figure that
unlike some other alternatives, SD SD has a unique minimizer as a function of the angle and the amplitude
of $z$.

\begin{table}[!t]
  \centering
  \caption{\label{tab:properties}Properties satisfied by the considered block penalties.}%
  \begin{tabular}{c@{\hspace{.5cm}}c@{\hspace{.5cm}}c@{\hspace{.5cm}}c@{\hspace{.5cm}}c@{\hspace{.5cm}}c@{\hspace{.5cm}}cc}
    \hline
    Properties & HO & HD & QO& QD&SO  &SD
    \\
    \hline
    \hline
    1 & $\surd$ & $\surd$ & $\surd$  &  $\surd$  & $\surd$ & $\surd$ 
    \\
    2 &  &  $\surd$ & & $\surd$ && $\surd$ 
    \\
    3 &        &          & $\surd$ & $\surd$ & $\surd$ & $\surd$
    \\
    4 &  &         &        && $\surd$ &      $\surd$
    \\
    \hline
    a & $\surd$ & $\surd$ & $\surd$  &  $\surd$  & $\surd$ & $\surd$ 
        \\
    b & $\surd$ & $\surd$ & $\surd$  &  $\surd$  & $\surd$ & $\surd$ 
    \\
    c &  &  $\surd$ & & $\surd$ && $\surd$ 
    \\
    d &  &  $\surd$ & & $\surd$ && $\surd$ 
    \\
    e &        &          & $\surd$ & $\surd$ & $\surd$ & $\surd$
    \\
    \hline
  \end{tabular}
\end{table}

Other block penalties are insensitive to directions (HO, QO and SO), completely intolerant (HD and HO) or too tolerant (QD) to small changes of orientations,
hence not satisfying. These drawbacks will be illustrated in our experiments conducted in Section \ref{sec:results}.

When $b=1$, the orientation-based penalties (QO, HO and SO) have absolutely no effect while the direction-based penalties HD and QD preserve the sign of $(\Gamma \hat x)_i$.
In this paper, when $b \geq 1$, we argue that the direction of the block
$(\Gamma \hat x)_i$ carries important information that is worth preserving when refitting,
at least to some extent.

\newcommand{\scaleX}{0.19\linewidth}
\newcommand{\scaleXB}{0.19\linewidth}
\newcommand{\scaleY}{0.205\linewidth}
\newcommand{\scaleZ}{0.21\linewidth}
\newcommand{\sidecapY}[1]{{\begin{sideways}\parbox{\scaleY}{\centering #1}\end{sideways}}}
\begin{figure*}[!t]
  \begin{center}%
    \hfill%
\sidecapY{HO}\hfill\hspace{0.00\linewidth}%
{\includegraphics[width=\scaleX]{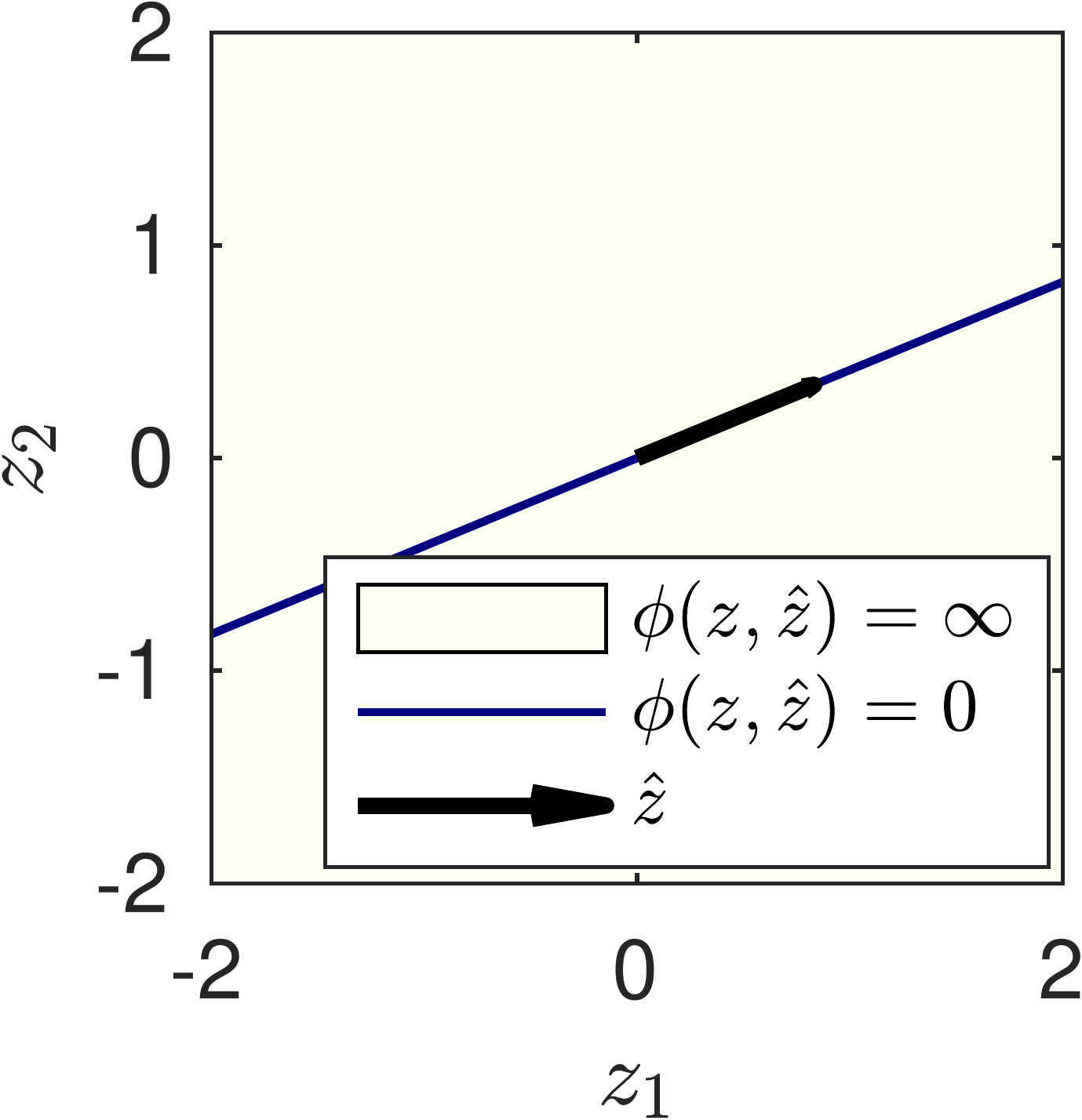}}\hfill%
{\includegraphics[width=\scaleY]{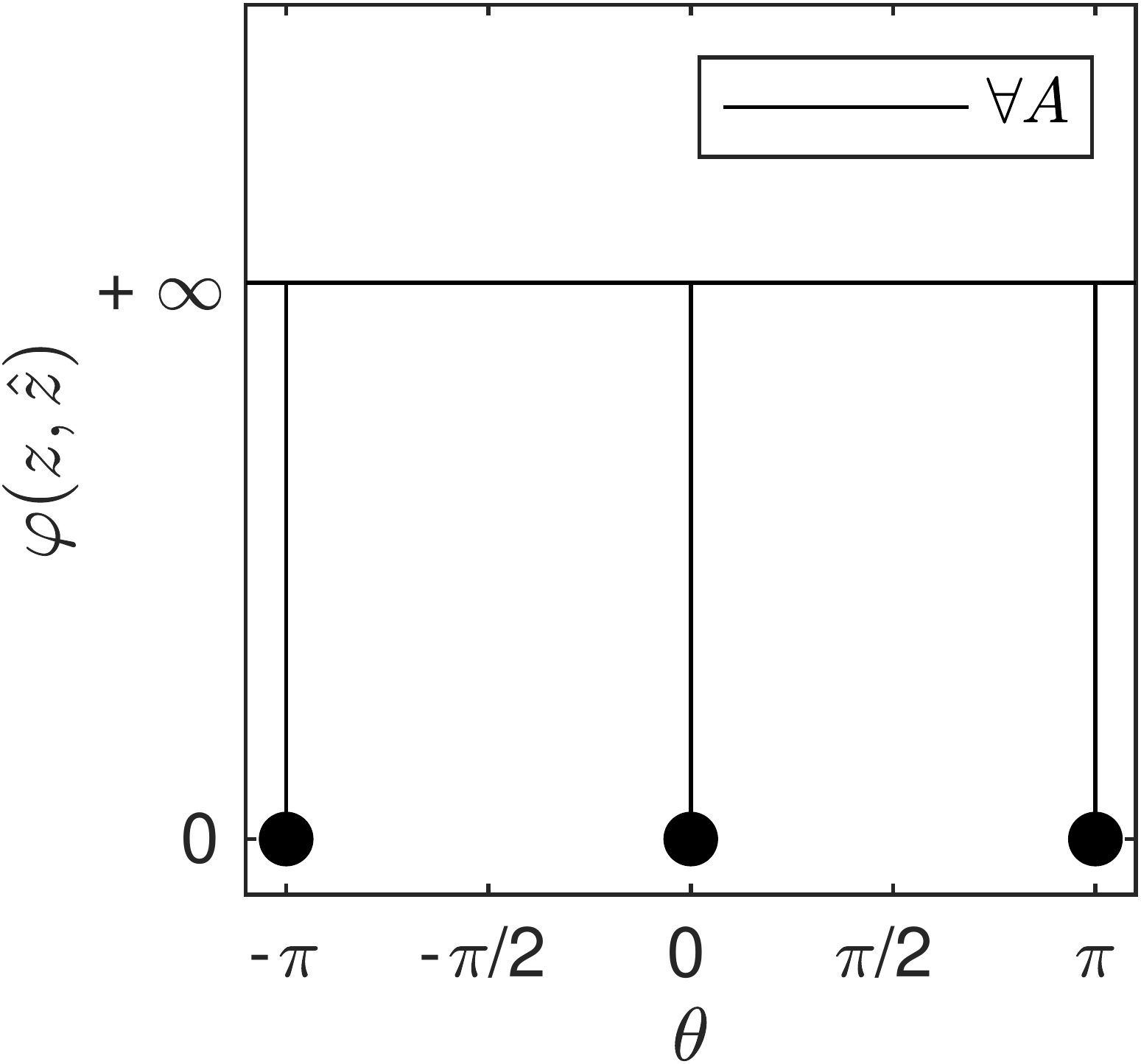}}\hfill%
{\includegraphics[width=\scaleZ]{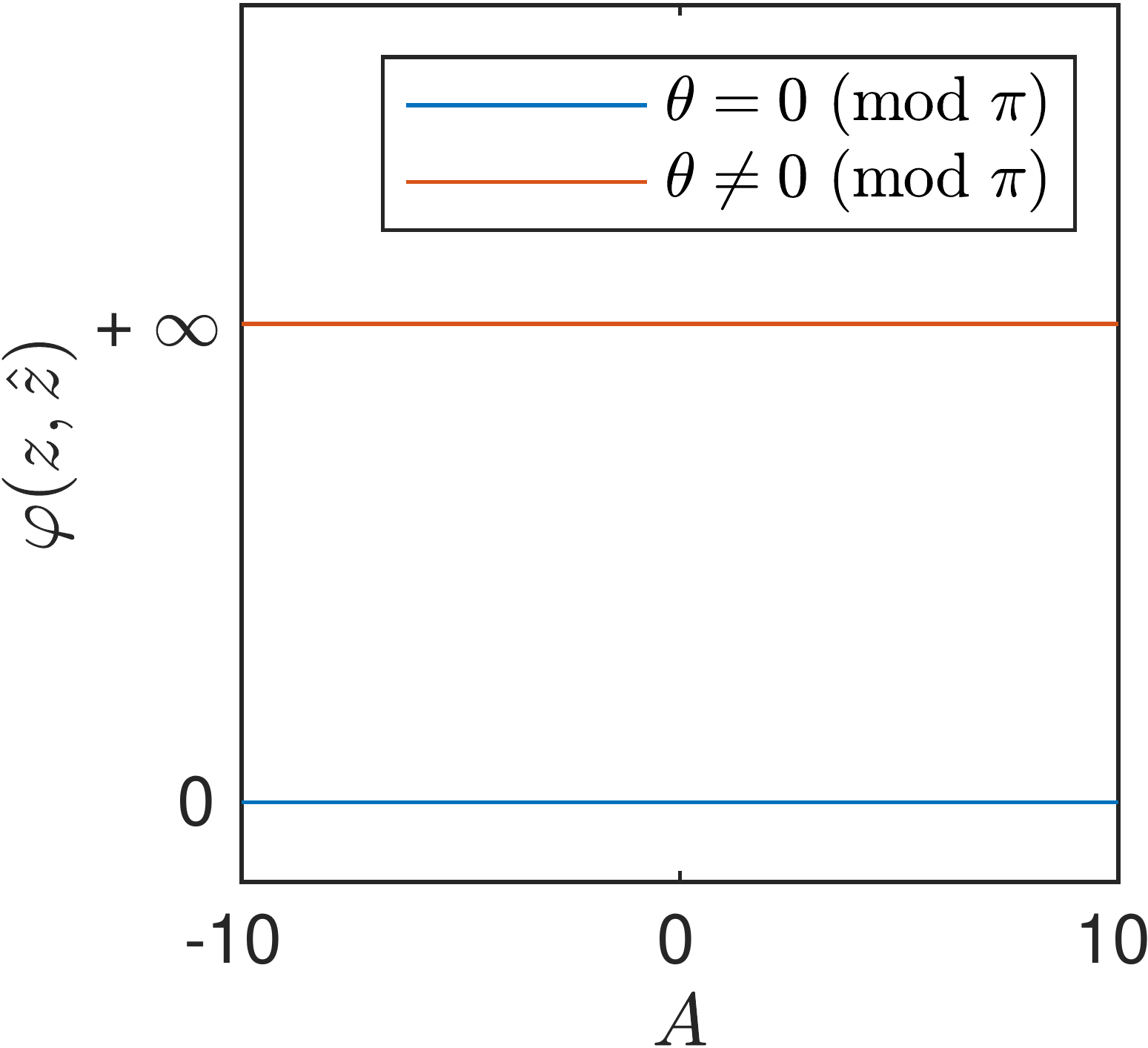}}\hfill%
    { }
    \\[-0.3em]%
    \hfill%
\sidecapY{HD}\hfill\hspace{0.00\linewidth}%
{\includegraphics[width=\scaleX]{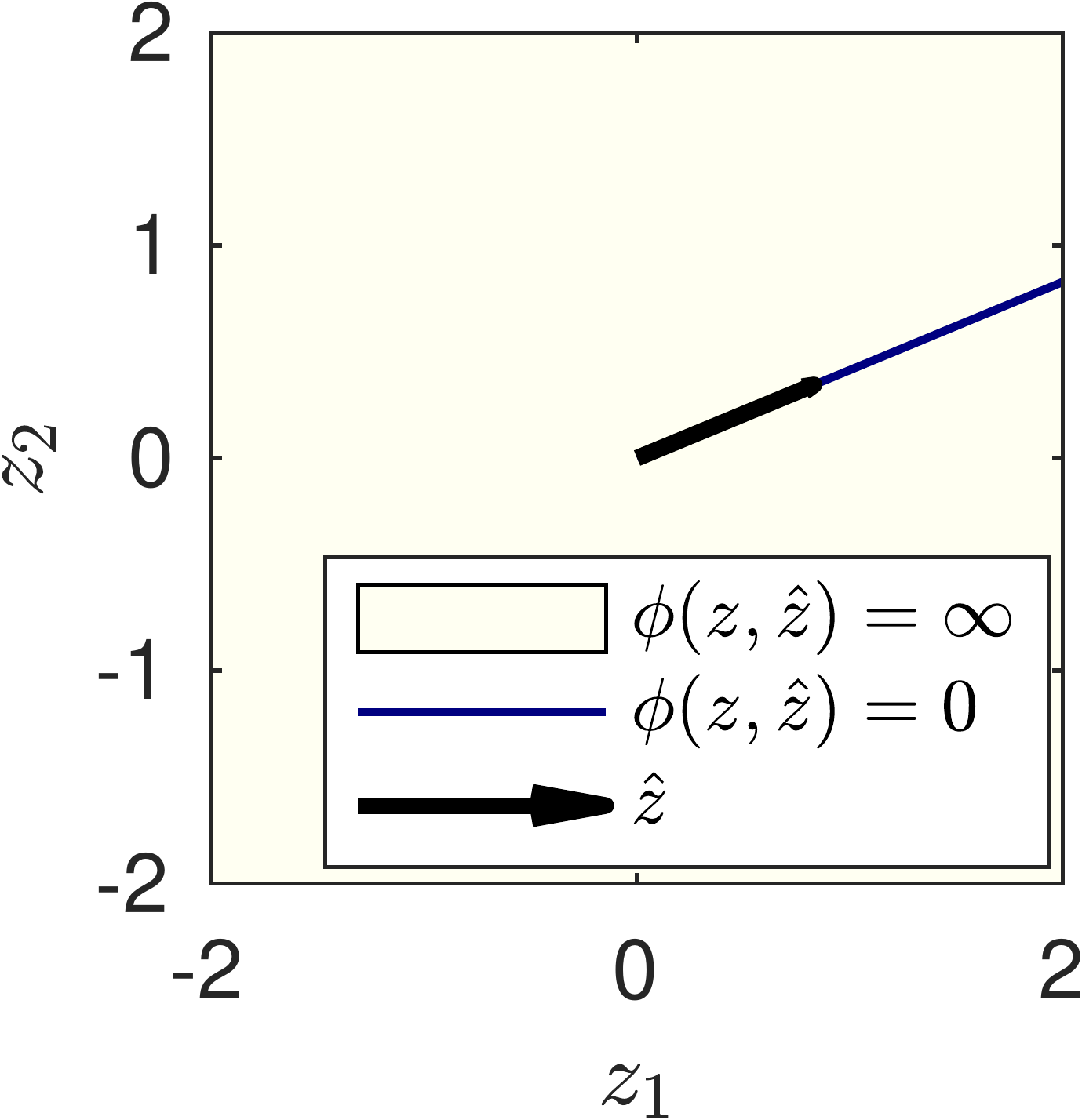}}\hfill%
{\includegraphics[width=\scaleY]{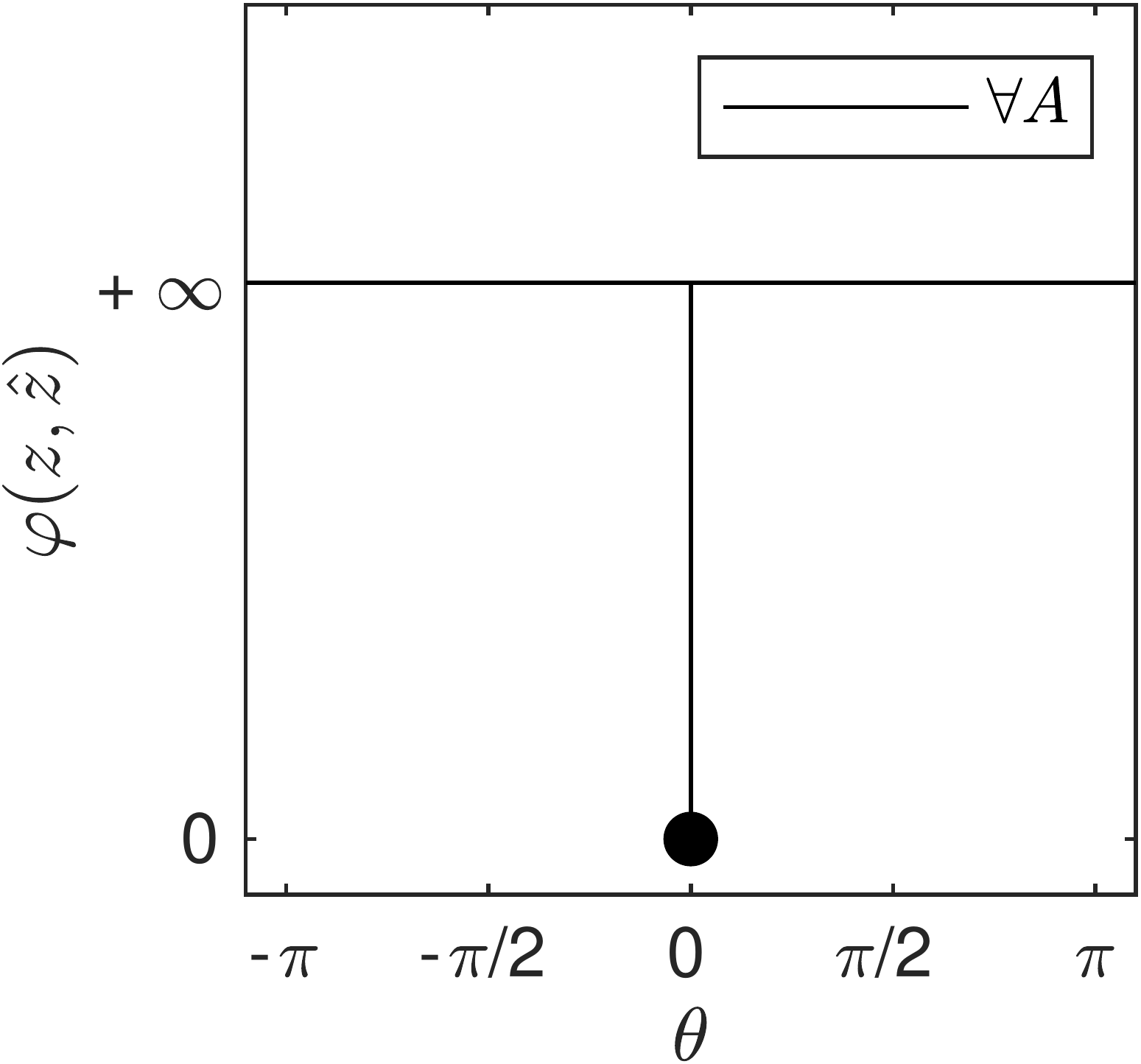}}\hfill%
{\includegraphics[width=\scaleZ]{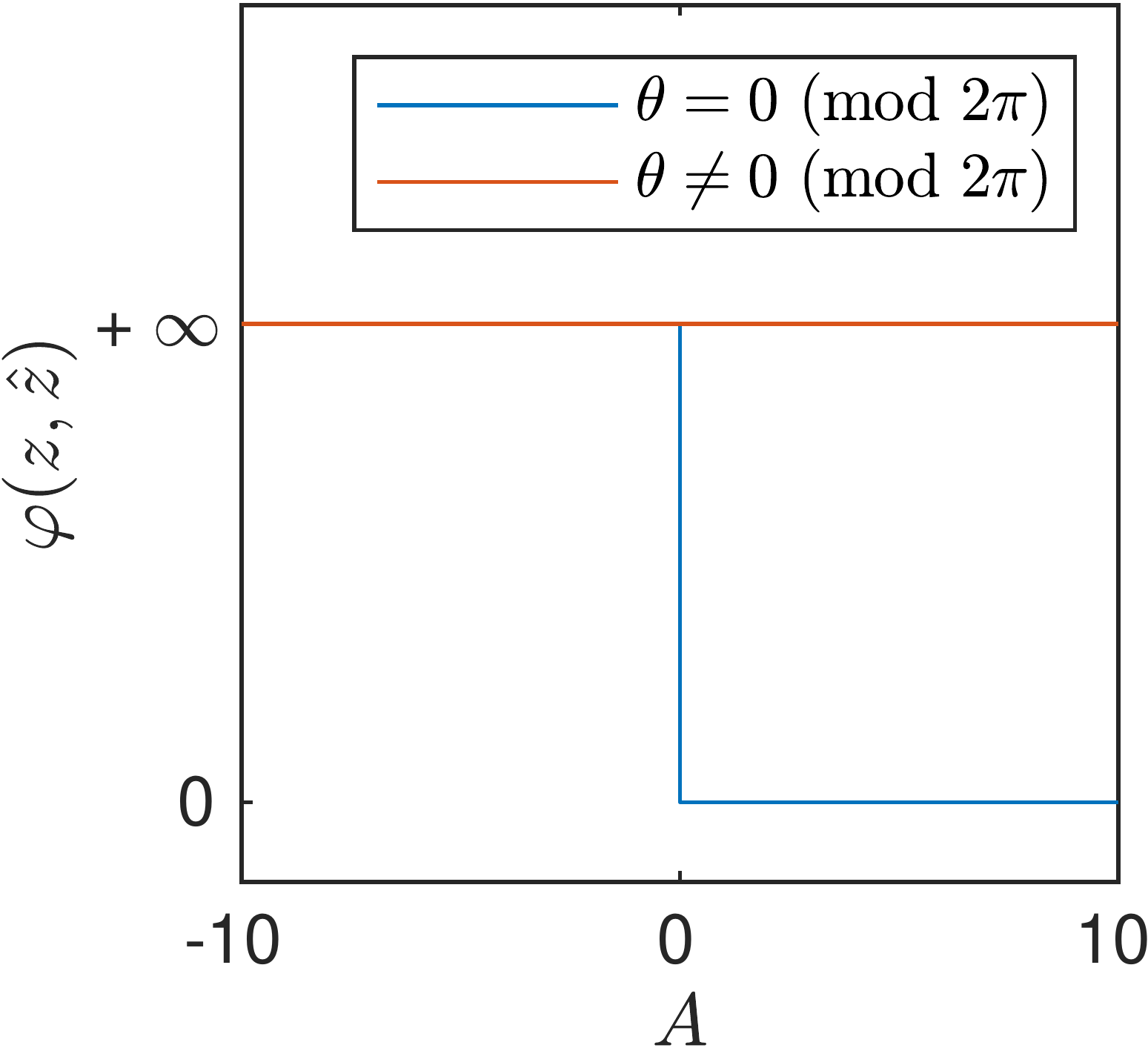}}\hfill%
    { }
    \\[-0.3em]%
    \hfill%
\sidecapY{QO}\hfill\hspace{0.00\linewidth}%
{\includegraphics[width=\scaleX]{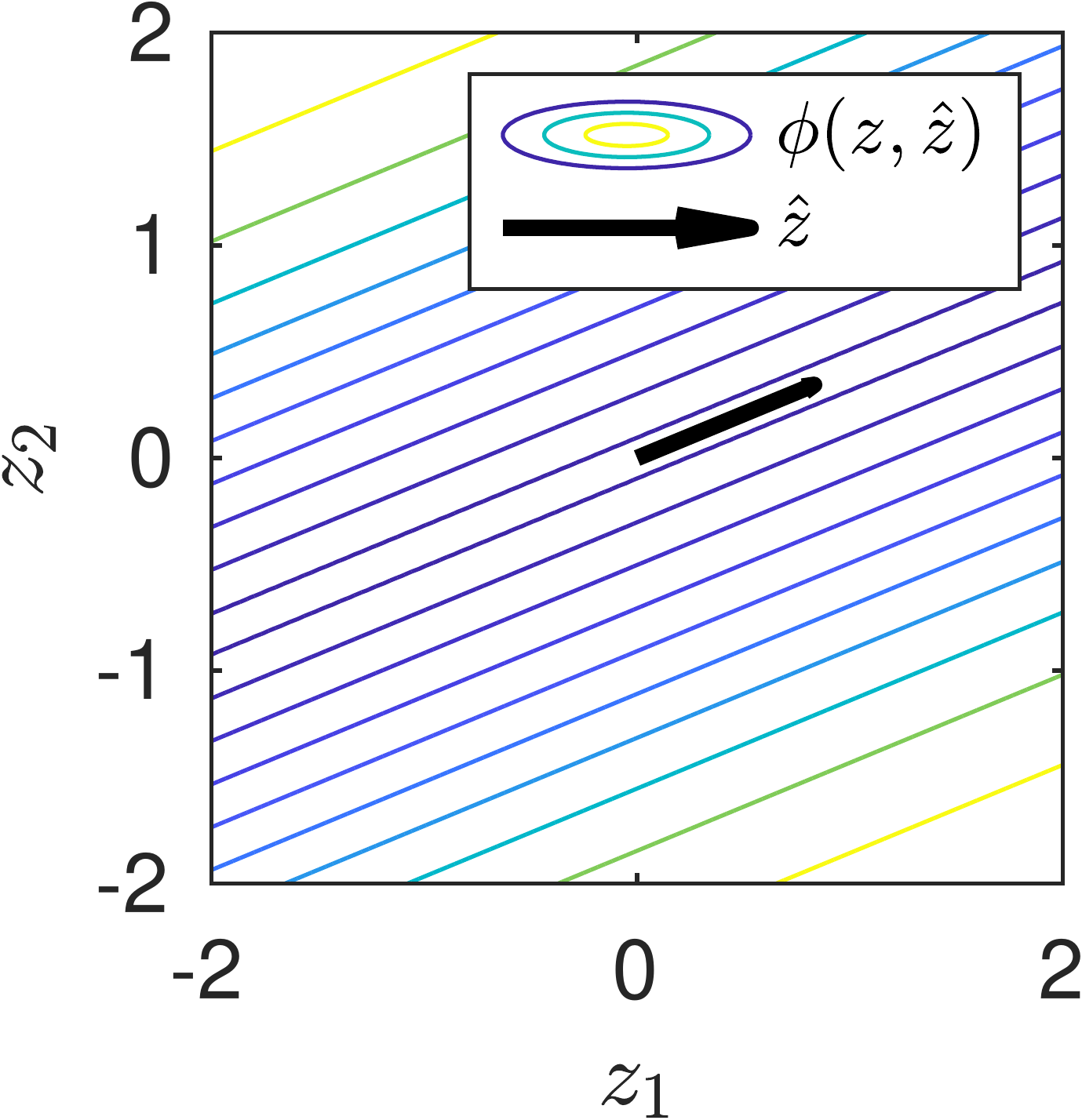}}\hfill%
{\includegraphics[width=\scaleY]{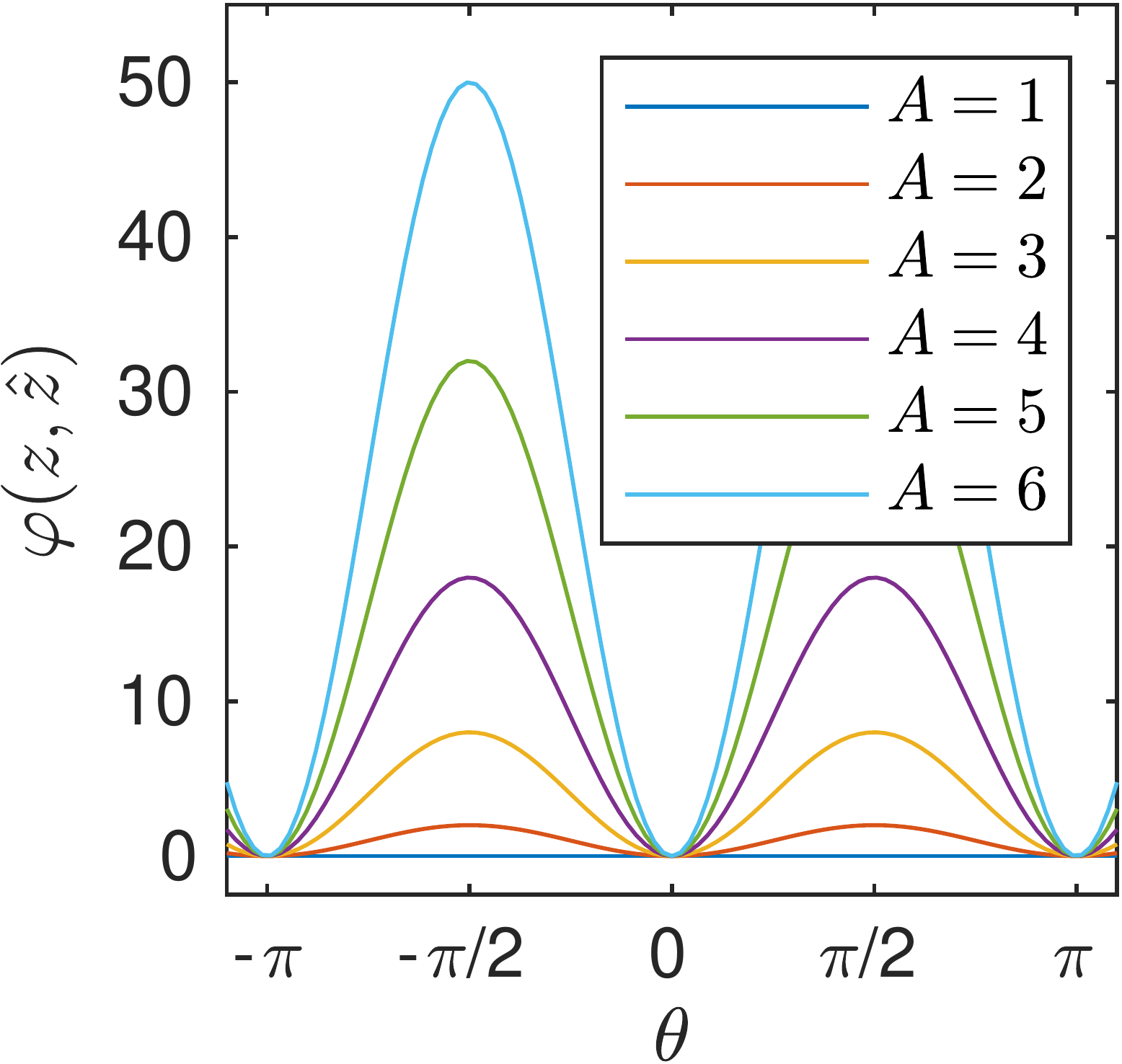}}\hfill%
{\includegraphics[width=\scaleZ]{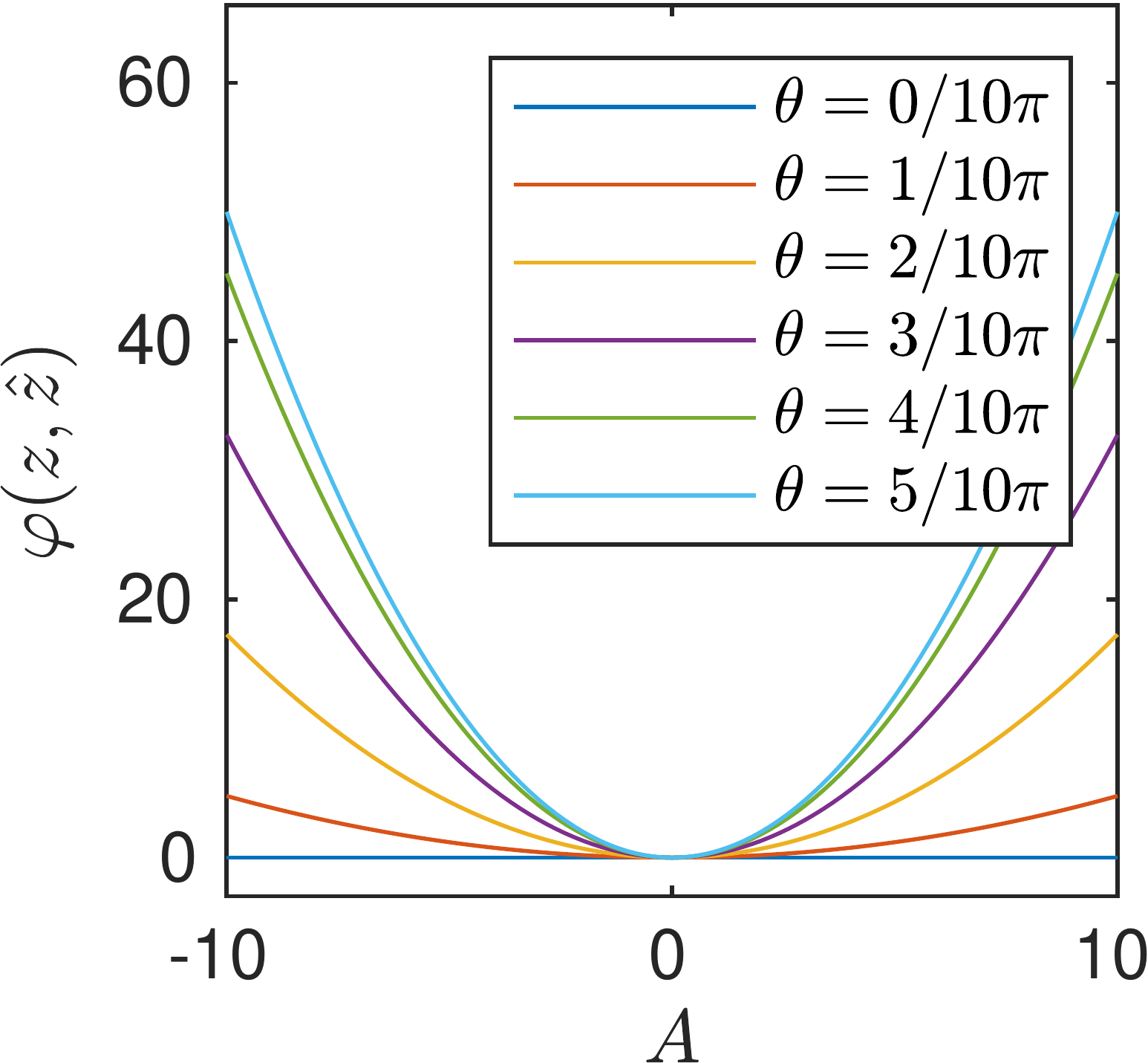}}\hfill%
    { }
    \\[-0.3em]%
        \hfill%
\sidecapY{QD}\hfill\hspace{0.00\linewidth}%
{\includegraphics[width=\scaleX]{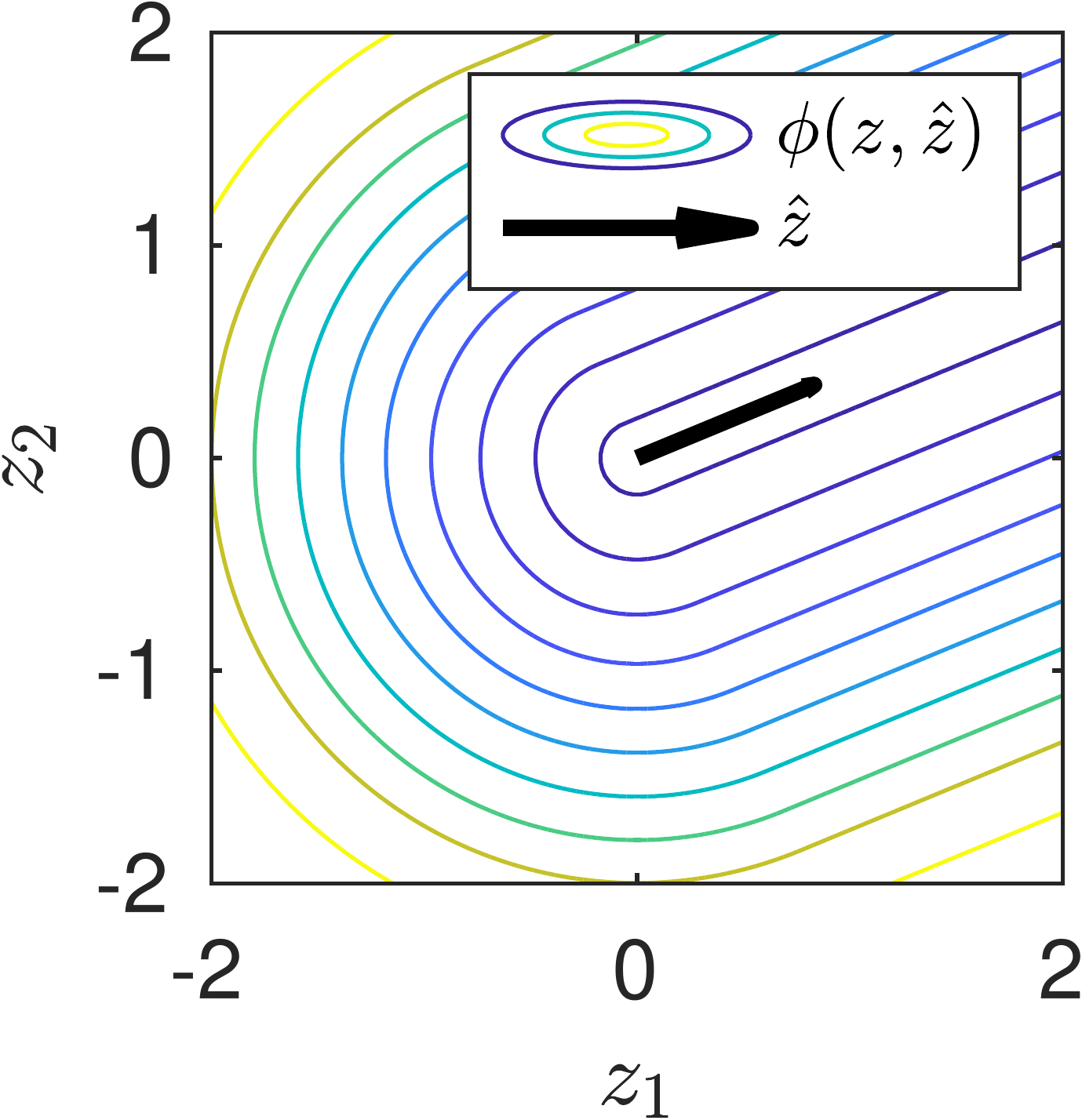}}\hfill%
{\includegraphics[width=\scaleY]{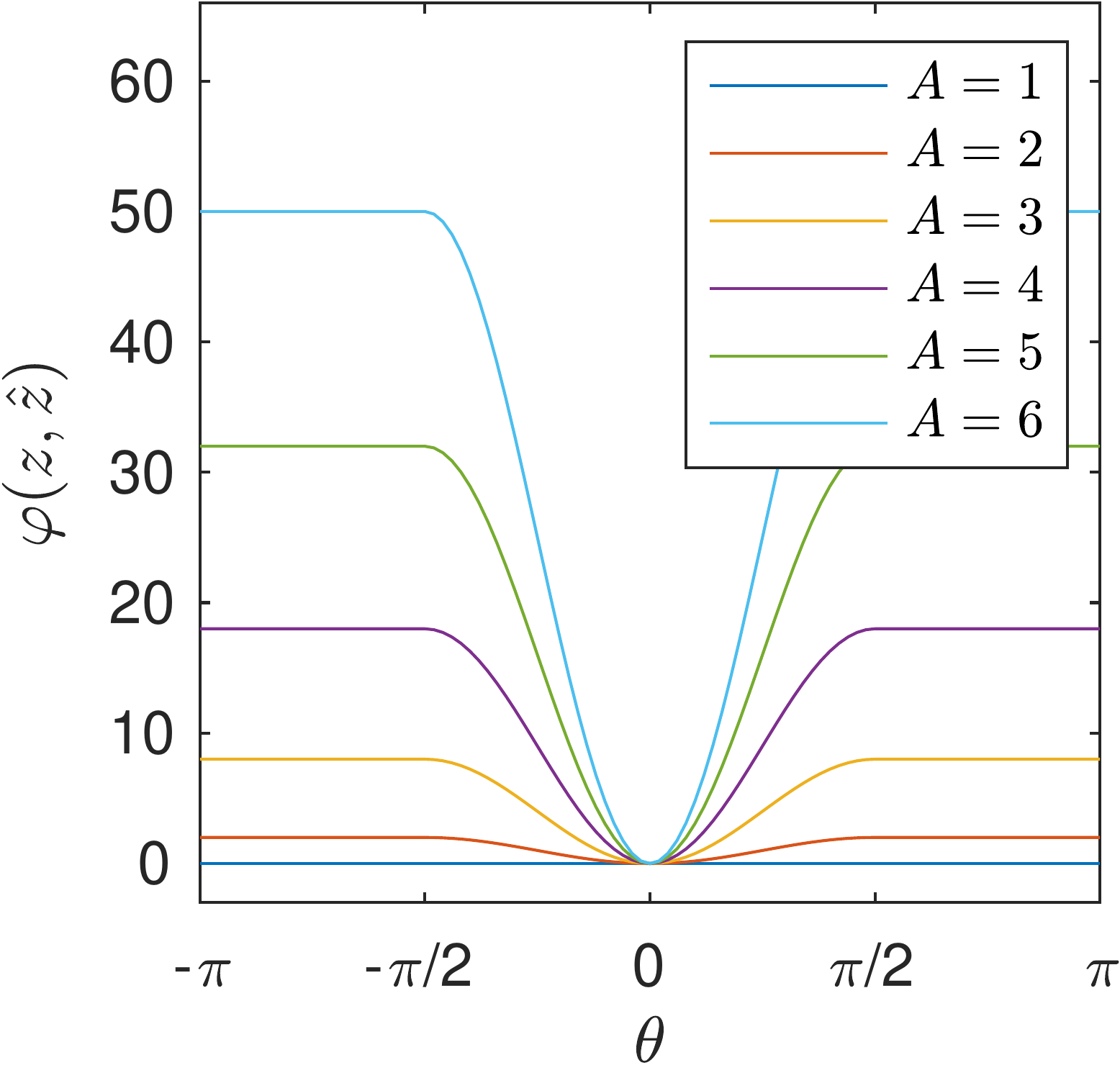}}\hfill%
{\includegraphics[width=\scaleZ]{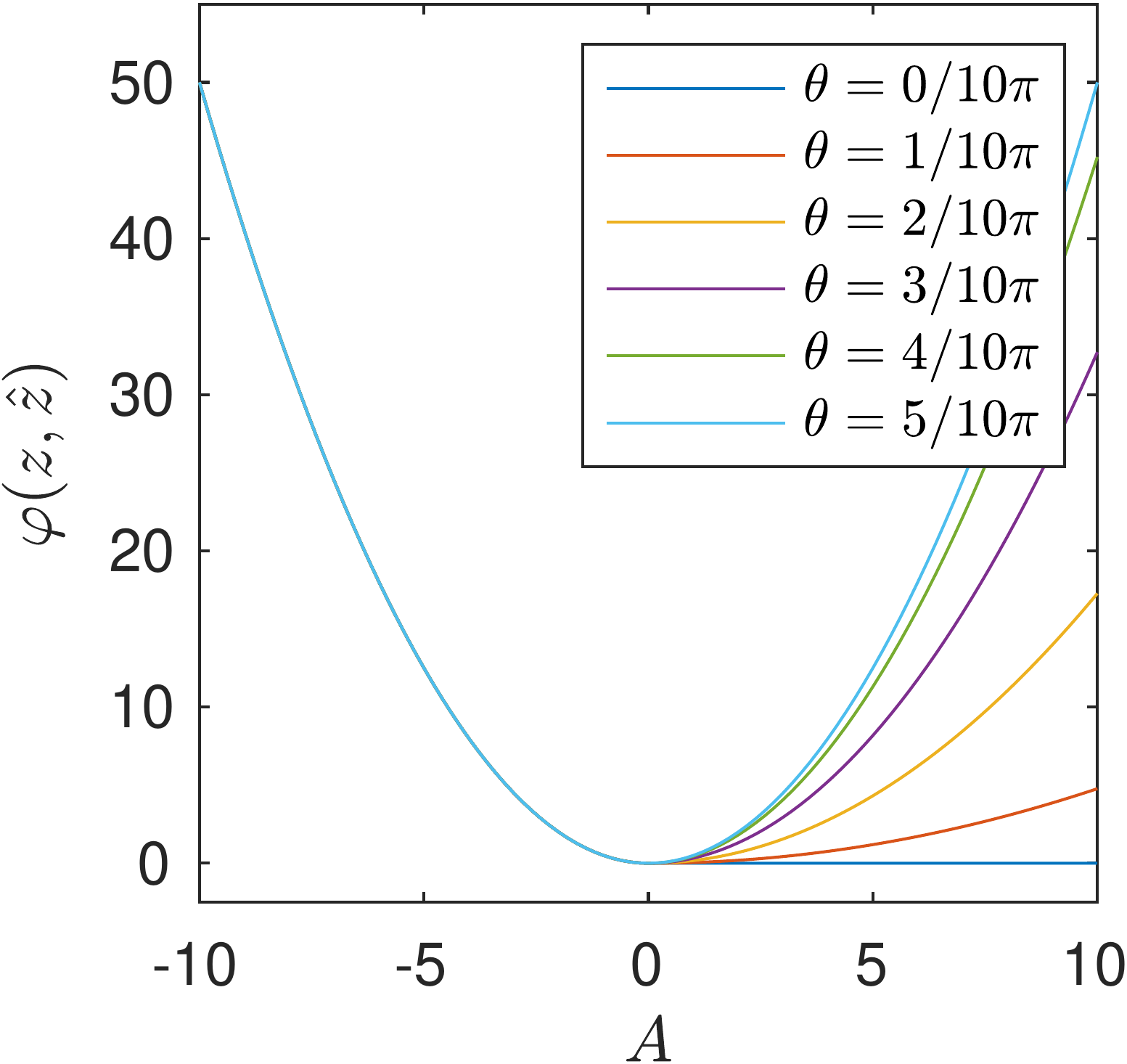}}\hfill%
    { }
    \\[-0.3em]%
    \hfill%
    \sidecapY{SO}\hfill\hspace{0.00\linewidth}%
{\includegraphics[width=\scaleX]{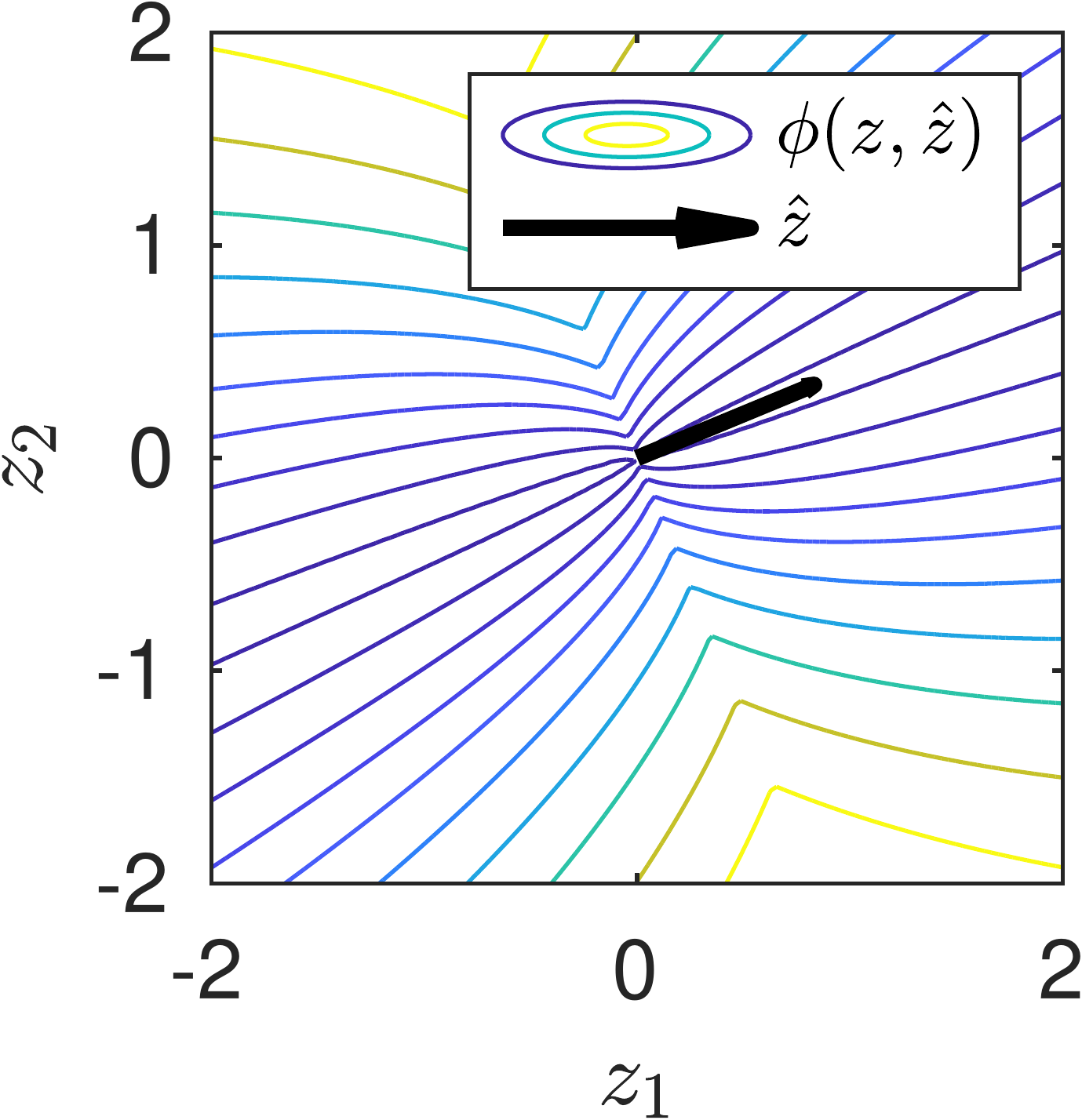}}\hfill%
{\includegraphics[width=\scaleY]{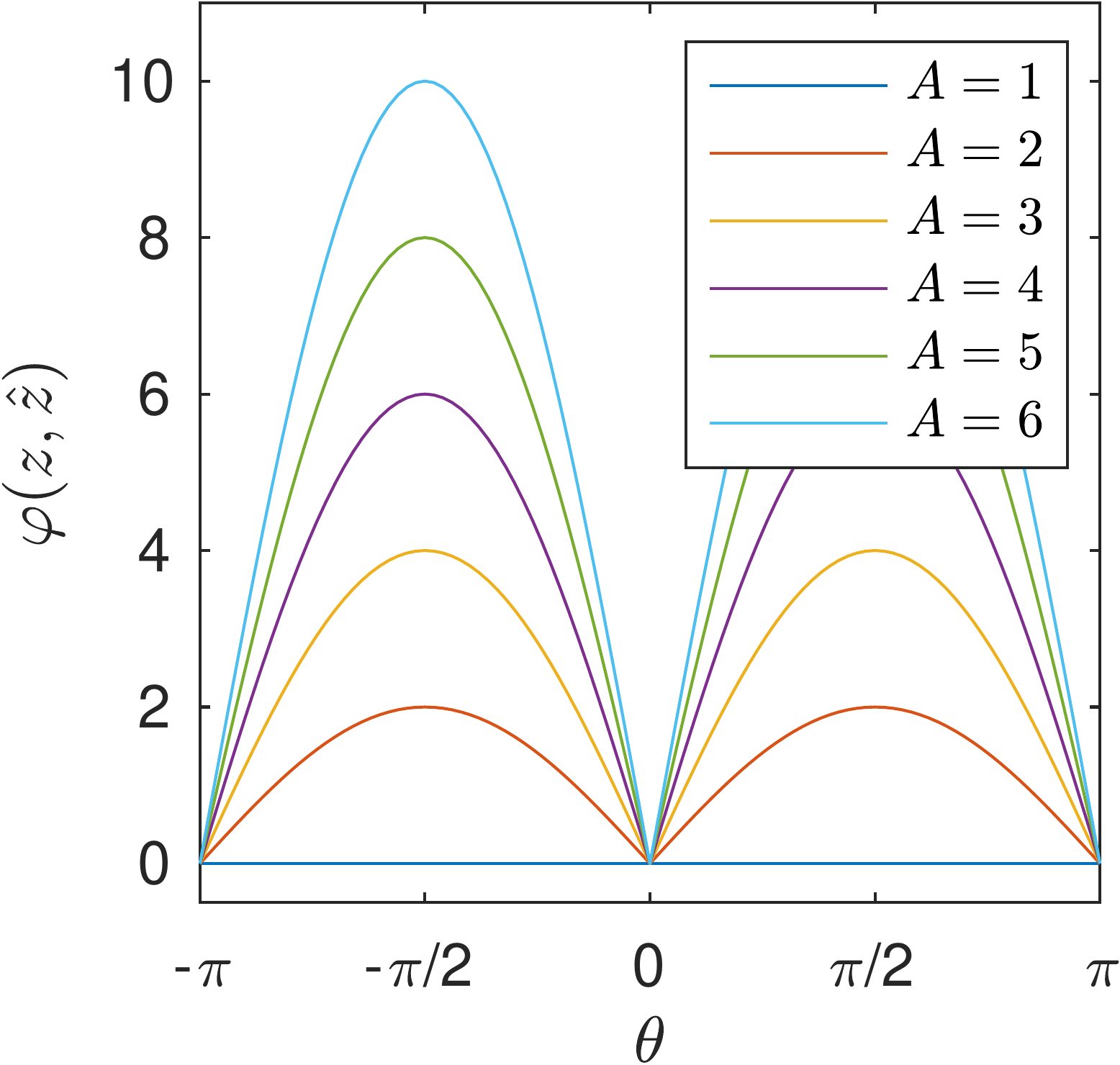}}\hfill%
{\includegraphics[width=\scaleZ]{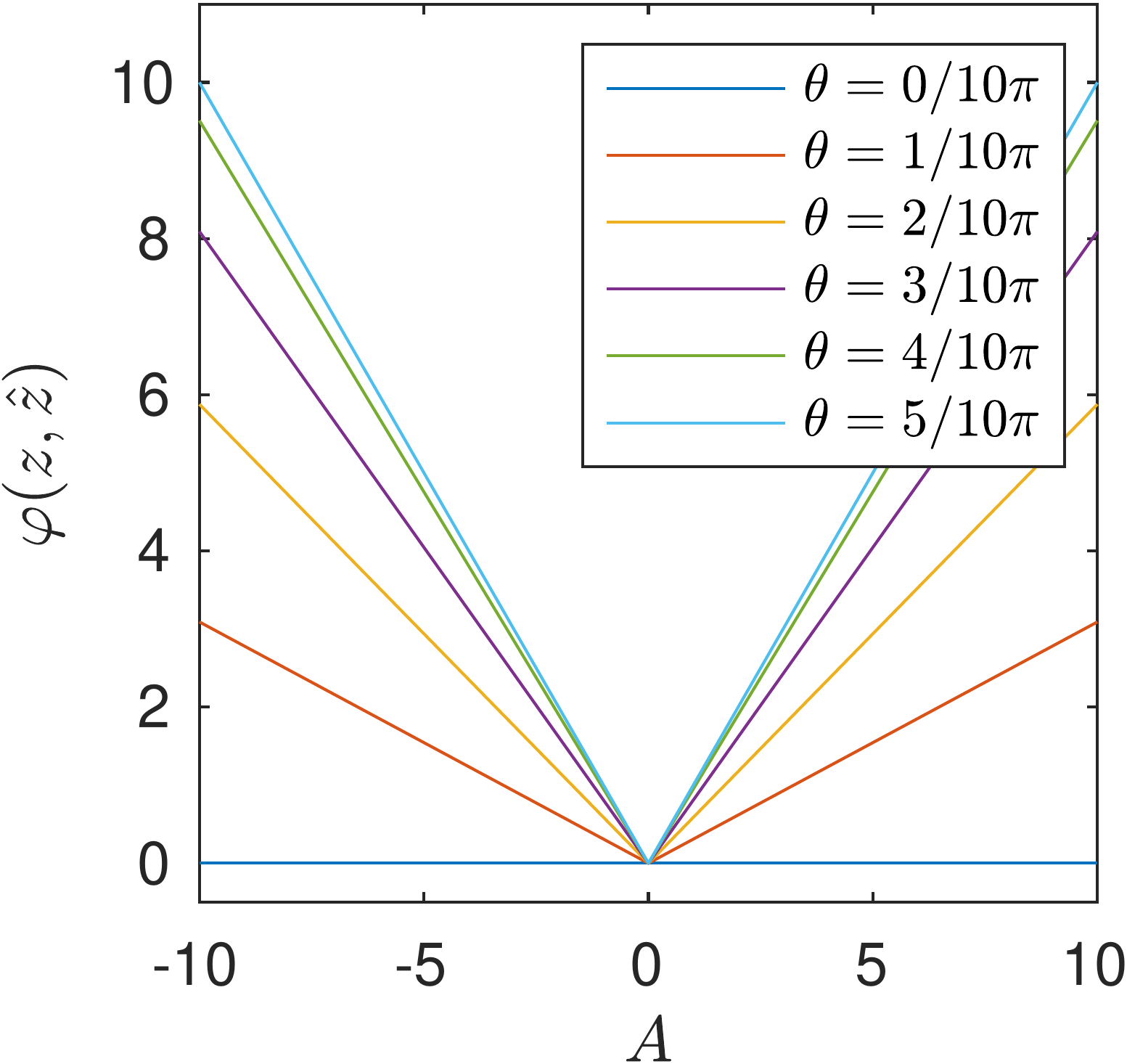}}\hfill%
    { }
    \\[-0.3em]%
    \hfill%
\sidecapY{SD}\hfill\hspace{0.005\linewidth}%
{\includegraphics[width=\scaleX]{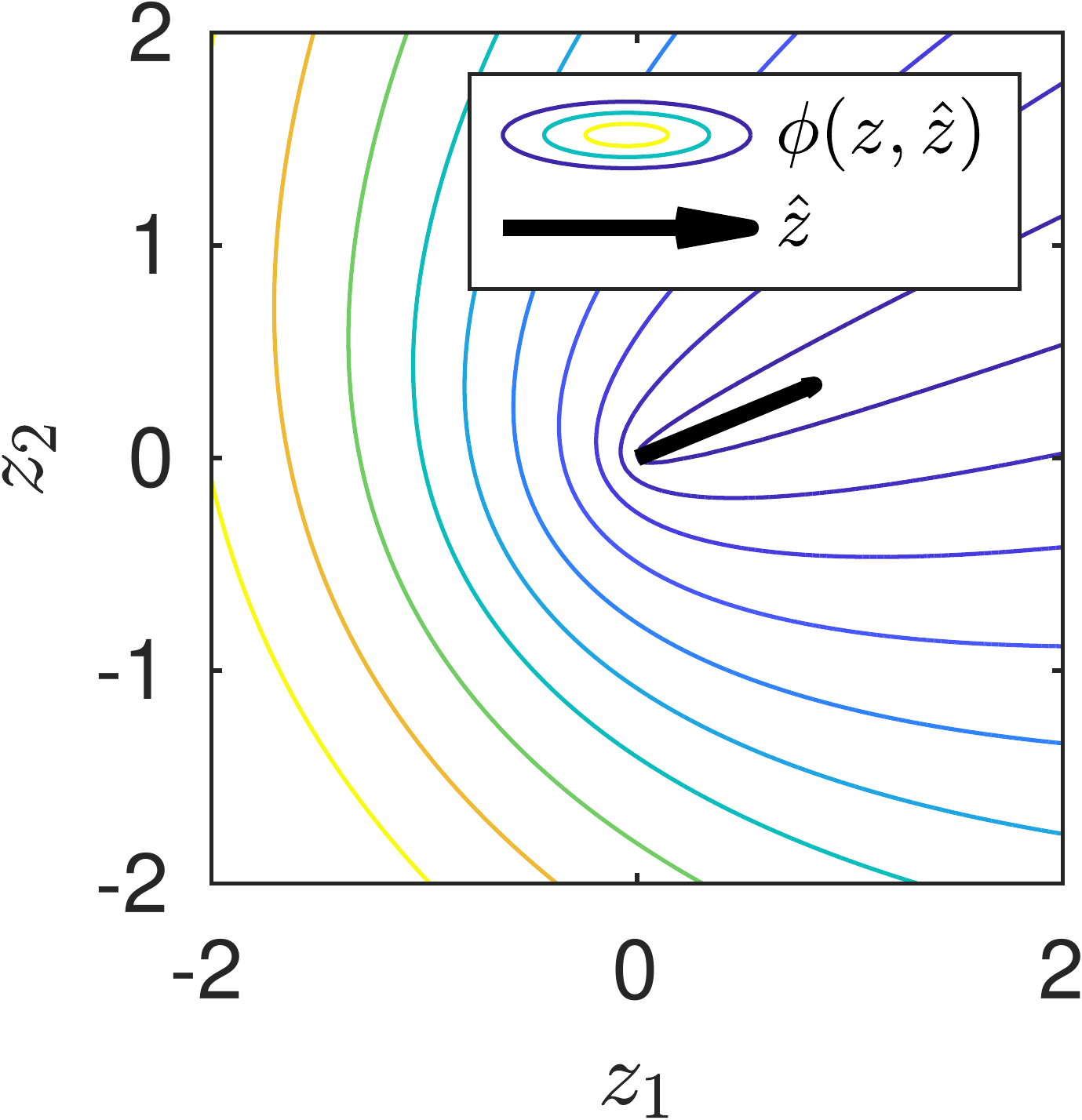}}\hfill%
{\includegraphics[width=\scaleY]{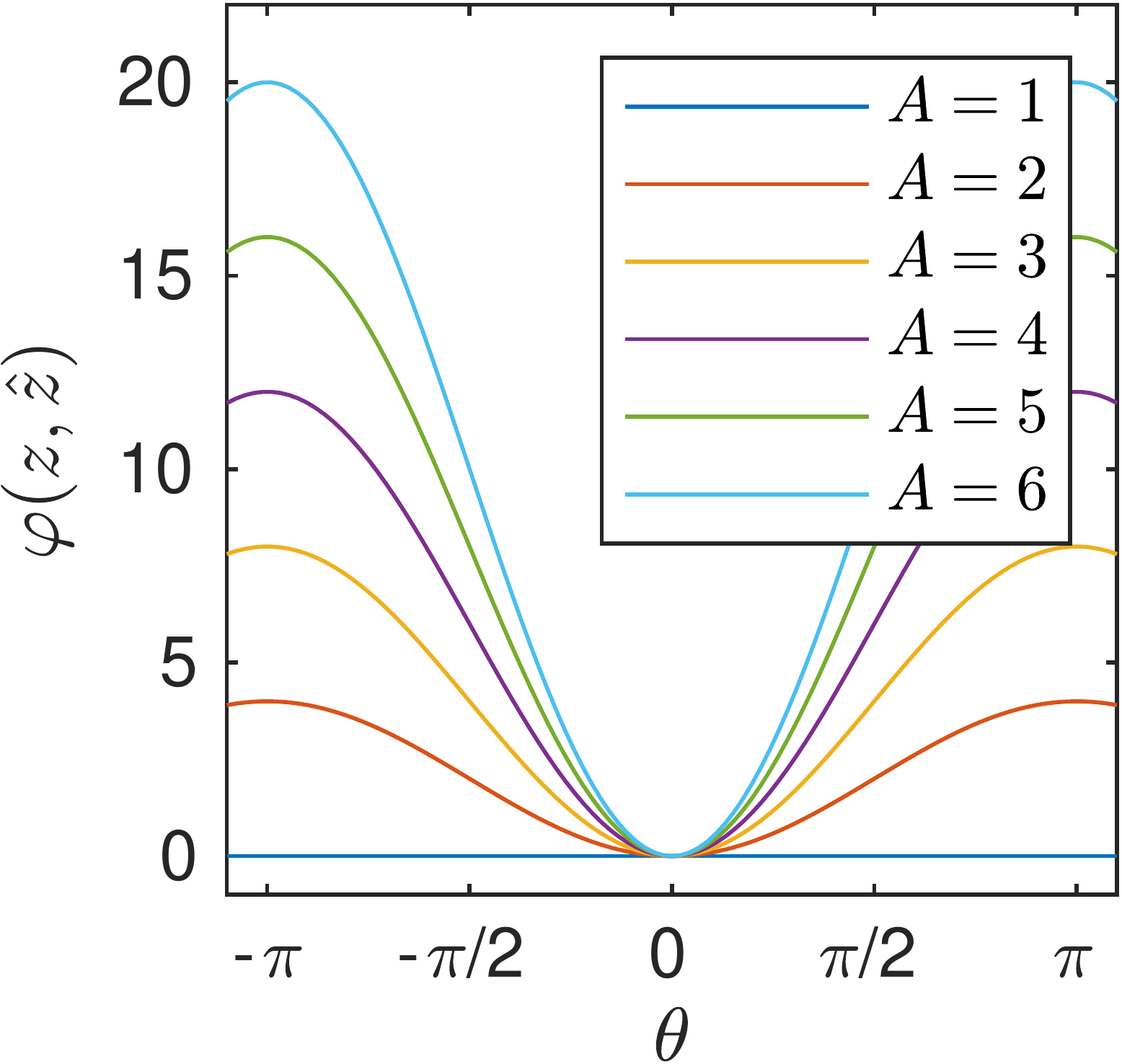}}\hfill%
{\includegraphics[width=\scaleZ]{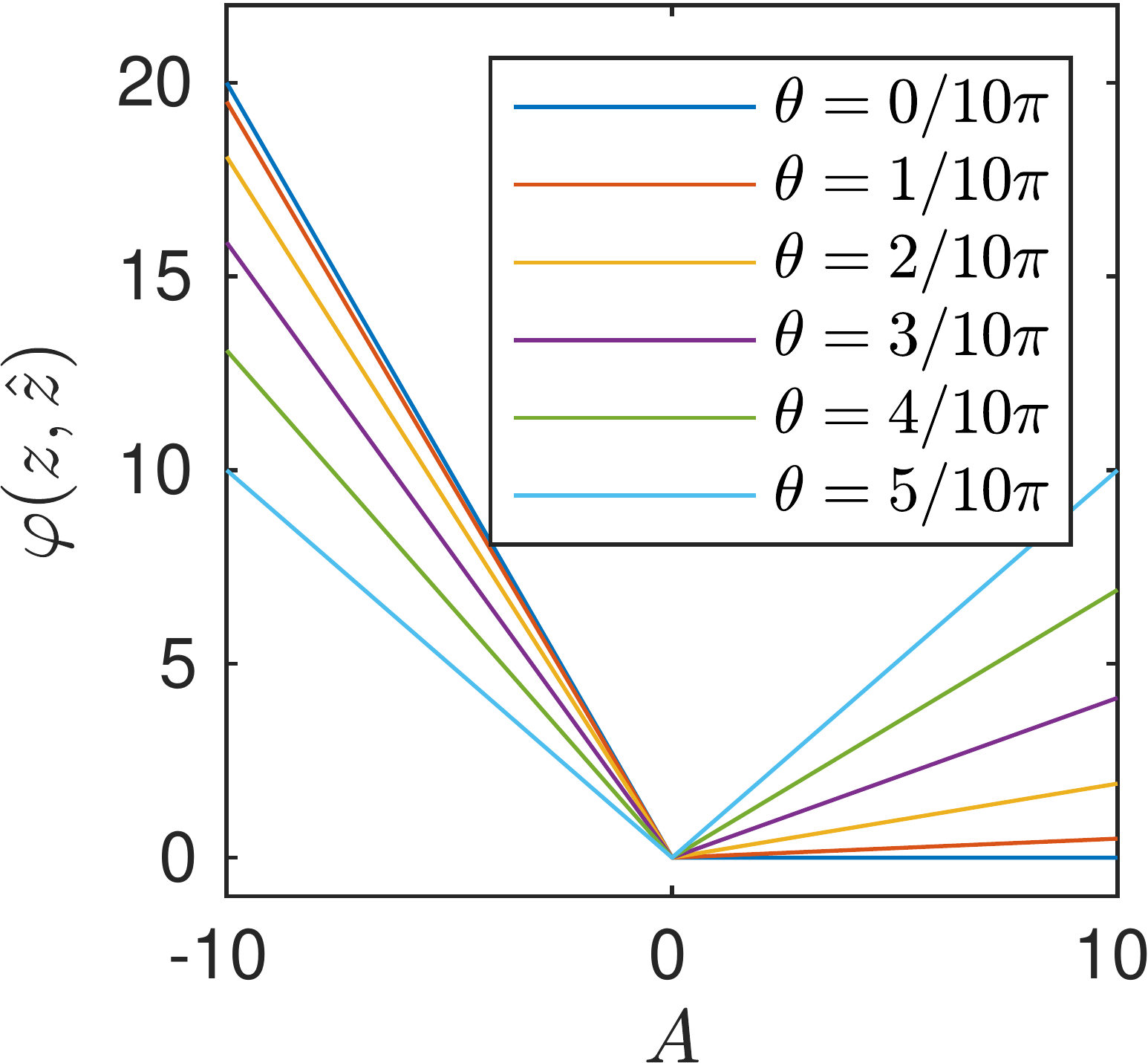}}\hfill%
    { }
    \\[-2em]
  \end{center}

  \newsavebox{\smlmat}
  \savebox{\smlmat}{$\left(\begin{smallmatrix} \cos \theta & -\sin \theta \\ \sin \theta & \cos \theta \end{smallmatrix}\right)$}
  \caption{\label{fig:BP} Illustration of block penalties:
    (left) 2D level lines of  $\phi$ for
    $z=(z_1,z_2) = A $\usebox{\smlmat}$\hat{z}$,
    (middle) evolution regarding the angle $\theta$ between $z$ and $\hat z$ and (right) evolution with respect to the modulus $A$ of $z$.
    HO and HD penalties have discontinuities while the ones on the four last rows are continuous and vary according to $A$. The penalties acting on directions (HD, QD and SD) are increasing w.r.t the amplitude of the angle $\theta$.
{In figure \ref{fig:BP}: $\varphi$ and $\phi$ mismatch in legends}
  }
\end{figure*}

\section{Refitting in practice}\label{sec:practice}

We now  introduce a general algorithm aiming to jointly solve the original problem \eqref{isotv} and the refitting one \eqref{general_refit} for any refitting block penalty $\phi$.
This framework has been extended from the stable projection onto the support developped in \cite{Deledalle_Papadakis_Salmon15} and
later adapted to refitting with the Quadratic Orientation penalty in \cite{deledalle2016clear}.

Given  $\hat x$ solution of \eqref{isotv},
a posterior refitting can be obtained by solving \eqref{general_refit} for  any refitting block penalty $\phi$.
To that end, we write the characteristic function of  support preservation as
$\sum_{i\in\hat{\Ii}^c} \iota_{\{0\}}(\xi_i),$
where $\iota_{\{0\}}(z)=0$ if  $z=0$ and $+\infty$ otherwise. By introducing the convex function
\begin{equation}\label{def:omega}
\omega_{\phi}(\xi,\Gamma \hat x,\hat{\Ii} )= \sum_{i\in\hat{\Ii}^c} \iota_{\{0\}}(\xi_i)+\sum_{i \in \hat{\Ii}} \phi(\xi_i, (\Gamma \hat x)_i)\enspace,
\end{equation}
the general refitting problem \eqref{general_refit} can be expressed as
\begin{equation}\label{general_refit_reform}
  \tilde{x}^{\phi} \in\uargmin{x \in \RR^n}
  \tfrac12 \norm{\Phi x - y}_2^2 + \omega_{\phi} (\Gamma x,\Gamma \hat x,\hat{\Ii})\enspace.
\end{equation}
We  now describe two iterative algorithms that can be used for the joint computation of $\hat x$ and $\tilde x$.

\subsection{Primal-dual formulation}
We  first  consider the primal dual formulation of the problems \eqref{isotv} and \eqref{general_refit_reform} that reads
\begin{align}\label{isotv:pd}
\min_{x \in \RR^n}\max_{\xi \in \RR^b}\;& \tfrac12 \norm{\Phi x - y }_2^2 + \langle \Gamma x,\xi\rangle - \iota_{B_2^\lambda}(\xi)\enspace,\\
\label{general_refit_reform:pd}
  \min_{x \in \RR^n}\max_{\xi \in \RR^b}\;&
  \tfrac12 \norm{\Phi x - y}_2^2 +\langle \Gamma x,\xi\rangle -\omega^*_{\phi} (\xi,\Gamma \hat x,\hat{\Ii}) \enspace,
\end{align}
where $\iota_{B_2^\lambda}$ is the indicator function of the $\ell_2$ ball of radius $\lambda$ (that is $0$ if $\norm{z_i}_2\leq \lambda$ for all $i\in[m]$ and $+\infty$ otherwise)
and
\begin{align}
  \omega_\phi^*(\xi,\Gamma \hat x,\hat{\Ii})=\sup_{\zeta \in \RR^b}\; \langle \xi,\zeta\rangle-\omega_{\phi}(\zeta,\Gamma \hat x,\hat{\Ii})
\end{align}
is the convex conjugate, with respect to the first argument of $\omega_\phi(\cdot,\Gamma \hat x,\hat{\Ii})$.

Two  iterative primal-dual algorithms are used to solve these problems. They involve the biased variables $(\hat z^k,\hat x^k)$ and the refitted ones $(\tilde z^k,\tilde x^k)$.
Let us now present the whole algorithm, defined for parameters $\kappa > 0$,
$\tau > 0$ and $\theta \in [0, 1]$ as:
\begin{equation}\label{algo_final}
 \left\{ \begin{array}{@{}ll@{\hspace{.4em}}}
   \hat \nu^{k+1}&=\hat z^k+\kappa \Gamma \hat v^k\\
   \tilde \nu^{k+1} &=\tilde z^k+\kappa \Gamma \tilde v^k\\
    \hat \xi^{k+1}&=\Pi(\hat \nu^{k+1},\lambda)\\
    \hat{\Ii}^{k+1}&=\enscond{i \in [m]}{\norm{\hat\nu_i^{k+1}}_2>\lambda+\beta}\\
    \tilde \xi^{k+1}&=\prox_{\kappa \omega_\phi^*}(\tilde \nu^{k+1},\Psi(\hat \nu^{k+1}),\hat{\Ii}^{k+1})\\
    \hat \x^{k+1}&=\Phi_\tau ^{-}\left(\hat\x^k+\tau( \Phi^t y-\Gamma^t \hat \xi^{k+1} )\right)\\
    \xt^{k+1}&=\Phi_\tau ^{-}\left(\xt^k+\tau(\Phi^t y-\Gamma^t\tilde\xi^{k+1})\right)\\
    \hat v^{k+1}&=\hat\x^{k+1}+\theta(\hat\x^{k+1}-\hat\x^k)\\
    \tilde v^{k+1}&=\xt^{k+1}+\theta(\xt^{k+1}-\xt^{k}),
\end{array}\right.
\end{equation}%
with the operator $\Phi_\tau ^{-}=(\Id+\tau\Phi^t\Phi)^{-1}$.  
The function $\Psi$ is considered to approximate the value of $\Gamma \hat x$ from the current numerical variables and will be specicifed in the next paragraph on {\em  Online direction and norm identification}, while the functions $\Pi$ and $\prox_{\kappa \omega_\phi^*}$ are  detailed below.
As we will see in Proposition \ref{prop:support}, the quantity $\beta > 0$ is used to
guarantee that we estimate the support of $\hat{x}$ correctly. In practice, we choose $\beta$
as the smallest available non-zero floating number.
The process for the biased variable  involves the orthogonal projection on the $\ell_2$ ball of radius $\lambda$ in $\RR^b$
$$\Pi(\hat \xi,\lambda)_i=\frac{\lambda\hat\xi_i}{\max(\lambda,\norm{\hat \xi_i}_2)}.$$
wheras the refitting algorithm relies on
the proximal operator  of $\omega^*_\phi$.
We recall that the proximal operator of a convex function $\omega$ at point $\xi_0$ reads
\begin{align}
  \prox_{\kappa \omega}(\xi)=\argmin_{\xi}\tfrac1{2\kappa}\norm{\xi-\xi_0}^2_2+\omega(\xi)~.
\end{align}
From the block structure of the function $\omega_\phi$ defined in \eqref{def:omega},  the computation of its  proximal operator  may be realized  pointwise. Since $\iota_{\{0\}}(\xi)^*=0$, we have
\begin{align}
  \prox_{\kappa \omega_\phi^*}(\xi^0,\Gamma \hat x,\hat{\Ii})_i&=
  \begin{cases}
    \prox_{\kappa \phi^*}(\xi^0_i,(\Gamma \hat x)_i), &\hspace{-0.2cm}\textrm{if } i\in \hat{\Ii}\,,\\
    \xi^0_i, &\hspace{-0.2cm}\textrm{otherwise}\,.
  \end{cases}
\end{align}
Table \ref{tab:prox_phi_dual} gives the expressions of the dual functions $\phi^*$ with respect to their first variable and their related proximal operators $\prox_{\kappa \phi^*}$ for the refitting block penalties considered in this paper. All details are given in the Appendix
\ref{sec:prox_block_proof}.

Following \cite{CP}, for  any positive scalars $\tau$ and  $\kappa$ satisfying $\tau\kappa \norm{\Gamma^t\Gamma}_2<1$ and $\theta\in[0,1]$, the estimates $(\hat  \xi^k,\hat x^k,\hat v^k)$ of the biased solution converge to $(\hat \xi, \hat x,\hat x)$, where $(\hat \xi,\hat x)$ is  a saddle point  of \eqref{isotv:pd}.
When the  last arguments  of the function $\omega^*_\phi$ are the converged  $\Gamma \hat x$ and its support $\hat \Ii$, the refitted variables converge to a saddle point of  \eqref{general_refit_reform:pd}.  However, we did not succeed to show  convergence for the refitting process, since the quantities $\Gamma \hat x$ and $\hat \Ii$ are  only estimated from the biased variables at the current iteration, as explained in the next paragraphs.

\paragraph{Online support identification.}
Estimating $\supp(\Gamma \hat x)$ from an estimation $\hat x^k$ is not stable numerically:
the support $\supp(\Gamma \hat x^k)$ can be far from $\supp(\Gamma \hat x)$  even though $\hat x^k$ is arbitrarily close to $\hat x$.
As in \cite{Deledalle_Papadakis_Salmon15}, we rather consider  the dual variable $\hat \xi^k$ to estimate the support.
We indeed expect  (see for instance  \cite{BGMEC16}) at convergence $\hat \xi^k$  to   saturate  on the  support of $\Gamma \hat x$ and to satisfy the optimality condition
$\hat \xi_i^k=\lambda\tfrac{(\Gamma \hat x)_i}{\norm{(\Gamma \hat x)_i}_2}$.
In practice, the norm of the dual variable $\hat \xi^k_i$ saturates to $\lambda$ relatively fast onto $\hat{\Ii}$.
As a consequence,  it is far more stable to detect the support of $\Gamma \hat x$ with the dual variable $\hat \xi^k$ than with the vector $\Gamma \hat x^k$ itself.
The next propostion, adapted from \cite{Deledalle_Papadakis_Salmon15},
 shows that
this approach can indeed converge towards the support $\hat \Ii$.

\begin{prop}\label{prop:support}
 Let $\alpha>0$ be the minimum non zero value of $\norm{(\Gamma \hat x)_i}_2$, and choose $\beta$ such that $\alpha\kappa>\beta>0$.
Then, denoting $\hat\nu^{k+1}=\hat \xi^k+\kappa \Gamma \hat v^k$,  $\hat \Ii^{k+1}=\enscond{i \in [m]}{\norm{\hat\nu_i^{k+1}}_2>\lambda+\beta}$ in Algorithm \eqref{algo_final} converges in finite time to the  true support $\hat{\Ii} = \supp(\Gamma\hat{x})$  of the biased solution\footnote{As  in \cite{brinkmann2016bias}, the extended support  $\norm{\hat \xi_i}_2=\lambda$ can be tackled  by testing  $\norm {\hat \nu^{k+1}_i}\geq\lambda$.}.
\end{prop}
\begin{proof}We just give a sketch of the proof. More details can be found in \cite{Deledalle_Papadakis_Salmon15}.
On the support, one has $(\hat  \xi^k_i,\hat x^k_i,\hat v^k_i)\to (\lambda \Gamma\hat x^k_i/\norm{(\Gamma\hat x^k)_i}_2,\hat x^k_i,\hat x^k_i)$.
Then for $k$ sufficiently large, $\norm{\hat \xi_i^k+\kappa (\Gamma \hat v^k)_i}\leq\lambda+\beta$ if and only if $i\in\Ii^c$.
\end{proof}
\paragraph{Online direction and norm identification.}
In algorithm  \eqref{algo_final}, the function
$\Psi(\hat\nu^{k+1})$ aims at approximating $\Gamma \hat x$ .
As for the estimation of the support, instead of directly considering $\Gamma \hat x^k$, we rather rely on $\hat\nu^{k+1}=\hat \xi^k+\kappa \Gamma \hat v^k $
to obtain a stable estimation $\hat z$ of the vector $\Gamma \hat x$.

Table \ref{tab:prox_phi_dual} shows that for all the considered block penalties, the computation of the proximal operator of $\phi^*$  involves the normalized vector $\hat z_i/\norm{\hat z_i}_2$ on the support $\hat \Ii$.
This direction is approximated at each iteration by normalizing $\hat \nu_i^{k+1}$.
The amplitude $\norm{\hat z_i}_2$ is just required for the models QO and QD.

The method in \cite{deledalle2016clear}  relies on the QO block penalty (see  section  \ref{sec:equivalence} for more details) and considers the algorithmic differentiation of the biased process to obtain a refitting algorithm. Given its good numerical performance, we leverage this approach to define the function $\Psi$
as shown in the next proposition. 
\begin{prop}\label{prop:Psiclear}
The approach of \cite{deledalle2016clear} corresponds to refit with the QO block penalty and the following convergent estimation $\Psi(\hat \nu^{k+1}_i)$ of $(\Gamma \hat x)_i$.
\begin{align}\label{Psi_function}
  \Psi(\hat \nu^{k+1}_i)
  =
  \tfrac{\norm{\hat\nu^{k+1}_i}_2 - \lambda}{\kappa \norm{\hat\nu^{k+1}_i}_2}
 \nu^{k+1}_i.
\end{align}

\end{prop}

\begin{proof}
The method in  \cite{deledalle2016clear} realizes the refitting through an
algorithmic differentiation of the projection  of the biased process.
This leads to the Algorithm \eqref{algo_final} except for the update of the dual variable $\tilde \xi_i$ on the support $i\in \hat \Ii^{k+1}$ that reads  in  \cite{deledalle2016clear}:
$$\tilde \xi^{k+1}_i=\tfrac{\lambda}{\norm{\hat \nu^{k+1}_i}_2}\left( \tilde \nu^{k+1}_i -P_{\hat \nu^{k+1}_i}(  \tilde \nu^{k+1}_i)\right),$$
instead of $$\tilde \xi^{k+1}_i=\prox_{\kappa {\phi^*_{QO}}(\cdot,\Psi(\hat \nu^{k+1}_i),\hat{\Ii}^{k+1} )}(\tilde \nu^{k+1}_i)$$ in Algorithm \eqref{algo_final}.
Using the proximal operator given by the  QO block penalty in Table \ref{tab:prox_phi_dual} gives :
${\norm{\hat \nu^{k+1}_i}_2}={\lambda+\kappa \norm{\Psi(\nu^{k+1}_i)}_2}$
which leads to the  function \eqref{Psi_function}  for estimating $\hat z^k_i$.


We deduce that $\Psi(\hat \nu^{k+1}_i)=\Psi(\hat \xi^k+\kappa \Gamma \hat v^k)\to(\Gamma \hat x)_i$ from  the convergence  of the biased variables $(\hat  \xi^k_i,\hat x^k_i,\hat v^k_i)\to (\lambda \Gamma\hat x^k_i/\norm{(\Gamma\hat x^k)_i}_2,\hat x^k_i,\hat x^k_i)$.
\end{proof}

\paragraph{Discussion.} This joint-estimation  considers at every iteration $k$  different  refitting functions $\omega_\phi^*(.,\Psi(\hat\nu^{k+1}),\hat{\Ii}^{k+1})$ in  \eqref{general_refit_reform}.
Then, unless $b=1$ (see \cite{Deledalle_Papadakis_Salmon15}),
we cannot show the convergence of the refitting scheme.
As in \cite{deledalle2016clear}, we nevertheless observe convergence and
a  stable behavior for this algorithm.

In addition to its better numerical stability,  the running time of  joint-refitting  is more interesting than the posterior approach.
In Algorithm \eqref{algo_final},  the refitted variables at iteration $k$ require the biased variables at the same iteration and the whole process  can be realized in parallel without significantly affecting the running time of the  original biased process. On the other hand, posterior refitting is necessarily sequential and  the running time is doubled in general.

\subsection{Douglas-Rachford formulation}

An alternative to obtain solutions $\hat x$ of \eqref{isotv} and $\tilde x$ of \eqref{general_refit_reform}
is to consider the splitting
TViso reformulation, as proposed in \cite{combettes2007douglas}, and given by
\begin{align}\label{isotv:dr}
\min_{\substack{x \in \RR^n \\ \xi \in \RR^b}}\;& \tfrac12 \norm{\Phi x - y }_2^2 + \iota_{\{x,\xi; \Gamma x=\xi\}}(x,\xi) +\lambda\norm{\xi}_{1,2}\enspace,\\
\label{refit:dr}
\min_{\substack{x \in \RR^n \\ \xi \in \RR^b}}\;& \tfrac12 \norm{\Phi x - y }_2^2 + \iota_{\{x,\xi; \Gamma x=\xi\}}(x,\xi) +\omega_\phi(\xi,\Gamma \hat x,\hat{\Ii})\enspace.
\end{align}
This problem can be solved with the Douglas-Rachford algorithm
\cite{douglas1956numerical,lions1979splitting}. Introducing 
the parameters $\alpha  \in(0,2)$ and $\tau>0$, the iterates read

\newcommand{\dualz}{z} 
\begin{equation}\label{eq:algo}
\left\{
  \begin{array}{@{\;}r@{\;=\;}l}
      \hat\upsilon^{k+1}    &\Gamma^{-}(2\hat x^k \!-\! \hat\mu^k \!+\Gamma^t(2  \hat\xi^k \!-\!\hat\zeta^k)),\\
                   \tilde\upsilon^{k+1}    &\Gamma^{-}(2\tilde x^k \!-\! \tilde\mu^k \!+\Gamma^t(2  \tilde\xi^k \!-\!\tilde\zeta^k)),\\
    \hat\mu^{k+1}    & \hat\mu^k +\alpha(\hat\upsilon^{k+1}  -\hat x^k),\\
   \tilde{\mu}^{k+1}    &  \tilde{\mu}^{k}+\alpha(\tilde\upsilon^{k+1} - \tilde x^k),\\
          \hat \zeta^{k+1}& \hat \zeta^{k} +\alpha(\Gamma\hat\upsilon^{k+1} -\hat\xi^k),\\
    \tilde{\zeta}^{k+1} &\tilde{\zeta}^{k}+\alpha (\Gamma\tilde\upsilon^{k+1} -\tilde\xi^k),\\
    \hat x^{k+1} &  \Phi_\tau ^{-} (\hat\mu^{k+1}+\tau\Phi^ty),\\
    \tilde{x}^{k+1} & \Phi_\tau ^{-} (\tilde\mu^{k+1}+\tau\Phi^ty),\\
    \hat  \xi^{k+1} & \rm{ST}(\hat \zeta^{k+1}, \tau\lambda), \\
        \hat{\Ii}^{k+1}&\enscond{i \in [m]}{\norm{\hat \zeta^{k+1}_i}_2 > \tau \lambda + \beta},\\
    \tilde{\xi}^{k+1}   & \prox_{\tau \omega_\phi(\cdot, \Upsilon(\hat \zeta^{k+1}),\hat{\Ii}^{k+1} )}(\tilde \zeta^{k+1})\\
  \end{array}
\right.
\!\!\!
\end{equation}
with $\Gamma^{-}=(\Id + \Gamma^t\Gamma)^{-1} $,
the block Soft Thresholding (ST) operator
\begin{eqnarray*}
  \rm{ST}(\hat \zeta, \lambda)_i
  &\!=\!&
  \left\{
    \begin{array}{ll}
      0 & \ifq \norm{\hat\zeta_i}_2 \leq \lambda,\\[-0.1em]
      \hat \zeta_i - \lambda \frac{\hat \zeta_i}{ \norm{\hat \zeta_i}_2}\enspace & \otherwise\enspace.\\
    \end{array}
  \right.\\
\end{eqnarray*}
and the proximal operator of $\omega_\phi$ that can be written pointwise as
\begin{align}
  &\prox_{\tau \omega_\phi}(\zeta^0,\Gamma \hat x,\hat{\Ii})_i\nonumber\\
  =&
  \begin{cases}
    \zeta^0-\tau \prox_{\phi^*/\tau}(\zeta^0_i/\tau,(\Gamma \hat x)_i), & \textrm{if } i\in \hat{\Ii}\enspace,\\
   0, & \textrm{otherwise}\enspace.
  \end{cases}
\end{align}
The estimates $(\hat x^k,\hat \xi^k,\hat \zeta^k)$ of the biased solution converge  to $(\hat x,\hat \xi,\hat \zeta)=(\hat x,\Gamma \hat x,(1+\tau\lambda)(\Gamma \hat x))$, from the optimality conditions of problem \eqref{isotv:dr}.  On the other hand, there is again no convergence guarantee for the refitted variables since the  last two arguments of the function $\omega_\phi$ are potentially modified at each iteration.

The support of $\Gamma \hat x$ is estimated in line from auxiliary variables as $\hat \Ii^{k+1}=\enscond{i \in [m]}{\norm{\hat\zeta_i^{k+1}}_2>\tau \lambda+\beta}$, where, as shown in   \cite{deledalle2015contrast},  $\beta$ has to be taken such that $0<\beta<\lambda\min_{i\in\hat\Ii}\norm{(\Gamma \hat x)_i}_2$ to have convergence in finite time of $\hat \Ii^{k+1}$ to the  true support $\hat{\Ii} = \supp(\Gamma\hat{x})$.
Again the quantity $\beta > 0$ is chosen in practice as being the smallest available non-zero floating number.

Considering the stable algorithmic differentiation strategy of \cite{deledalle2016clear} suggested in the previous subsection, the vector $\Gamma \hat x$ is approximated on the support $\hat \Ii^{k+1}$ at each iteration  with the function
$\Upsilon(\hat \zeta)=\frac{\norm{\hat \zeta}_2-\lambda\tau}{\norm{\hat \zeta}_2}\hat \zeta$.

\section{Related re-fitting works}\label{sec:related}

\begin{table*}[!t]
  \centering
  \caption{\label{tab:prox_phi_dual} Convex conjugates and proximal operators of the studied block penalties $\phi$.
  }
  \begin{tabular}{l@{\hspace{1cm}}l@{\hspace{1cm}}l}
    \hline
    \multicolumn{1}{c}{$\phi$}& \multicolumn{1}{c}{$\phi^*(z,\hat z)$}& \multicolumn{1}{c}{$\prox_{\kappa \phi^*}(z_0,\hat z)$}\\
    \hline
    \hline
    HO&$
    \choice{
      0,
      & \quad\quad \ifq\cos (z, \hat{z}) =0\\
      +\infty,
      & \quad\quad
      \otherwise
    }$
    &$z_0 -P_{\hat z}(z_0)$\\
    \hline
    HD&
    $\choice{
      0,
      & \quad\quad \ifq \cos(z, \hat{z}) \leq 0\\
      +\infty,
      & \quad\quad
      \otherwise
    }$
    &$ \choice{
      z_0 -P_{\hat z}(z_0),
      & \quad \ifq \langle z_0,\hat z\rangle\geq 0\\
      z_0,
      &\quad
      \otherwise
    }$\\
    \hline
    QO&$\choice{
      \tfrac{\norm{\hat z}_2}{2\lambda}\norm{z}^2_2, &\;\;\, \ifq \cos(z, \hat{z})  =0\\
      +\infty,
      &\;\;\,
      \otherwise
    }$&$\tfrac{\lambda}{\lambda+\kappa \norm{\hat z}_2}\left(z_0 -P_{\hat z}(z_0)\right)$
    \\
    \hline
     QD&$
    \choice{
      \frac{\norm{\hat z}_2}{2\lambda}\norm{z}_2^2 & \ifq \cos\theta(z, \hat{z})  \leq 0\\
     +\infty
     &
      \otherwise
    }$&$ \frac{\lambda}{\lambda+\kappa \norm{\hat z}_2}\choice{
   z_0 -P_{\hat z}(z_0)
      & \ifq \langle z_0,\hat z\rangle\geq 0\\
      z_0
      &
      \otherwise
    }$\\
    \hline
     SO&
     $ \choice{
      0,
      & \quad\quad \ifq\cos(z,\hat z)=0 \qandq \norm{z}_2\leq \lambda\\
      +\infty,
      &\quad\quad
      \otherwise
    }$
  &
    $\lambda\tfrac{z_0 -P_{\hat z}(z_0)}{\max(\lambda,\norm{z_0 -P_{\hat z}(z_0)}_2)}$\\
     \hline
    SD&
    $ \choice{
      0,
      & \quad\quad \ifq\norm{z+\lambda \tfrac{\hat z}{\norm{\hat z}_2}}\leq \lambda {}^\dag\\
      +\infty,
      &\quad\quad
      \otherwise
    }$
    &$\lambda\left(\tfrac{z_0+\lambda\tfrac{\hat z}{\norm{\hat z}_2}}{\max(\lambda,\norm{z_0+\lambda \tfrac{\hat z}{\norm{\hat z}_2}}_2)}-\tfrac{\hat z}{\norm{\hat z}_2}\right)$\\
    \hline
  \end{tabular}
  \\
    {\scriptsize ${}^\dag$: note that the condition implies that $\cos(z,\hat z)\leq 0$.}
\end{table*}

We now review some other related refitting methods.
First we will discuss of refitting methods based on Bregman divergences,
next the refitting approch developped in \cite{deledalle2016clear}, and then we will see how
these techniques are related to the block penalites introduced in Section \ref{sec:refitting}.

\subsection{Bregman-based Refitting}

\subsubsection{Bregman divergence of $\ell_{12}$ structured regularizers}\label{sec:intro_breg}
In the literature \cite{Osher,brinkmann2016bias}, Bregman divergences have proven to be well suited to measure the discrepancy between the biased solution $\hat x$ and its refitting $\tilde x$.
We recall that for a convex, proper and lower semicontinuous function $\psi$, the associated (generalized) Bregman divergence between $x$ and $\hat x$ is, for any subgradient $\hat p\in \partial \psi(\hat x)$:
\begin{align}
  D^{\hat{p}}_\psi(x,\hat x)=\psi(x)-\psi(\hat x)-\langle \hat p,x-\hat x\rangle\geq 0\enspace.
\end{align}
If $\psi$ is an absolutely 1-homogeneous function, i.e.,
$\psi(\alpha x)=\abs{\alpha} \psi(x)$, $\forall\alpha \in\RR$,
then
\begin{align}\label{eq:sub_grad_prop_1homo}
  p\in \partial \psi(x)\Rightarrow \psi(x)=\dotp{p}{x}\enspace,
\end{align}
and the Bregman divergence simplifies into
\begin{align}
  D^{\hat{p}}_\psi(x,\hat x)=\psi(x)-\langle \hat p,x\rangle\enspace.
\end{align}

As an example, let us consider $\psi(x)=\norm{x}_2$.
Since $\psi$ is 1-homogeneous,
it follows that
\begin{align}\label{sub_l12}
  & D^{\hat{p}}_{\norm{\cdot}_2}(x,\hat x)
  =
  \norm{x}_2 - \dotp{\hat p}{x}
  \\
  \text{where } &
  \hat p \in \partial \norm{\cdot}_2(\hat x)
  \Leftrightarrow
  \left\{\begin{array}{ll}
  \hat p = \tfrac{\hat x}{\norm{ \hat x}_2} \quad \quad & \textrm{if } \hat x \ne 0 \enspace,\\
  \norm{\hat p}_2\leq 1&\textrm{otherwise}\enspace.
  \end{array}\right.
\end{align}

For regularizers of the form $\psi(x)=\norm{\Gamma x}_{1,2}$,
we introduce the following notations
\begin{align}
  &\Omega(\hat x) = \partial\norm{\Gamma \cdot}_{1,2}(\hat x)~,\\
  &\Delta(\hat x) =
  \enscond{
    \hat \eta
  }{
    \Gamma^t \hat \eta \in \partial\norm{\Gamma \cdot}_{1,2}(\hat x)
  }~.\label{Delta}
\end{align}

For all $\hat p \in \RR^n$, we have  \cite{BGMEC16}:
\begin{align}\label{opt_z}
  \hat p \in \Omega(\hat x)
  \Leftrightarrow\;
  &
  \exists \hat{\eta} \in \Delta(\hat x) \text{ s.t. }
  \hat p=\Gamma^t  \hat{\eta}
\end{align}%
or in short
$\Omega(\hat x) = \Gamma^t \Delta(\hat x)$.
Interestingly $\Delta(x)$ enjoys a separability property in terms of all subgradients
associated to the $\ell_2$ norms of the blocks
\begin{align}\label{separation}
  \Delta(\hat x)
  =
  \enscond{
    \hat \eta \in \RR^m
  }{
    \hat{\eta}_i\in\partial \norm{\cdot}_2((\Gamma \hat x)_i),\, \forall i \in [m]
  }~.
\end{align}
Additionally, the next proposition shows that there is a similar separability property
for the Bregman divergence. 
\begin{prop}\label{prop:Breg_separation}
Let $\hat{\eta} \in \Delta(\hat x)$ and $\hat p = \Gamma^t \hat{\eta}$. Then
\begin{align}\label{Breg_separation}
  D^{\hat{p}}_{\norm{\Gamma \cdot}_{1,2}} (x,\hat x)
  & = \sum_{i=1}^m D^{\hat{\eta}_i}_{ \norm{\cdot}_2} ((\Gamma x)_i,(\Gamma \hat x)_i)~.
\end{align}
\end{prop}

\begin{proof}
  The proof is obtained with straigthforward computations:
  \begin{align}\label{breg}
    D^{\hat{p}}_{\norm{\Gamma \cdot}_{1,2}} (x,\hat x)
    & =\norm{\Gamma x}_{1,2} -\langle \Gamma^t \hat{\eta}, x \rangle\\
    & =\norm{\Gamma x}_{1,2} -\langle \hat{\eta}, \Gamma x \rangle\\
    & =\sum_{i=1}^m
    \underbrace{\norm{(\Gamma x)_i}_2 -\langle  \hat  {\eta}_i, (\Gamma x)_i\rangle}_{
      D^{\hat{\eta}_i}_{ \norm{\cdot}_2} ((\Gamma x)_i,(\Gamma \hat x)_i)
    }
    \enspace.
  \end{align}
\end{proof}

In the following we consider $\hat{x}$ and its support $\hat{\Ii}$
to be fixed, and we denote by $D^{\hat{\eta}}_i(\Gamma x)$ the following
\begin{align}\label{D_i}
  D^{\hat{\eta}}_i(\Gamma x) &=
  D^{\hat{\eta}_i}_{ \norm{\cdot}_2} ((\Gamma x)_i,(\Gamma \hat x)_i)~.
\end{align}
The next proposition shows that such a divergence measures the fit of
directions between $(\Gamma x)_i$ and $(\Gamma \hat x)_i$,
but also partially captures the support $\hat \Ii$.

\begin{prop}\label{prop:zero_breg}
  Let $\hat{\eta} \in \Delta(\hat{x})$.
  We have for $i\in \hat{\Ii}$
\begin{equation}
  D_i^{\hat{\eta}}(\Gamma x)=0 \Leftrightarrow \exists
  \alpha_i\geq 0\textrm{ s.t. }(\Gamma  x)_i=\alpha_i (\Gamma \hat x)_i \enspace.
  \label{sub_supp1}
\end{equation}
and we have for $i\in \hat{\Ii}^c$
\begin{align}
  &D_i^{\hat{\eta}}(\Gamma x) = 0
  \Leftrightarrow
       {(\Gamma x)_i}=0_b   \text{ or }   \hat{\eta}_i = \frac{(\Gamma x)_i}{\norm{(\Gamma x)_i}_2}\enspace.
       \label{sub_supp2}
\end{align}
\end{prop}

\begin{proof}
  $\bullet$ For $i\in\hat{\Ii}$, i.e. $\norm{(\Gamma x)_i}_2>0$, we get $ \hat{\eta}_i = \frac{(\Gamma \hat x)_i}{\norm{(\Gamma \hat x)_i}_2}$ from \eqref{separation}. We conlude from the definition $D_i^{\hat{\eta}}(\Gamma x)=\norm{(\Gamma x)_i}_2-\langle \hat \eta_i,(\Gamma x)_i\rangle=0$.
  \medskip

  \noindent
  $\bullet$ For $i\in\hat{\Ii}^c$ and using  \eqref{separation} we distinguish two cases from \eqref{sub_l12}. If $\norm{\hat \eta_i}_2<1$, then $D_i^{\hat{\eta}}(\Gamma x)=0$ iff ${(\Gamma x)_i}=0_b$. Otherwise one necessarily gets $ \hat{\eta}_i = \frac{(\Gamma x)_i}{\norm{(\Gamma x)_i}_2}$.
\end{proof}

As noticed in the image fusion model in \cite{ballester2006variational}, minimizing $D_i^{\hat{\eta}}(\Gamma x)$ enforces the alignment of  the direction $(\Gamma x)_i$ with $\hat \eta_i$ and it is an efficient way to avoid  contrast inversion.

In the following sections, unless stated otherwise, we will always
consider $\hat{\eta} \in \Delta(\hat x)$ and
$\hat p = \Gamma^t \hat{\eta} \in \Omega(\hat x)$.

\subsubsection{Iterative Bregman regularization}
The  Bregman process \cite{Osher} reduces the bias of solutions of \eqref{isotv}
by successively solving problems of the form
\begin{align}\label{iterative_OSHER}
  &
  \tilde{x}_{l+1} \in\uargmin{x \in \RR^n}
  \tfrac12 \norm{\Phi x - y }_2^2 + \lambda  D^{\tilde{p}_l}_{\norm{\Gamma \cdot}_{1,2}} (x,\tilde{x}_{l})~.
\end{align}
with $\tilde{p}_l \in \Omega(\tilde{x}_l)$.
We consider a fixed $\lambda$, but different strategies can be considered with decreasing parameters $\lambda_l$ as in \cite{Schetzer01,Tadmor04}.
For $l=0$, setting  $\tilde{x}_{0}=0_n$, and taking  $\tilde{p}_0=0_n\in \Omega(\tilde{x}_{0})$ so that
$D^{\tilde{p}_0}_{\norm{\Gamma \cdot}_{1,2}} (x,\tilde{x}_{0})= \norm{\Gamma x}_{1,2}$, the first step exactly gives the biased solution of \eqref{isotv} with $\tilde{x}_{1}=\hat x$. 
We denote by $\tilde{x}^{\mathrm{IB}(\hat p)}=\tilde{x}_{2}$ the solution obtained after $2$ steps of the Iterative Bregman (IB)  procedure  \eqref{iterative_OSHER}:
\begin{align}\label{iterative_refit}
  \tilde{x}^{\mathrm{IB}(\hat{p})}=\tilde{x}_{2}\in\uargmin{x \in \RR^n} \tfrac12 \norm{\Phi x - y }_2^2 +
  \lambda D^{\hat{p}}_{\norm{\Gamma \cdot}_{1,2}} (x,\hat x)\enspace.
\end{align}
As underlined in relation \eqref{sub_supp1}, by minimizing
$D^{\hat{p}}_{\norm{\Gamma \cdot}_{1,2}} (x,\hat x) = \sum_{i=1}^m D_i^{\hat{\eta}}(\Gamma x)$,
one aims at preserving the direction of $\Gamma \hat x$ on the support $\hat{\Ii}$, without ensuring $\supp(\Gamma \tilde x^{\textrm{IB}(\hat{\p})})\subseteq \hat{\Ii}$.

For the iterative framework, the support of the previous solution may indeed not be preserved ($\norm{(\Gamma \tilde x_l)_i}_2=0 \nRightarrow \norm{(\Gamma \tilde x_{l+1})_i}_2=0$) and can hence grow.
The support of $\Gamma x_0$ for $\tilde{x}_{0}=0_n$ is for instance totally empty whereas the one of  $\hat x= \tilde{x}_{1}$ may not  (and should not) be empty.
For $l\to\infty$, the process actually converges to some $x$ such that $\Phi x=y$.
Because the IB procedure does not preserve the support of the solution,
it cannot be considered as a refitting procedure and is more related to
boosting approaches as discussed in Section \ref{sec:intro}.

\subsubsection{Bregman based refitting}

In order to respect the support of the biased solution $\hat x$ and to keep track of
the direction $\Gamma \hat x$ during the refitting, the authors of \cite{brinkmann2016bias}
proposed the following model:
\begin{align}\label{refit_german:HD}
  &
  \tilde{x}^{\mathrm{B}(\hat{\p})} \in\uargmin{x; \hat \p \in \Omega(x)}
  \tfrac12 \norm{\Phi x - y}_2^2 \enspace.
\end{align}
This model enforces  the  Bregman divergence  to be  $0$,
since, from eq.~\eqref{eq:sub_grad_prop_1homo}, we have:
\begin{align}\label{prop_HD}
  \hat p \in \Omega(x)
  \Rightarrow
  D^{\hat{p}}_{\norm{\Gamma \cdot}_{1,2}} (x,\hat x)
  = 0
  \enspace.
\end{align}

We see from \eqref{sub_supp1} that for $i\in\hat{\Ii}$,
the direction of $(\Gamma \hat x)_i$ is preserved in the refitted solution.
From \eqref{sub_supp2}, we also observe that the absence of support is also preserved
for any $i\in \hat{\Ii}^c$ where $\norm{\hat{\eta}_i}_2<1$.
Note that extra elements in the support $\Gamma \tilde x^{\textrm{B}(\hat p)}$
may be added at coordinates $i\in\hat{\Ii}^c$ where $\norm{\hat{\eta}_i}_2=1$.

\subsubsection{Infimal Convolutions of Bregman (ICB) distances based refitting}

To get rid of the direction dependency,
the ICB (Infimal Convolutions of Bregman distances)
model is also proposed in \cite{brinkmann2016bias}:
\begin{align}\label{refit_german:HO}
  &
  \tilde{x}^{\mathrm{ICB}(\hat{p})} \in\uargmin{x; \pm \hat p \in \Omega(x)}
  \tfrac12 \norm{\Phi x - y}_2^2 \enspace.
\end{align}
The orientation model may nevertheless involve contrast inversions between biased and refitted solutions.
In practice, relaxations
are used in \cite{brinkmann2016bias}
by solving, for a large value $\gamma>0$,
\begin{align}\label{refit_german:HD2}
  &
  \tilde{x}_\gamma^{\mathrm{B}(\hat{p})} \in\uargmin{x \in \RR^n}
  \tfrac12 \norm{\Phi x - y}_2^2 +\gamma D^{\hat{p}}_{\norm{\Gamma \cdot}_{1,2}} (x,\hat x)~.
\end{align}
The main advantage of this refitting strategy is that no support identification is required since everything is implicitly encoded in the subgradient $\hat p$. This makes the process stable even if  the estimation of $\hat x$ is not highly accurate.
The support of $\Gamma \hat x$ is nevertheless only approximately preserved, since the constraint
$D^{\hat p}_{\norm{\Gamma \cdot}_{1,2}} (x,\hat x)=0$ can never be ensured numerically with a finite value of $\gamma$.

\subsubsection{The best of both Bregman worlds}
From relations \eqref{Breg_separation} and \eqref{D_i}, the refitting models  given in  \eqref{iterative_refit} and \eqref{refit_german:HD} can be
reexpressed as a function of $\hat{\eta}$
\begin{align}
\label{iterative_refit2}
\tilde{x}^{\mathrm{IB}(\hat{\eta})}&\in\uargmin{x \in \RR^n} \tfrac12 \norm{\Phi x - y }_2^2 +\lambda\sum_{i=1}^m D_i^{\hat{\eta}}(\Gamma x)\enspace,\\
\label{refit_german:HD3}
  \tilde{x}^{\mathrm{B}(\hat{\eta})} &\in \uargmin{x \in \RR^n}
  \tfrac12 \norm{\Phi x - y}_2^2 \,\,\mathrm{ s.t. }\,\,
  D_i^{\hat{\eta}}(\Gamma x)=0,\,\forall i\in [m]\enspace.
\end{align}
Alternatively we now introduce a mixed model, that we coin Best of Both Bregman (BBB), as
\begin{align}\label{SD:bregman}
  \tilde{x}^{\mathrm{BBB(\hat \eta)}} \in \uargmin{x \in \RR^n}
  \tfrac12 \norm{\Phi x - y}_2^2 +\lambda \sum_{i\in\hat{\Ii}} D_i^{\hat{\eta}}(\Gamma x)\enspace,\\
  \quad \quad \mathrm{s.t. }\,\, D_i^{\hat{\eta}}(\Gamma x)=0,\,\forall i\in\hat{\Ii}^c\enspace.
\end{align}
With such reformulations, connections between refitting models \eqref{iterative_refit2}, \eqref{refit_german:HD3} and \eqref{SD:bregman} can be clarified.
The solution $\tilde{x}^{\mathrm{IB}}$ \cite{Osher} is too relaxed, as it only penalizes the directions $(\Gamma  x)_i$ using $(\Gamma \hat x)_i$, without aiming at preserving the support of $\hat x$.
The solution $\tilde{x}^{\mathrm{B}}$ \cite{brinkmann2016bias} is too constrained: the direction within the support is required to be preserved exactly.
The proposed refitting $\tilde{x}^{\mathrm{BBB}}$ lies in-between: it preserves the support, while authorizing some directional flexibility, as illustrated by the sharper square edges in Figure 1(g).

An important difference with BBB is that we consider
local inclusions of subgradients of the function $\lambda\norm{\cdot}_{1,2}$
at point $(\Gamma x)_i$ instead of the global inclusion of subgradients of
the function $\lambda\norm{\Gamma \cdot}_{1,2}$ at point $x$ as in
\eqref{refit_german:HD} and \eqref{refit_german:HO}. Such a change of
paradigm allows to adapt the refitting locally by preserving the
support while including the flexibility of the original Bregman
approach \cite{Osher}.

\subsection{Covariant LEAst Square Refitting (CLEAR)}

We now describe an alternative way for performing variational refitting.
When specialized to $\ell_{1,2}$ sparse analysis regularization, CLEAR, a general refitting framework \cite{deledalle2016clear}, consists in computing
\begin{align}\label{refit:QO}
&  \tilde{x}^{\mathrm{CLEAR}} \in \uargmin{x; \; \supp(\Gamma x) \subseteq \hat{\Ii}}
  \tfrac12 \norm{\Phi x - y}^2_2
  \nonumber
  \\
& \quad +
   \sum_{i\in\hat{\Ii}} \tfrac{\lambda}{2\norm{(\Gamma\hat x)_i}_2} \snorm{
   (\Gamma  x)_i-P_{(\Gamma \hat x)_i}((\Gamma x)_i)
   }^2\enspace,
\end{align}
where we recall that $P_{(\Gamma \hat x)_i}(.)$ is the orthogonal projection onto $\Span(\Gamma \hat x)_i$.
This model promotes refitted solutions preserving to some extent the orientation
$\Gamma \hat x$ of the biased solution. It also shrinks the amplitude
of $\Gamma x$ all the more that the amplitude of $\Gamma \hat x$
are small.
This penalty does not promote any kind of direction preservation,
and as for the ICB model, contrast inversions may be observed between
biased and refitted solutions.
The quadratic term also over-penalizes large changes of orientation.

\subsection{Equivalence}\label{sec:equivalence}

In the next proposition we show that, for a given vector $\etaspec$,
the concurrent Bregman based refitting techniques coincides with two
of the refitting block penalties introduced in Section \ref{sec:refitting},
namely HD, HO, and that SD is nothing else than the proposed
Best of Both Bregman based penalty.
\begin{prop}\label{prop:breg_equivalence}
  Let $\etaspec \in \RR^m$ be defined as
  \begin{align}
  \etaspec = \arg \min_{\eta \in \Delta(\hat{x})} \norm{\eta}_2 ~.
  \end{align}
  Then, for any $y$, we have the following equalities
  \begin{enumerate}
      \item $\tilde{x}^{\rm B(\etaspec)} = \tilde{x}^{\rm HD}$~,
    \item $\tilde{x}^{\rm ICB(\etaspec)} = \tilde{x}^{\rm HO}$~,
    \item $\tilde{x}^{\rm BBB(\etaspec)} = \tilde{x}^{\rm SD}$~.
  \end{enumerate}
  where the equalities have to be understood as an equality
  of the corresponding sets of minimizers.
\end{prop}
\begin{proof}[Proposition \ref{prop:breg_equivalence}]
  First notice that from \eqref{sub_l12}, $ \etaspec$ can be written explicitly as
  \begin{align}  \etaspec_i=
    \choice{\frac{(\Gamma \hat x)_i}{\norm{(\Gamma \hat x)_i}_2}
      \quad{ }
      & \text{if } i \in  \hat\Ii~,\\
      0_b & \text{otherwise,}
    } \quad
    \text{ for all $i \in [m]$}\enspace.
  \end{align}

  \noindent
  $\bullet$ [1.] This is a direct consequence of relation \eqref{prop_HD} and Proposition \ref{prop:zero_breg}.
  Using \eqref{sub_supp1}, we have for $ i\in \hat\Ii$ that $D_i^{\etaspec}((\Gamma x)_i)=0\Leftrightarrow  \cos((\Gamma x)_i, (\Gamma \hat x)_i)=1 \Leftrightarrow \cos((\Gamma x)_i, \etaspec_i)=1$. This is also valid for potential vanishing  components $(\Gamma x)_i$ with the considered convention  $\cos(0_b,\hat z)=0$.
  We thus recover the penalty function \rm{HD} in\eqref{pen:HD}. Next, for all $i \in  \hat\Ii^c$, we have $\etaspec_i = 0$ by assumption,
  hence according to eq.~\eqref{sub_supp2}:
  \begin{align}
    D_i^{\hat{\eta}}(\Gamma x)=0
    \Leftrightarrow (\Gamma x)_i = 0_b
    \Leftrightarrow \supp(\Gamma x) \subseteq  \hat\Ii .
  \end{align}
  It follows that $\tilde{x}^{\rm B(\etaspec)} = \tilde{x}^{\rm HD}$.  This case $ \hat\Ii^c$ is the same for all next points.
  \medskip

  \noindent
  $\bullet$ [2.] With the ${\rm{ICB}}$ model, we have $\pm \hat p \in \Omega(x)$. For $i\in  \hat\Ii$, it gives with  \eqref{sub_supp1} that
  $D_i^{\pm \etaspec}(\Gamma x)=0\Leftrightarrow \norm{(\Gamma  x)_i)}_2=\pm \langle (\Gamma  x)_i),\hat\eta_i\rangle$. This is equivalent to $\abs{\cos((\Gamma x)_i, (\Gamma \hat x)_i)}=1$ that corresponds to the penalty function \rm{HO} in \eqref{pen:HO}.
  \medskip

  \noindent
  $\bullet$ [3.]
  For all $i \in  \hat\Ii$, we have
  \begin{align}
    D_i^{\hat{\etaspec}}((\Gamma x)_i)
    &=
    \norm{(\Gamma x)_i}_2 - \dotp{\frac{(\Gamma \hat x)_i}{\norm{(\Gamma \hat x)_i}_2}}{(\Gamma x)_i}
    \\
    &=
    \norm{(\Gamma x)_i}_2(1 - \cos((\Gamma x)_i, (\Gamma \hat x)_i))~,
  \end{align}
  that gives  the penalty function \rm{SD} in \eqref{pen:SD}.
  We then get that $\tilde{x}^{\rm BBB(\etaspec)} = \tilde{x}^{\rm SD}$.
\end{proof}
Additionaly, we can show that CLEAR corresponds to the QO refitting block
penalty also introduced in Section \ref{sec:refitting}.
\begin{prop}\label{prop:equiv_clear}
  For any $y$, we have
  $\tilde{x}^{\rm CLEAR} = \tilde{x}^{\rm QO}$
  where the equality has to be understood as an equality
  of the sets of minimizers.
\end{prop}

\begin{proof}
The case $\hat \Ii^c$ is treated in Proposition \ref{prop:breg_equivalence}. For $i\in \hat\Ii$ with the \rm{CLEAR} model, we just have to observe that
$
\snorm{ (\Gamma  x)_i-P_{(\Gamma \hat x)_i}((\Gamma x)_i)}^2=
  \norm{(\Gamma x)_i}^2_2(1- \cos^2((\Gamma x)_i,(\Gamma \hat x)_i)),$
 which corresponds to the penalty function \rm{QO} in \eqref{pen:QO}.\\
\end{proof}

\section{Experiments and Results}\label{sec:results}

\begin{figure}[!t]
  \centering
  \subfigure[Original]{\includegraphics[width=0.328\linewidth]{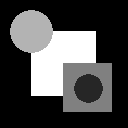}}\hfill%
  \subfigure[Noisy (14.19)]{\includegraphics[width=0.328\linewidth]{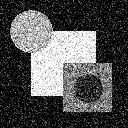}}\hfill%
  \subfigure[TViso (15.68)]{\includegraphics[width=0.328\linewidth]{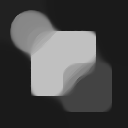}}\\
  \subfigure[HO (20.30)]{\includegraphics[width=0.328\linewidth]{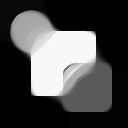}}\hfill%
  \subfigure[HD (20.36)]{\includegraphics[width=0.328\linewidth]{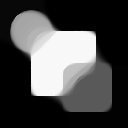}}\hfill%
  \subfigure[QO (22.60)]{\includegraphics[width=0.328\linewidth]{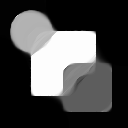}}\\
  \subfigure[QD (22.47)]{\includegraphics[width=0.328\linewidth]{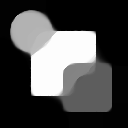}}\hfill%
  \subfigure[SO (21.43)]{\includegraphics[width=0.328\linewidth]{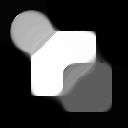}}\hfill%
  \subfigure[SD (23.16)]{\includegraphics[width=0.328\linewidth]{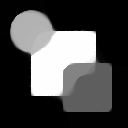}}\\
  \setcounter{subfigure}{0}%
  \caption{\label{ex:TViso-cc} Comparison of alternative refitting approaches with the proposed SD model on a synthetic image.}
\end{figure}

\begin{figure}[!t]
  \centering
  \subfigure[Original]{\includegraphics[width=0.32\linewidth]{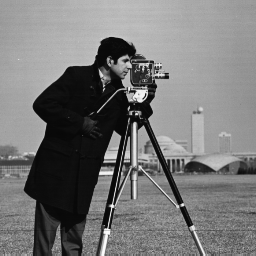}}\hfill%
  \subfigure[Noisy (22.14)]{\includegraphics[width=0.32\linewidth]{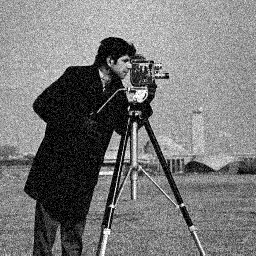}}\hfill%
  \subfigure[TViso (26.00)]{\includegraphics[width=0.32\linewidth]{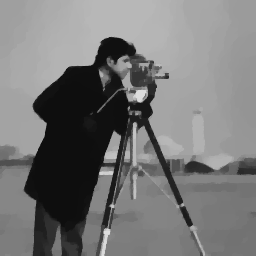}}\\
  \subfigure[HO (26.35)]{\includegraphics[width=0.32\linewidth]{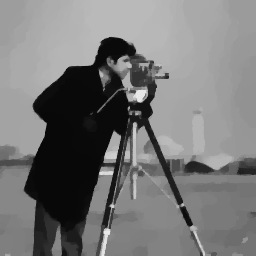}}\hfill%
  \subfigure[HD (26.34)]{\includegraphics[width=0.32\linewidth]{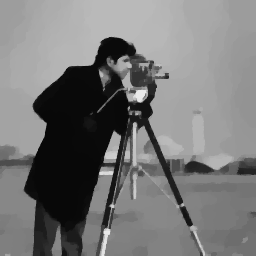}}\hfill%
  \subfigure[QO (28.21)]{\includegraphics[width=0.32\linewidth]{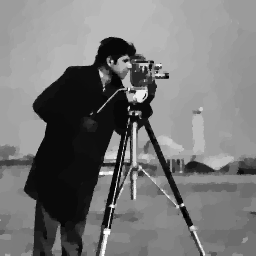}}\\
  \subfigure[QD (28.21)]{\includegraphics[width=0.32\linewidth]{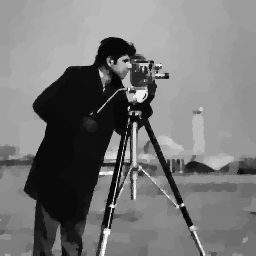}}\hfill%
  \subfigure[SO (28.22)]{\includegraphics[width=0.32\linewidth]{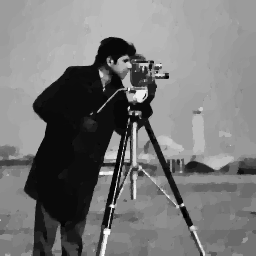}}\hfill%
  \subfigure[SD (28.40)]{\includegraphics[width=0.32\linewidth]{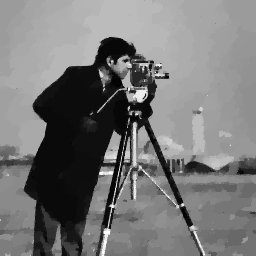}}\\
  \setcounter{subfigure}{0}%
  \caption{\label{ex:TViso-cameraman} Comparison of alternative refitting approaches with the proposed SD model on the Cameraman image.}
\end{figure}

\subsection{Toy experiments with TViso}

We first consider TViso regularization of grayscale images
on denoising problems of the form $y = x + w$ where $w$ is an additive white Gaussian
noise with standard deviation $\sigma$. We start with a simple toy example
of a $128 \times 128$ image%
\footnote{In this paper, we always consider images whose values are in the range $[0, 255]$.}
composed of elementary geometric shapes, and we chose $\sigma = 50$.
The corresponding image and its noisy version are given in Figs.~
\ref{ex:TViso-cc}.(a) and (b) respectively.
To highlight the different behaviors between the six refitting block penalties,
we chose to set the regularisation parameter of TViso to a large value
$\lambda = 750$ leading to a strong bias in the solution. Fig.~\ref{ex:TViso-cc}.(c) depicts
this solution in which not only some structures are lost (the
darkest disk), but a large loss of contrast can be observed (the white large
square became gray). Unlike boosting, the purpose
of refitting is not to recover the lost structures, but only to recover
the correct amplitudes of the reconstructed objects without changing their
geometrical aspect. In a second scenario, we consider the standard
{\it cameraman} image of size $256 \times 256$, and we set
$\sigma = 20$ and $\lambda = 36$.
The corresponding images are given in Figs.~\ref{ex:TViso-cameraman}.(a-c).

We applied the iterative primal-dual algorithm with our joint-refitting
(Algorithm \eqref{algo_final}) for $4,000$ iterations,
with $\tau = 1/\norm{\nabla}_2$, $\kappa = 1/\norm{\nabla}_2$ and $\theta = 1$,
and where $\norm{\nabla}_2 = 2 \sqrt{2}$.
Results of refitting with HO, HD, QO, QD, SO and SD are given for the two
scenarios on Figs.~\ref{ex:TViso-cc}.(d-i) and Figs.~\ref{ex:TViso-cameraman}.(d-i),
respectively.
Numerically, the  difference when changing the block penalty  comes from the computation of the proximal operator, presented in Table \ref{tab:prox_phi_dual},  that is only used in the fifth (resp. last) step of Algorithm \ref{algo_final} (resp. Algorithm \ref{eq:algo}). As a consequence, the overall computational cost is stable for all considered block penalties.

  \newcommand{\sizex}{2.2cm}
  \newcommand{\sidecapX}[1]{ {\begin{sideways}\parbox{\sizex}{\centering #1}\end{sideways}} }
 \begin{figure*}[ht!]
 \begin{tabular}{p{.2cm}ccccc}
 && \includegraphics[width=\sizex]{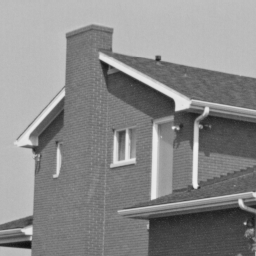}&  \includegraphics[width=\sizex]{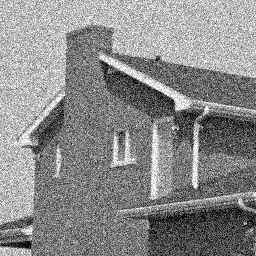} &  \includegraphics[width=\sizex]{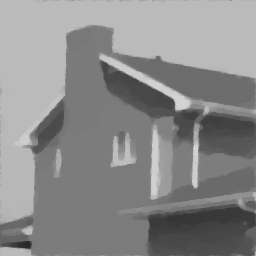}  \\
 &&  Original&Noisy (14.19) &TViso (27.41)\\
 \hline
&\multicolumn{5}{c}{Refitting with SD model}\\
&PSNR/iterations&Sequential&Sequential&Joint&Joint\\
&&(500 it.)&(1000 it.)&(500 it.)&(1000 it.)\\
\sidecapX{$\beta=10^{-4}$}&\includegraphics[height=\sizex,]{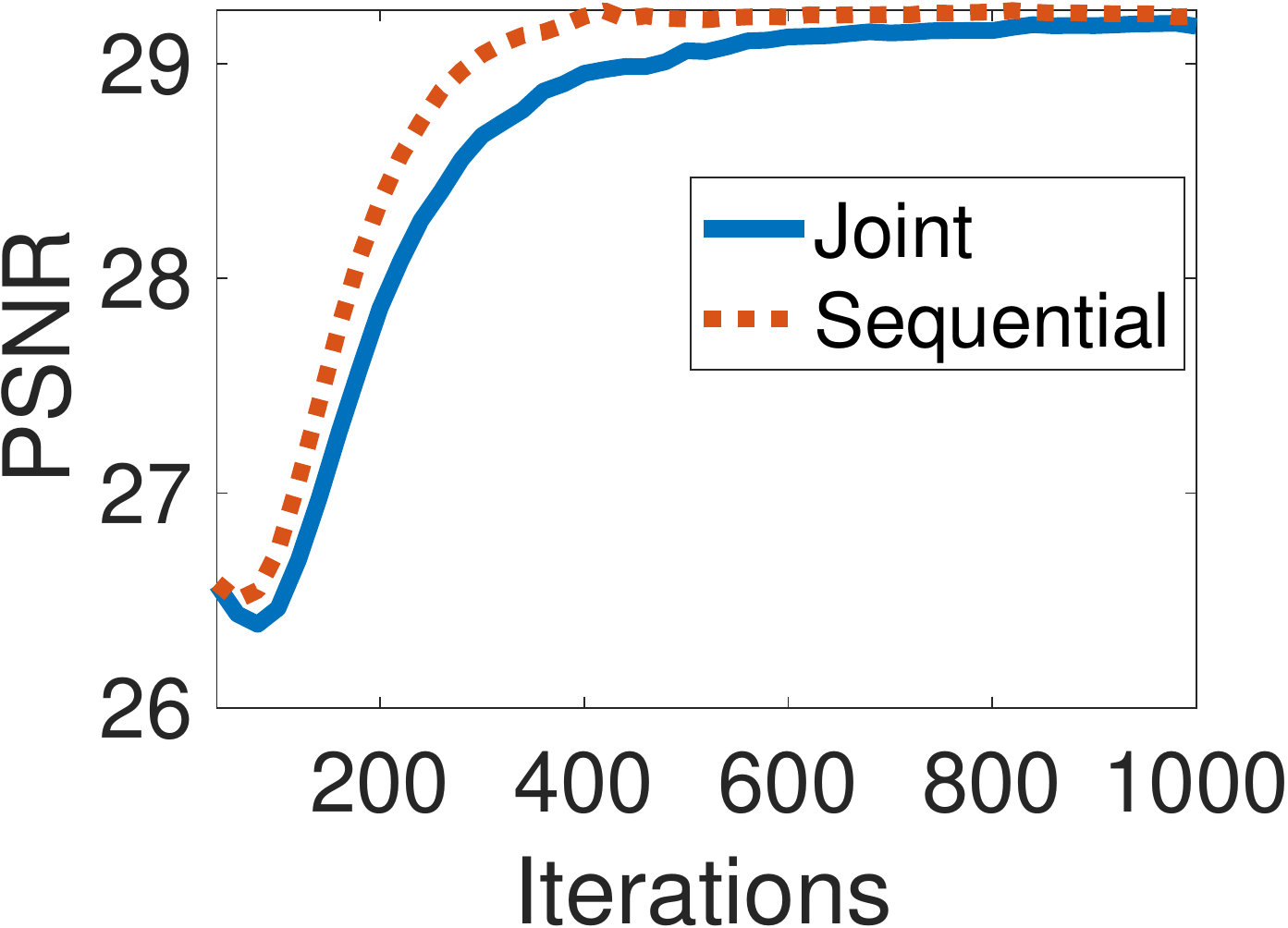}&\includegraphics[width=\sizex,viewport=40 30 180 170,clip]{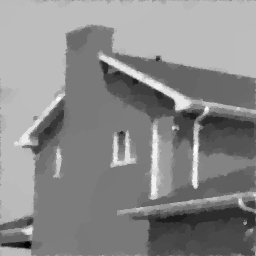}&\includegraphics[width=\sizex,viewport=40 30 180 170,clip]{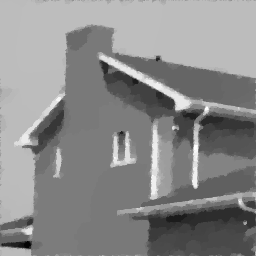}&\includegraphics[width=\sizex,viewport=40 30 180 170,clip]{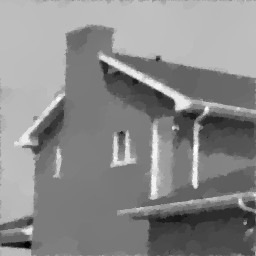}&\includegraphics[width=\sizex,viewport=40 30 180 170,clip]{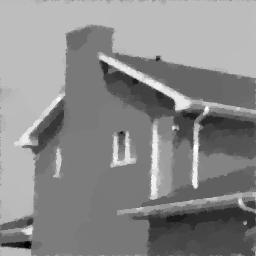}\\
&&(29.21)&(29.22)& (29.06)& (29.17)\\
\sidecapX{$\beta=10^{-8}$}&\includegraphics[height=\sizex]{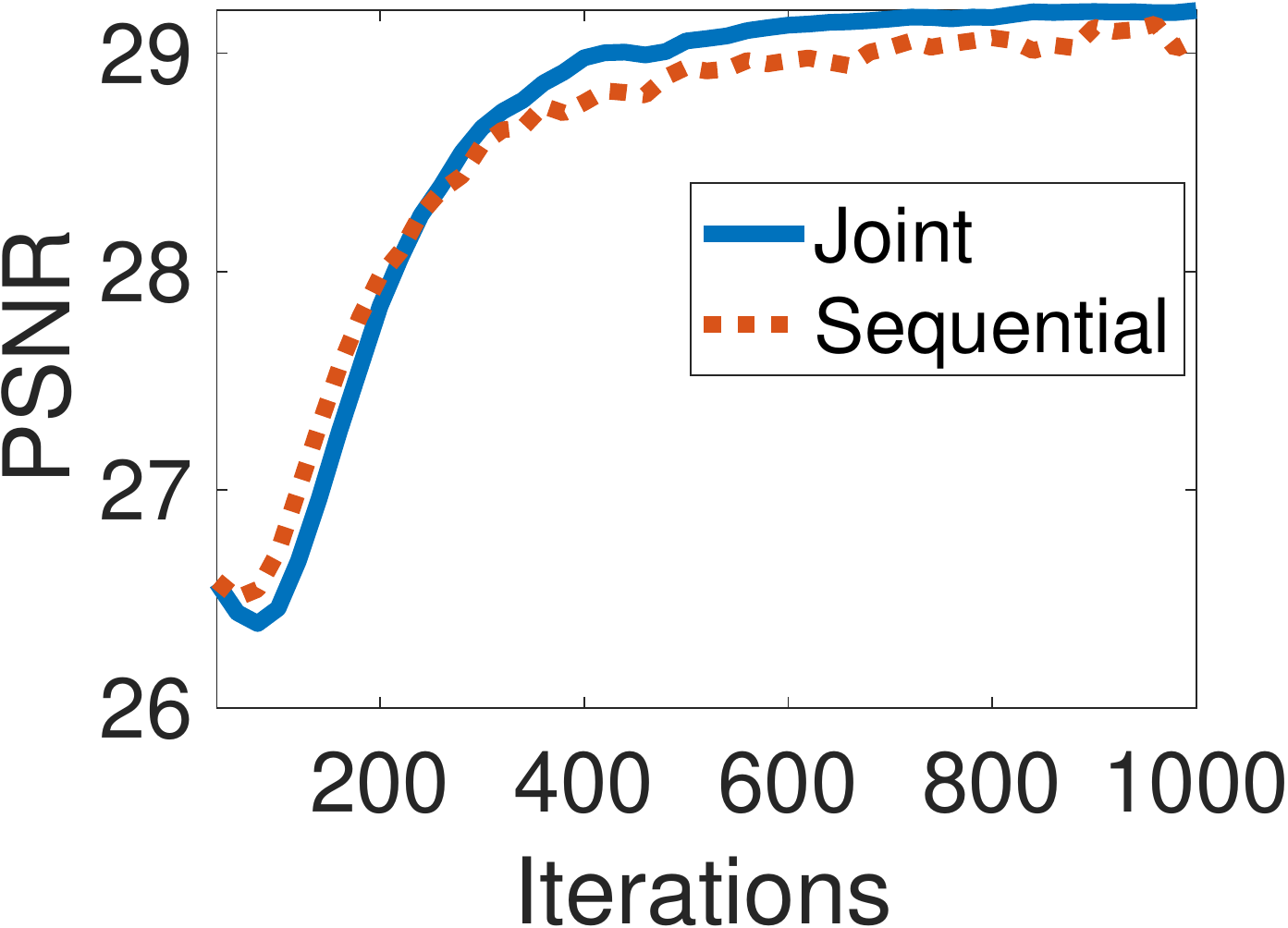}&\includegraphics[width=\sizex,viewport=40 30 180 170,clip]{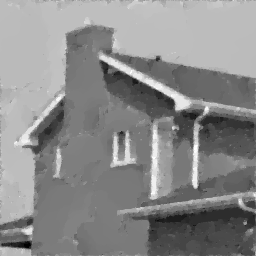}&\includegraphics[width=\sizex,viewport=40 30 180 170,clip]{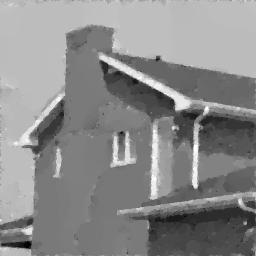}&\includegraphics[width=\sizex,viewport=40 30 180 170,clip]{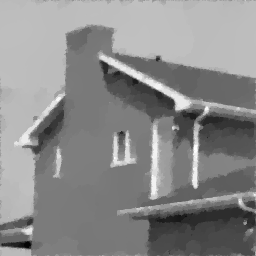}&\includegraphics[width=\sizex,viewport=40 30 180 170,clip]{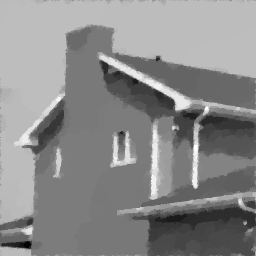}\\
&&(28.94)&(29.07)& (29.06)& (29.20)\\
 \end{tabular}
 \caption{\label{fig:post_joint} Influence of the $\beta$ parameter (support detection) on the sequential and joint approaches. The joint refitting algorithm is more robust  to the choice of the value $\beta$. }

 \end{figure*}

As a quantitave measure of refitting between an estimate $\hat{x}$ and
the underlying image $x$, we consider the Peak Signal to Noise Ratio (PSNR) defined in this case as
$
  {\rm PSNR}
  =
  -10 \log_{10} \tfrac1n \norm{x - \hat{x}}_2^2$.
The higher the PSNR, the better the refitting is supposed to be.
In Fig.~\ref{ex:TViso-cc}, we first observe that block penalties that are
only based on orientation (HO, QO and SO) lead to refitted solution with
contrast inversions compared to the original solution of TViso
(see the banana dark shapes at the interface between the two squares).
In that scenario, HD, QD and SD provides statisfying visual quality,
the lattest achieving highest PSNR. In the context of Fig.~\ref{ex:TViso-cameraman},
the image being more complex, the TViso solution presents more level lines
than in the previous scenario. The block penalties HO and HD must preserve exactly
these level lines. The refitted problem becomes too constrained and this lack of flexibility
prevents HO and HD to deviate much from TViso. Their refitted solution
remains significantly biased. Comparing now the quadratic penalties (QO and QD) against the soft
ones (SO and SD) reveals that quadratic penalties does not allow to
reenhance as much the contrast of some objects (see, \eg the second tower on the right of the image).
Again SD achieved the highest PSNR in that 
scenario.

\subsection{Comparison between sequential and joint refitting}
The joint refitting algorithm is more stable with respect to the threshold $\beta$ considered for detecting the support:
\begin{align}\label{eq:support}
    \mathcal{I}^k=\{\|\Gamma\hat x^k\|_i>\beta\}\enspace.
\end{align}
We illustrate this property in Figure \ref{fig:post_joint}, where the Douglas-Rachford solver with the SD block penalty has been tested  both sequentially and jointly,  with the same iteration budget and the same parameter $\beta$ in \eqref{algo_final} and \eqref{eq:support}.  In this experiment, we consider a $256\times 256$ image with   $\sigma = 30$ and $\lambda=50$. For the parameters of the algorithm \ref{eq:algo} ,
we set  $\alpha=0.5$ and $\tau = 0.01$.

  When considering a large threshold $\beta=10^{-4}$ to detect the support of the solution, both  approaches leads to visually fine results. The support of the TViso solution is nevertheless underestimated. On the other hand, with  a low threshold $\beta=10^{-8}$, the detection of the support is not stable with the sequential approach. It leads to noisy refitted solutions, whereas the joint approach still gives  clean refittings.  This stable behaviour is also underlined by looking at the PSNR values with respect to the number of iterations: the sequential approach (dotted curve)  presents fluctuating PSNR values with a low (and thus accurate)  threshold.

\subsection{Experiments with color TViso}

We now consider TViso regularization of degraded color images with three channels
Red, Green, Blue (RGB).
We defined blocks obtained by applying
$\Gamma=(\nabla^1_R,\nabla^2_R,\nabla^1_G,\nabla^2_G,\nabla^1_B,\nabla^2_B)$,
where $m=n$, $b=6$, and $\nabla_C^d$ denotes the forward discrete
gradient in the direction $d \in \{ 1, 2 \}$
for the color channel $C \in \{ R, G, B \}$.

We first focused on a denoising problem $y = x + w$ where
$x$ is a color image
and $w$ is an additive white Gaussian
noise with standard deviation $\sigma=20$.
We next focused on a deblurring problem $y = \Phi x + w$ where
$x$ is a color image, $\Phi$ is a convolution
simulating a directional blur, and  $w$ is an additive white Gaussian
noise with standard deviation $\sigma=2$.
In both cases, we chose $\lambda=4.3 \sigma$.
We apply the iterative primal-dual algorithm with our joint-refitting
(Algorithm \eqref{algo_final}) for $1,000$ iterations,
with $\theta = 1$, $\tau = 1/\norm{\Gamma}_2$ and $\kappa = 1/\norm{\Gamma}_2$, where, as for the grayscale case, $\norm{\Gamma}_2 = 2 \sqrt{2}$.
Note that in order to use the primal-dual algorithm, we have to implement
$\Gamma^t$, defined for $z = (z_R, z_G, z_B) \in \RR^{n \times 6}$ where
$z_R$, $z_G$, $z_B$ are 2d vector fields, as
given by
\begin{align}
  \Gamma^t z=
  \begin{pmatrix}
    -\diverg z_R && -\diverg z_G && -\diverg z_B\\
  \end{pmatrix} ~.
\end{align}

Results are provided on Fig.~\ref{fig:denoising} and \ref{fig:deblurring}.
Comparisons of refitting with our proposed SD block penalty,
HO, HD, QO, QD and SO are provided.
Using our proposed SD block penalty offers the best refitting
performances in terms of both visual and quantitative measures.
The loss of contrast of TViso is well-corrected,
amplitudes are enhanced while smoothness and sharpness of TViso is
preserved. Meanwhile, the approach does not  create artifacts,
invert contrasts, or reintroduce information
that were not recovered by TViso.

\begin{figure*}[!t]
  \subfigure[Original]{%
    \begin{minipage}{.313\linewidth}
    \includegraphics[width=1\linewidth,viewport=0 90 322 462, clip]{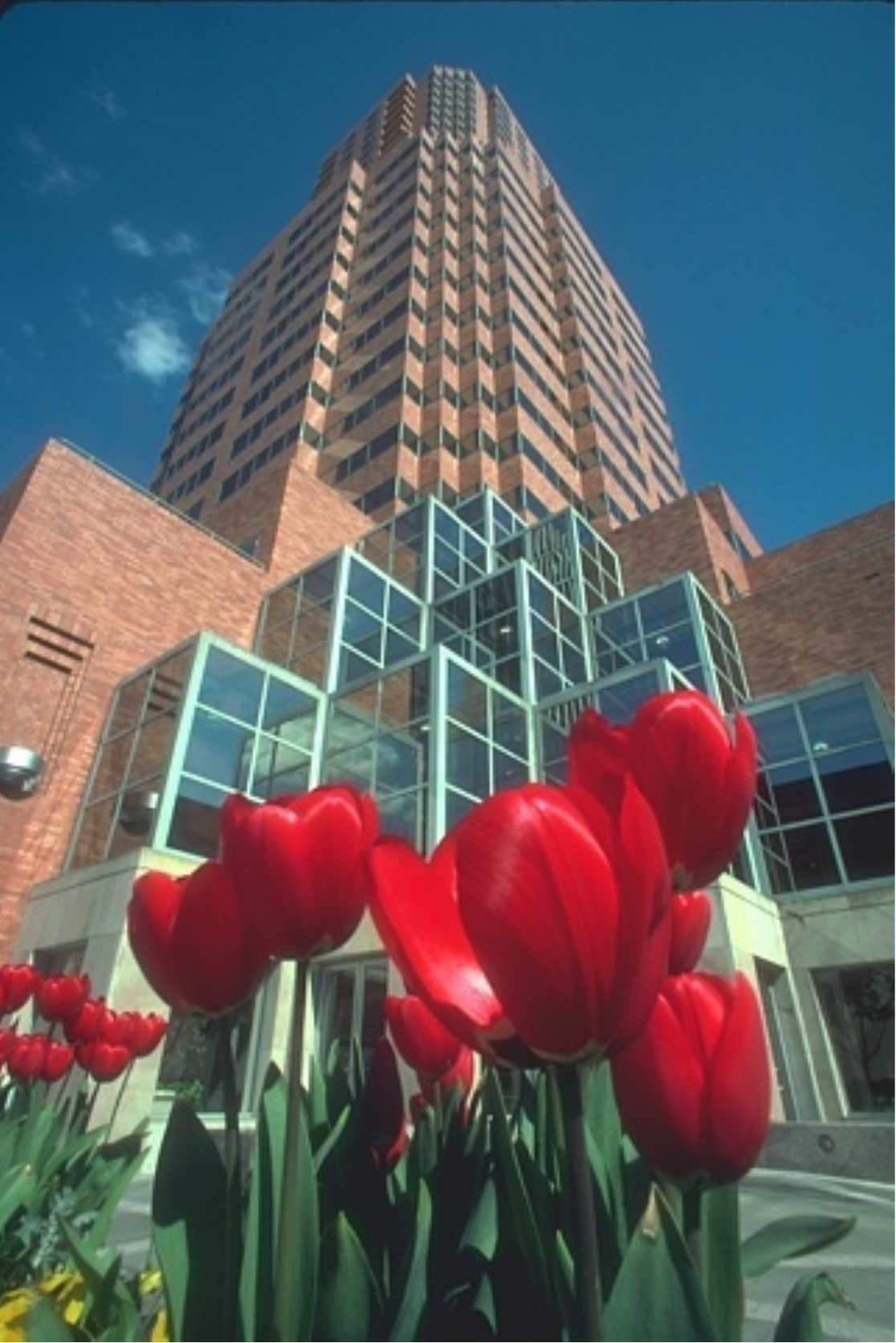}\\[-.8em]
    \end{minipage}%
  }\hfill
  \subfigure[Noisy {\scriptsize (22.10)}]{%
    \begin{minipage}{.313\linewidth}
    \includegraphics[width=1\linewidth,viewport=0 90 322 462, clip]{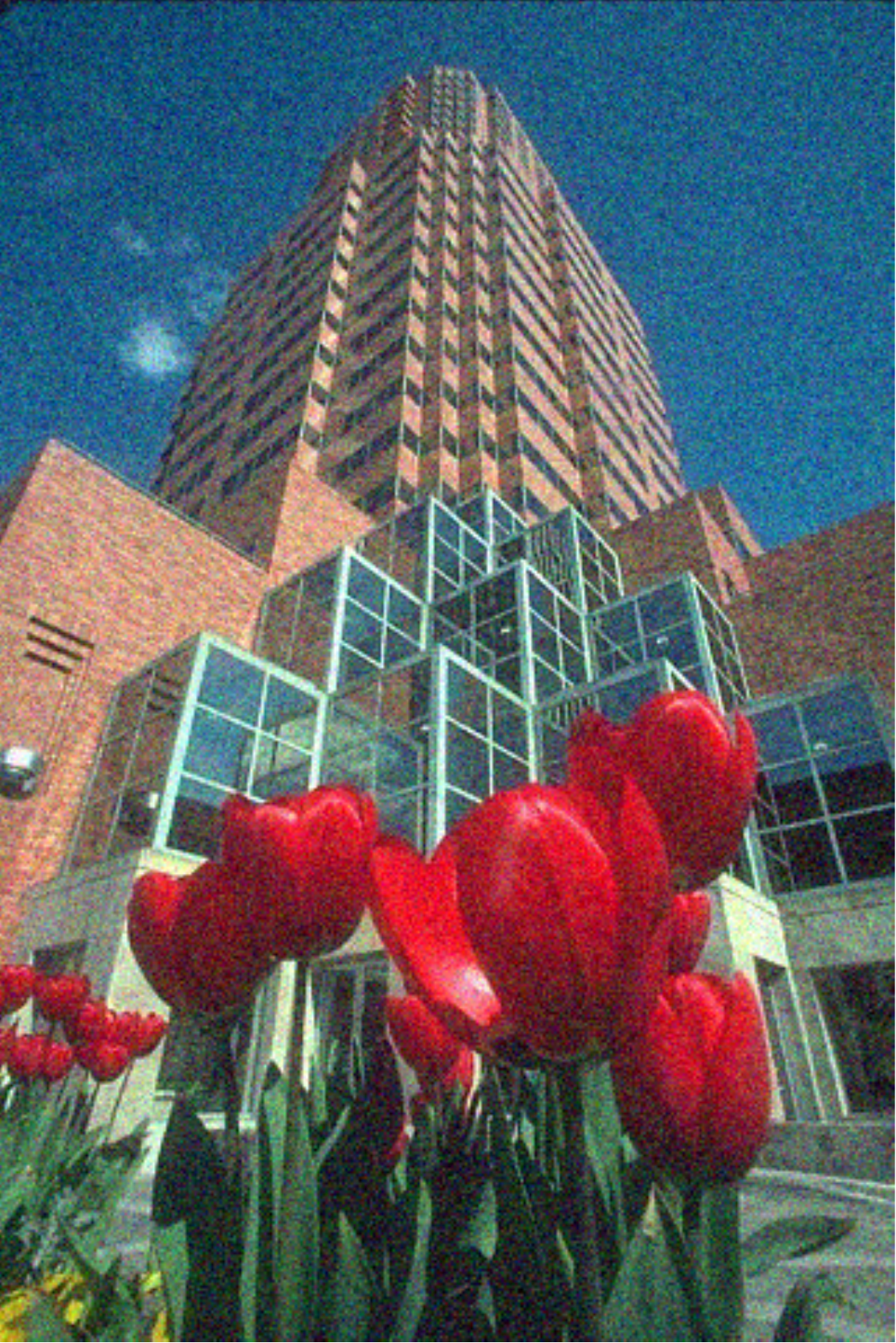}\\[-.8em]
    \end{minipage}%
  }\hfill
  \subfigure[TViso {\scriptsize (23.28)}]{%
    \begin{minipage}{.313\linewidth}
      \includegraphics[width=1\linewidth,viewport=0 90 322 462, clip]{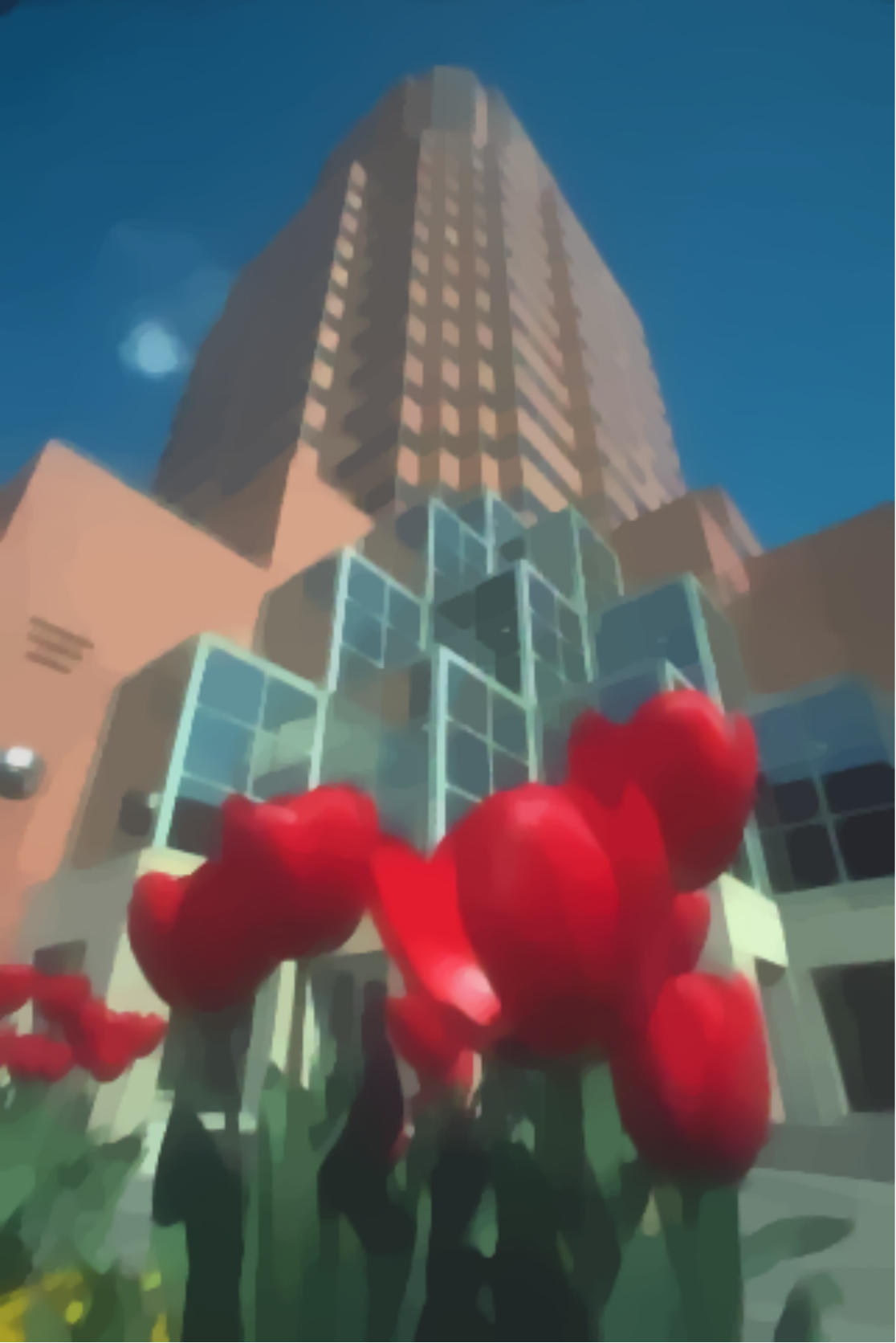}\\[-.8em]
    \end{minipage}%
  }\hfill
  \subfigure[HO {\scriptsize (23.75)}]{%
    \begin{minipage}{.313\linewidth}
      \includegraphics[width=1\linewidth,viewport=0 90 322 462, clip]{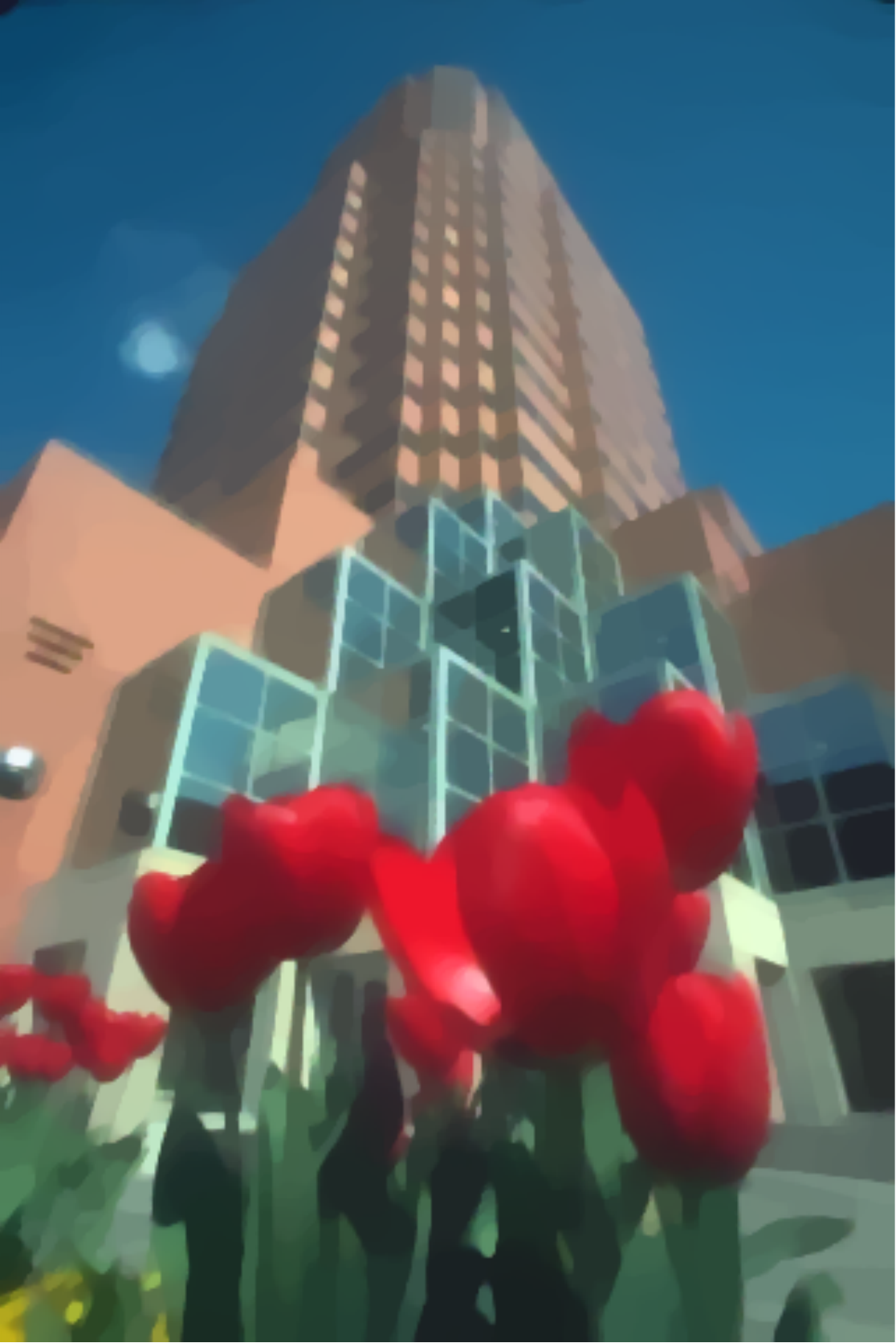}\\[-.8em]
    \end{minipage}%
  }\hfill
  \subfigure[HD {\scriptsize (23.75)}]{%
    \begin{minipage}{.313\linewidth}
      \includegraphics[width=1\linewidth,viewport=0 90 322 462, clip]{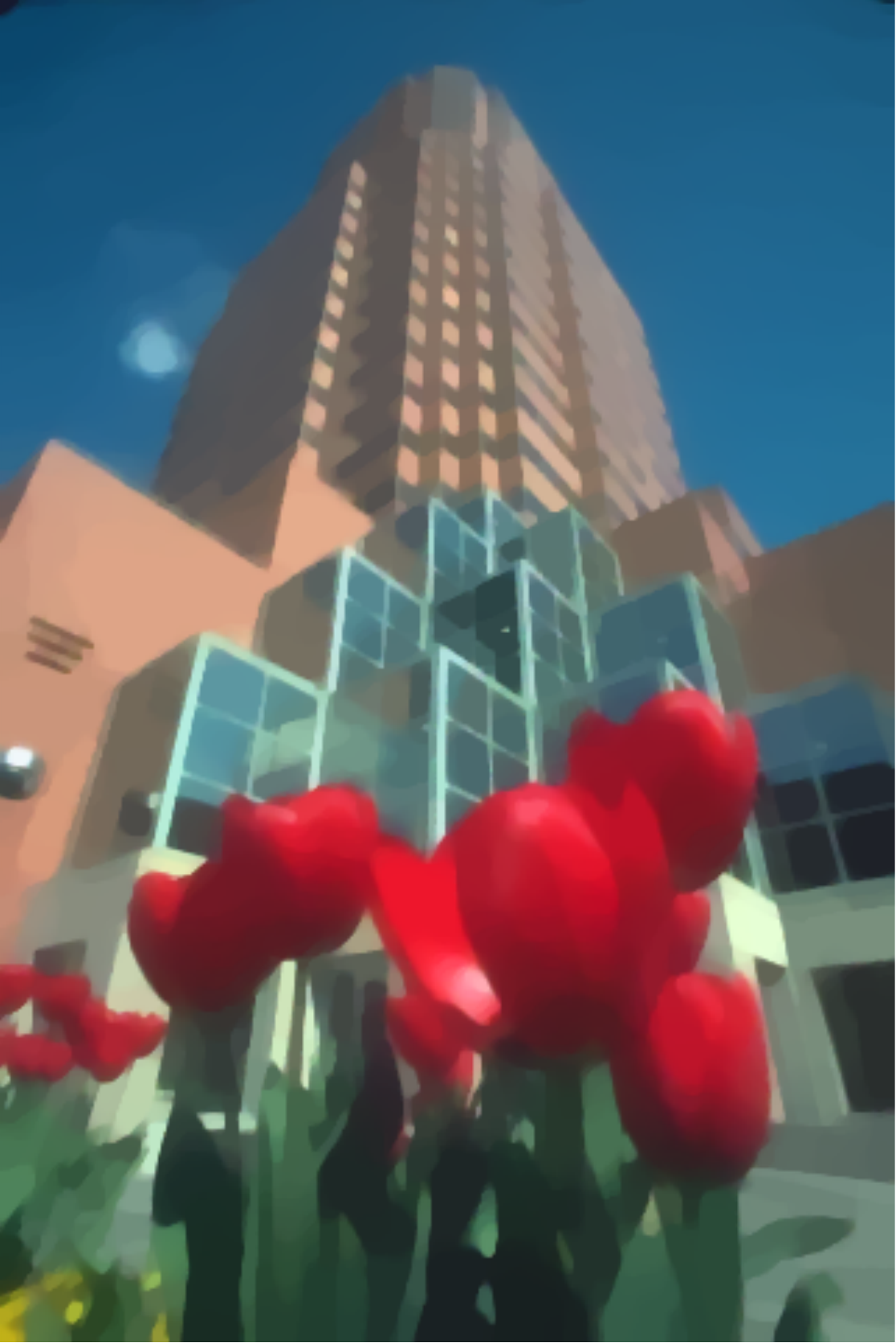}\\[-.8em]
    \end{minipage}%
  }\hfill
  \subfigure[QO {\scriptsize (26.12)}]{%
    \begin{minipage}{.313\linewidth}
      \includegraphics[width=1\linewidth,viewport=0 90 322 462, clip]{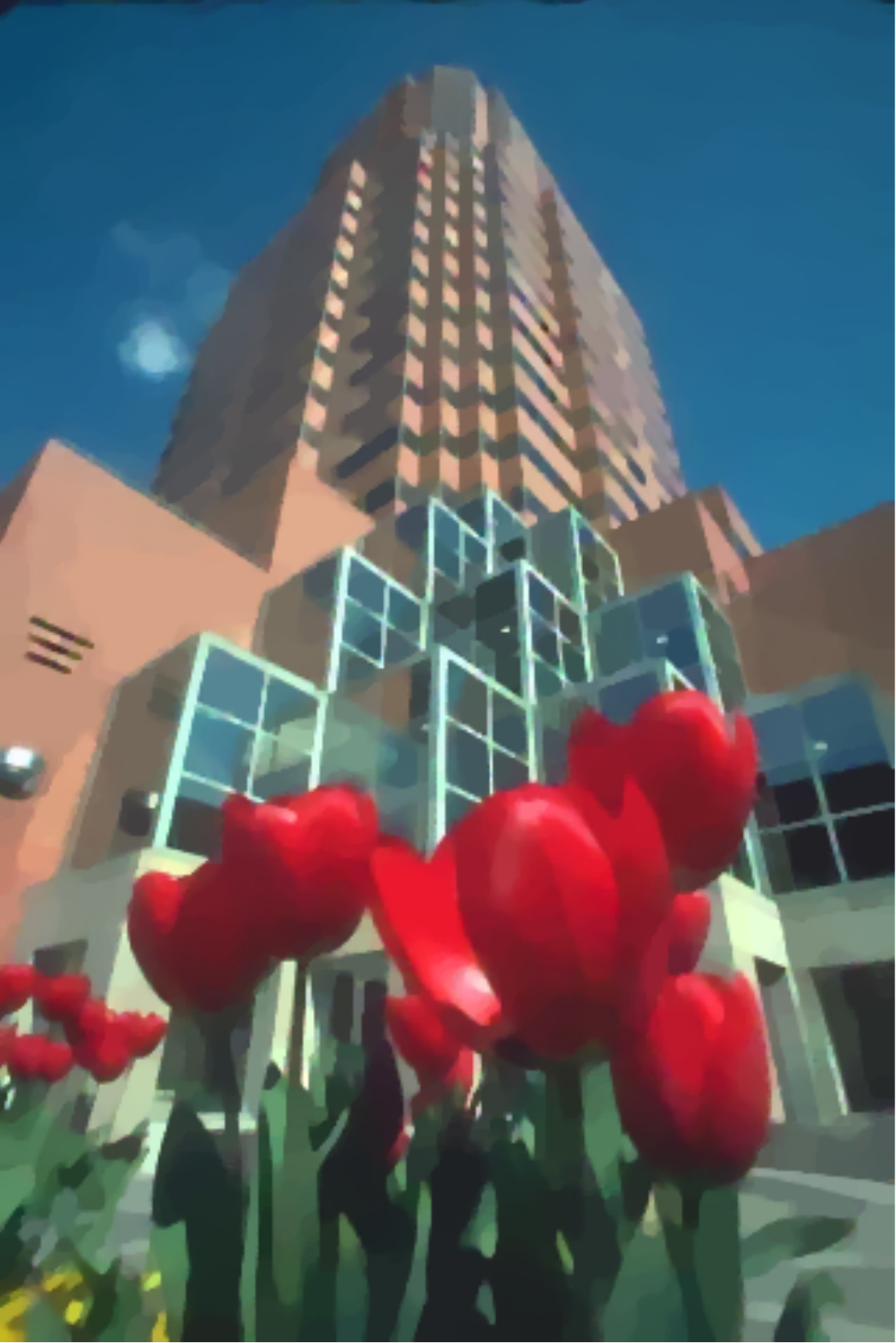}\\[-.8em]
    \end{minipage}%
  }\hfill
  \subfigure[QD {\scriptsize (26.10)}]{%
    \begin{minipage}{.313\linewidth}
      \includegraphics[width=1\linewidth,viewport=0 90 322 462, clip]{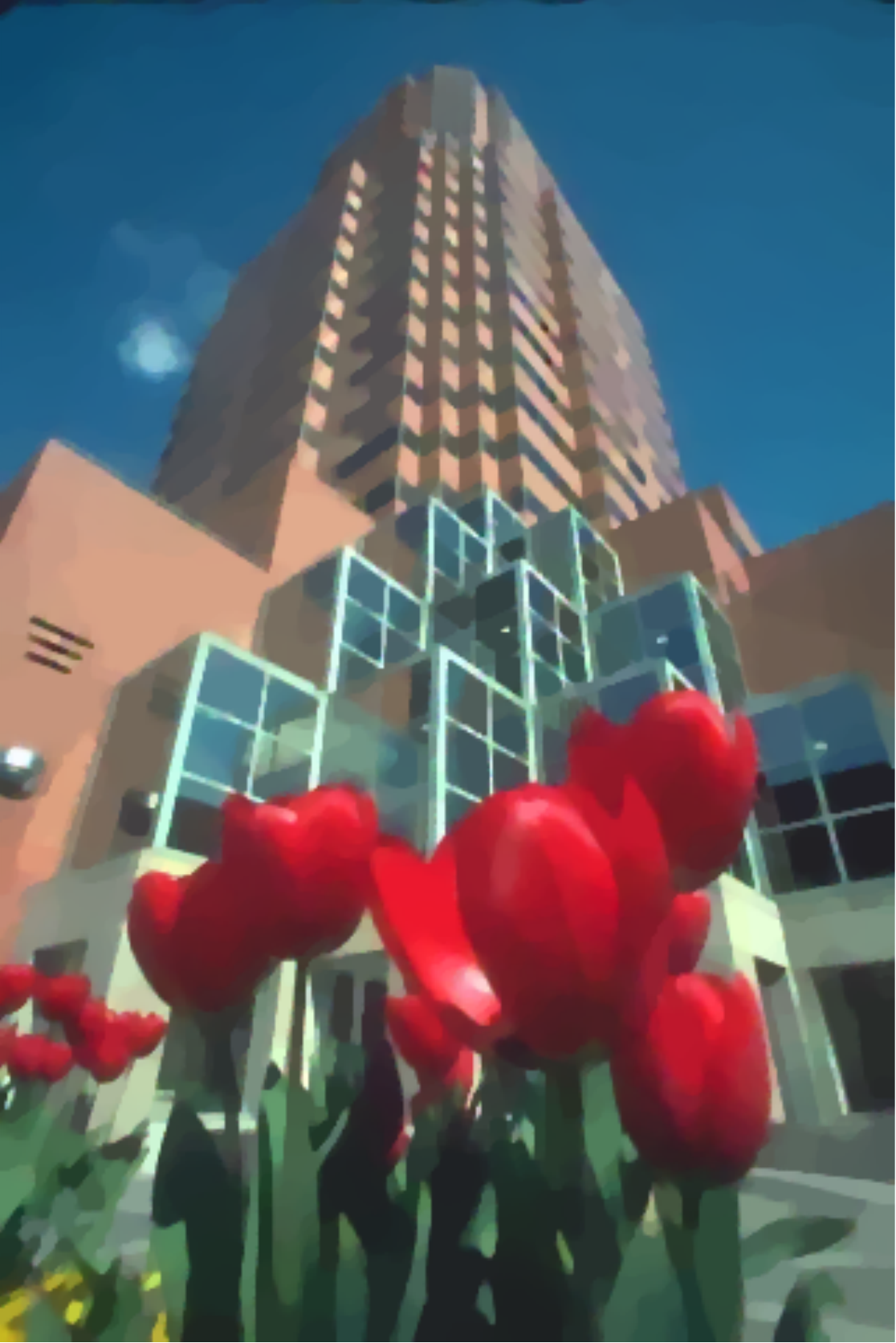}\\[-.8em]
    \end{minipage}%
  }\hfill
  \subfigure[SO {\scriptsize (26.15)}]{%
    \begin{minipage}{.313\linewidth}
      \includegraphics[width=1\linewidth,viewport=0 90 322 462, clip]{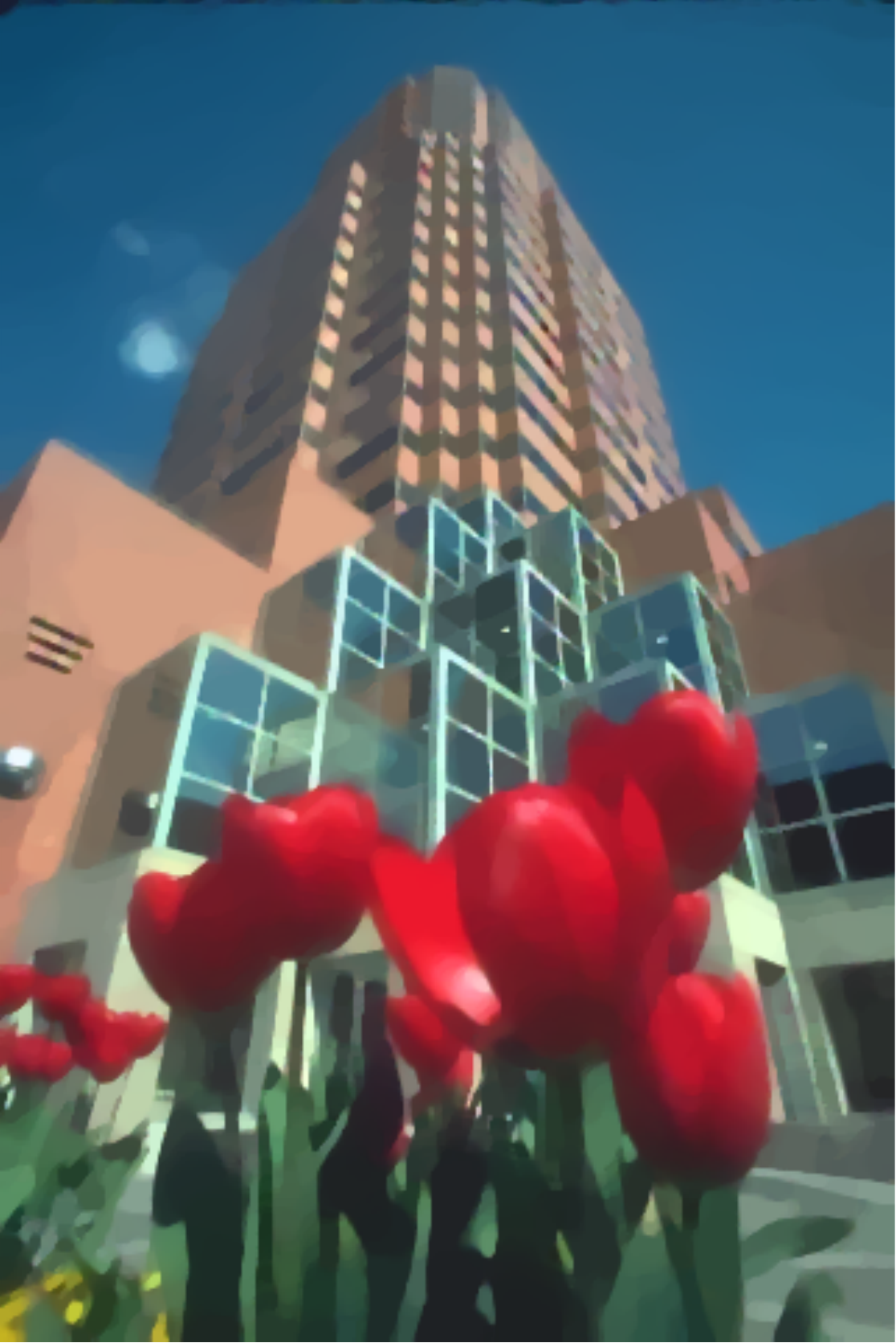}\\[-.8em]
    \end{minipage}%
  }\hfill
  \subfigure[SD~{\scriptsize (27.68)}]{%
    \begin{minipage}{.313\linewidth}
      \includegraphics[width=1\linewidth,viewport=0 90 322 462, clip]{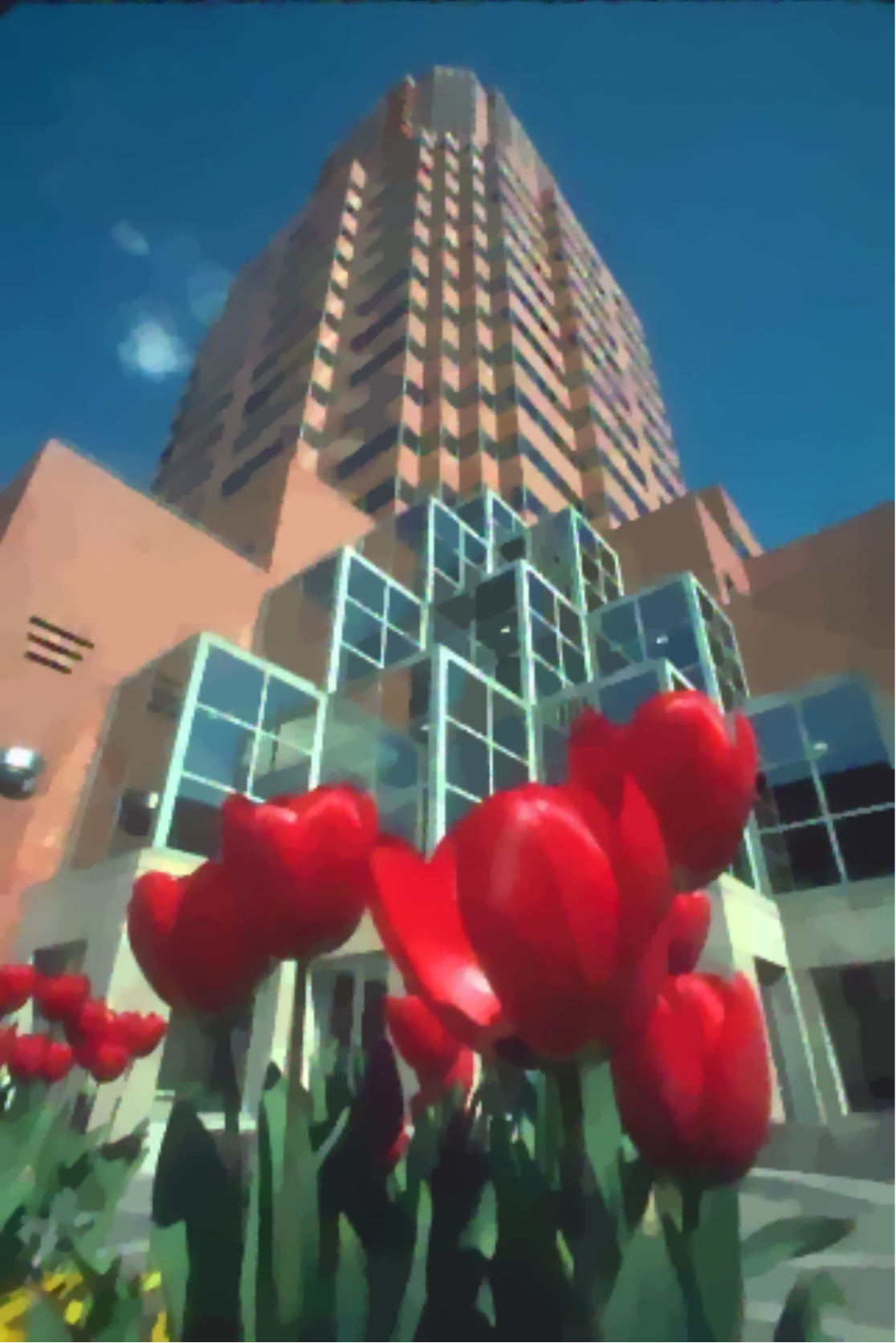}\\[-.8em]
    \end{minipage}%
  }\\

  \caption{
    (a) A color image. (b) A corrupted version by Gaussian noise with standard deviation $\sigma=20$.
    (b) Solution of TViso.
    Debiased solution with (d) HO, (e) HD, (f) QO, (g) QD, (h) SO and (i) SD.
    The Peak Signal to Noise Ratio (PSNR) is indicated in brackets bellow each image.}
  \label{fig:denoising}
\end{figure*}

\begin{figure*}[!t]
  \subfigure[Original]{%
    \begin{minipage}{.322\linewidth}
    \includegraphics[width=1\linewidth]{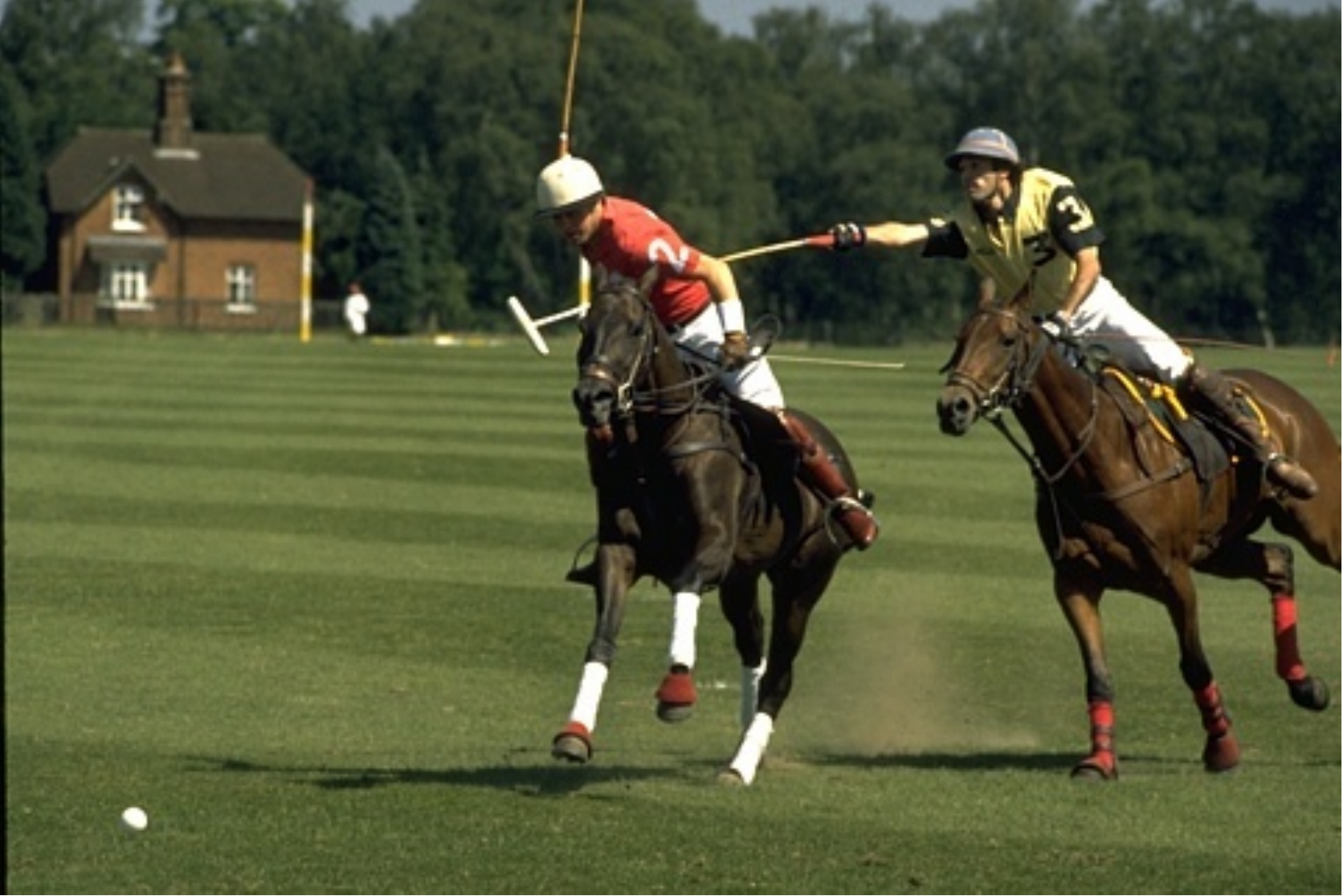}\\[.05em]
      \includegraphics[width=.495\linewidth,viewport=20 200 135 255,clip]{jockey_blur_ori}\hfill%
      \includegraphics[width=.495\linewidth,viewport=270 160 385 215,clip]{jockey_blur_ori}\\[-.7em]
    \end{minipage}%
  }\hfill
  \subfigure[Blurry {\scriptsize (23.14)}]{%
    \begin{minipage}{.322\linewidth}
    \includegraphics[width=1\linewidth]{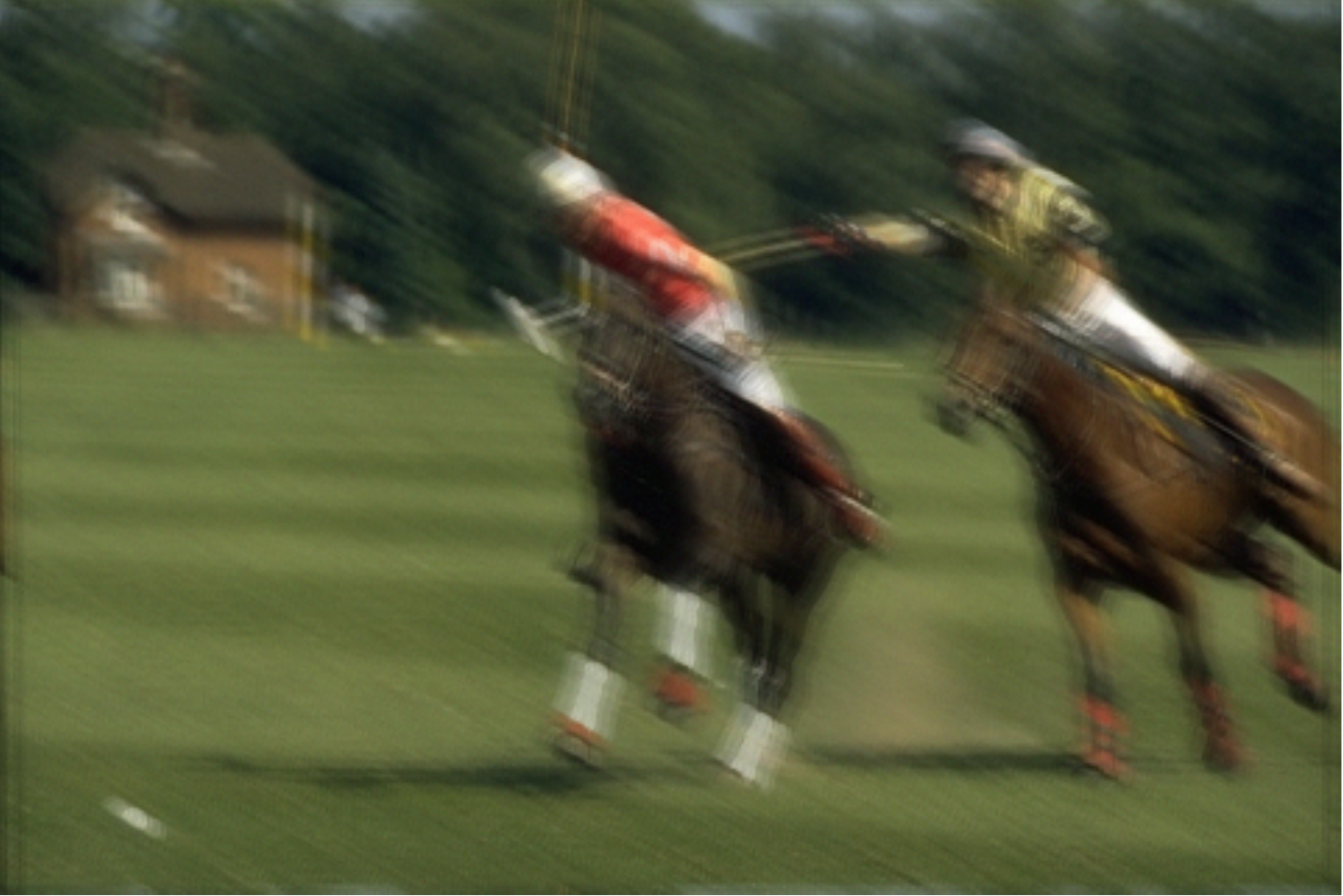}\\[.05em]
      \includegraphics[width=.495\linewidth,viewport=20 200 135 255,clip]{jockey_blur_subimg1}\hfill%
      \includegraphics[width=.495\linewidth,viewport=270 160 385 215,clip]{jockey_blur_subimg1}\\[-.7em]
    \end{minipage}%
  }\hfill
  \subfigure[TViso {\scriptsize (26.95)}]{%
    \begin{minipage}{.322\linewidth}
    \includegraphics[width=1\linewidth]{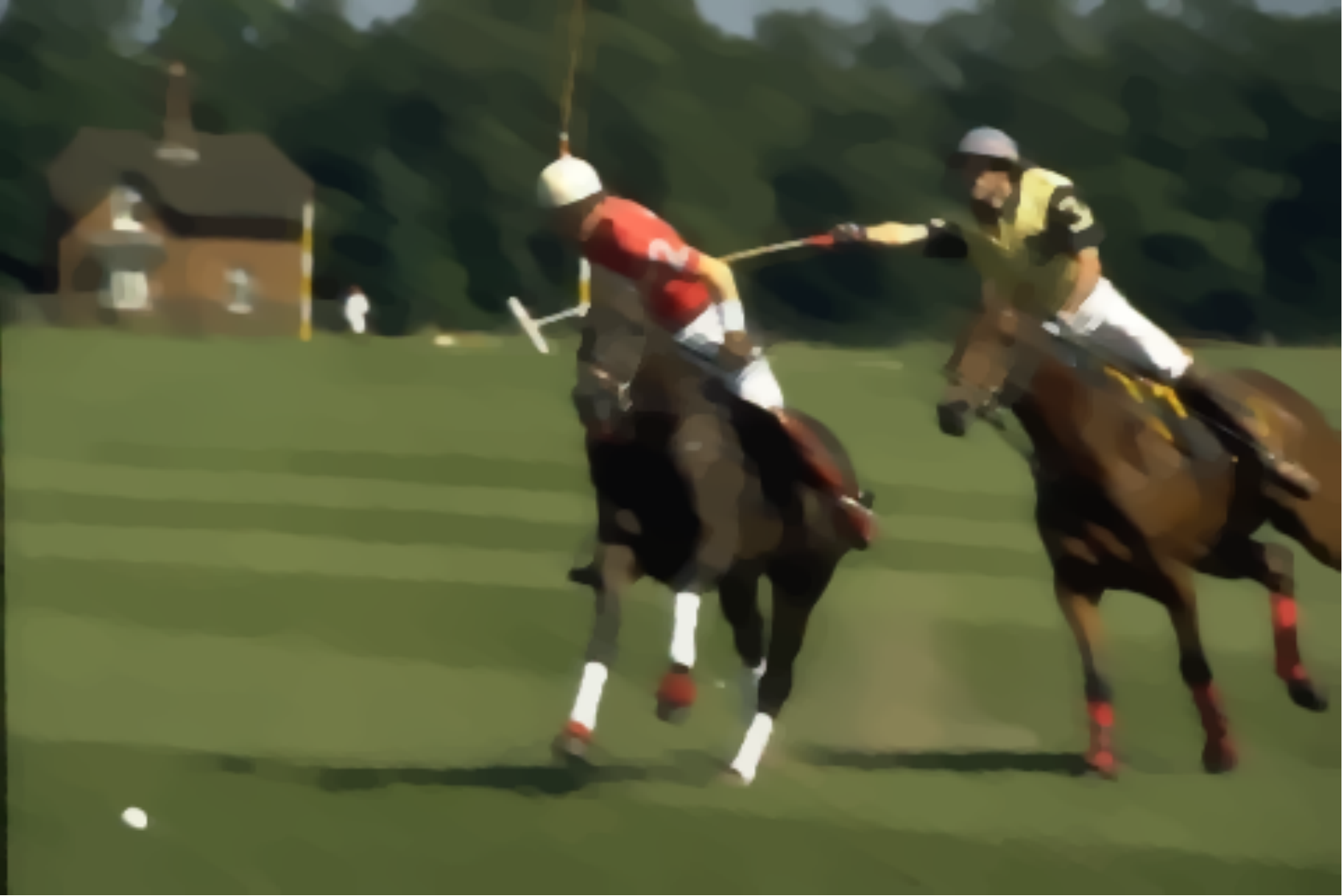}\\[.05em]
      \includegraphics[width=.495\linewidth,viewport=20 200 135 255,clip]{jockey_blur_subimg2}\hfill%
      \includegraphics[width=.495\linewidth,viewport=270 160 385 215,clip]{jockey_blur_subimg2}\\[-.7em]
    \end{minipage}%
  }\\
  \subfigure[HO {\scriptsize (27.09)}]{%
    \begin{minipage}{.322\linewidth}
      \includegraphics[width=1\linewidth]{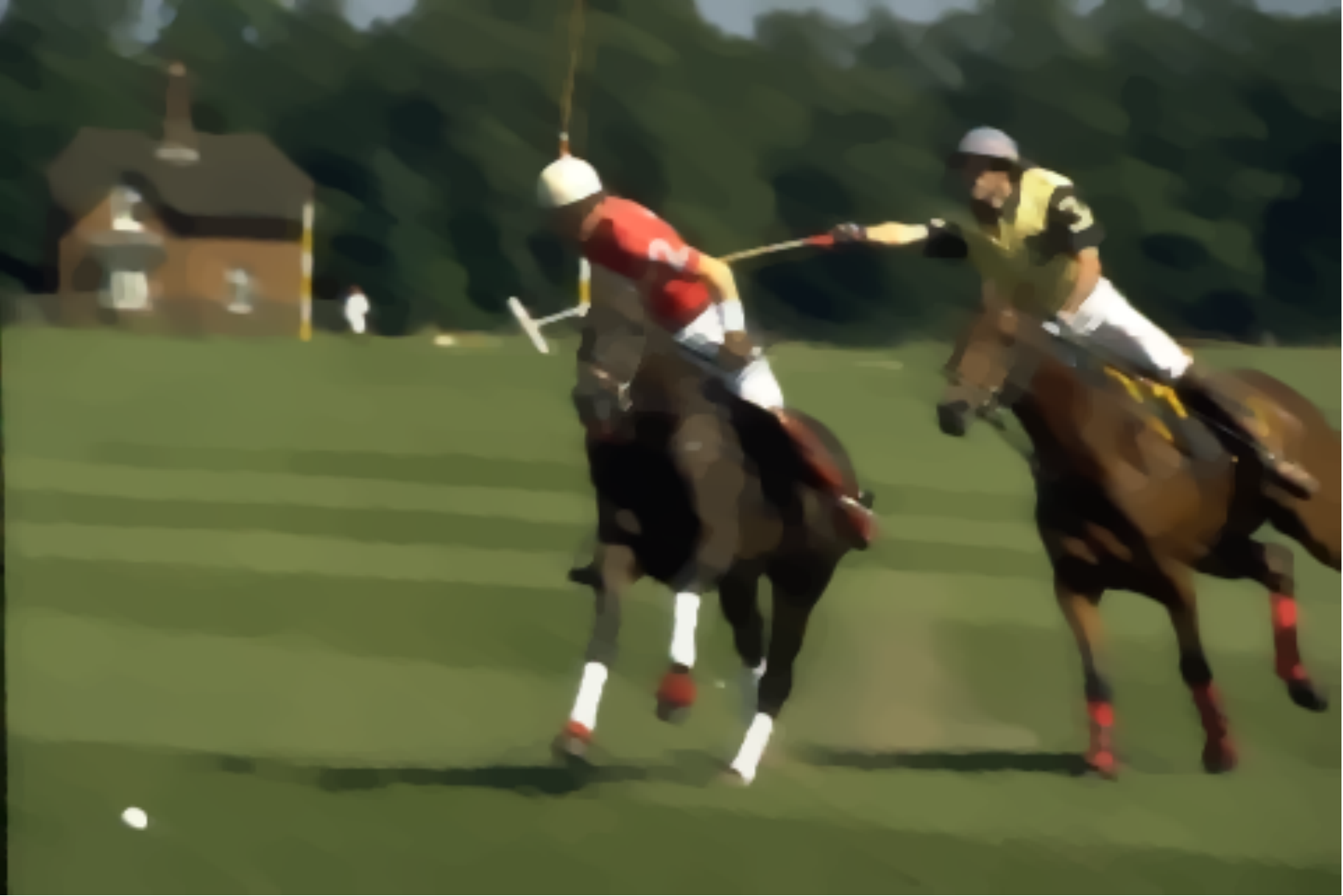}\\[.05em]
      \includegraphics[width=.495\linewidth,viewport=20 200 135 255,clip]{jockey_blur_subimg3}\hfill%
      \includegraphics[width=.495\linewidth,viewport=270 160 385 215,clip]{jockey_blur_subimg3}\\[-.7em]
    \end{minipage}%
  }\hfill
  \subfigure[HD {\scriptsize (27.09)}]{%
    \begin{minipage}{.322\linewidth}
      \includegraphics[width=1\linewidth]{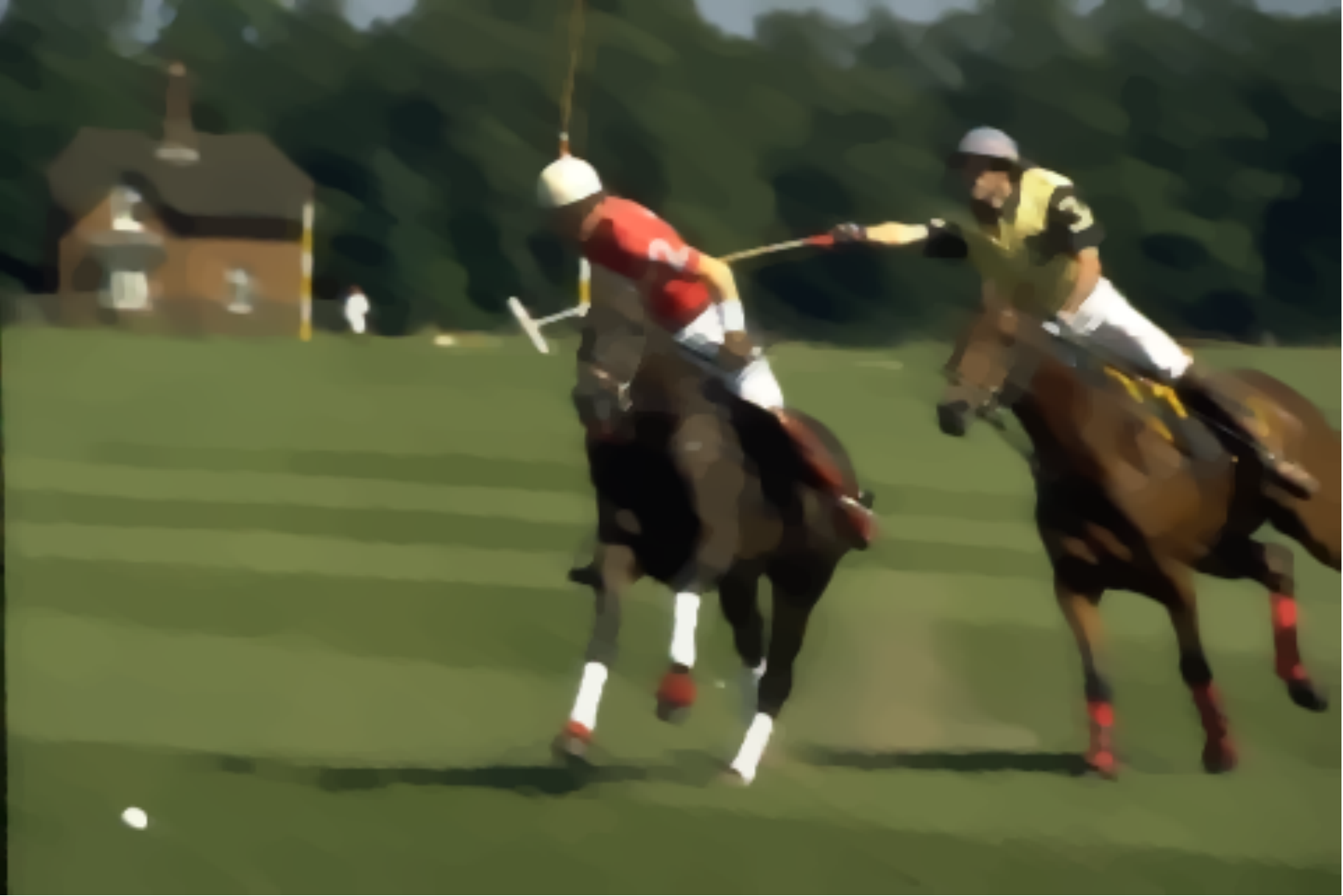}\\[.05em]
      \includegraphics[width=.495\linewidth,viewport=20 200 135 255,clip]{jockey_blur_subimg4}\hfill%
      \includegraphics[width=.495\linewidth,viewport=270 160 385 215,clip]{jockey_blur_subimg4}\\[-.7em]
    \end{minipage}%
  }\hfill
  \subfigure[QO {\scriptsize (29.38)}]{%
    \begin{minipage}{.322\linewidth}
      \includegraphics[width=1\linewidth]{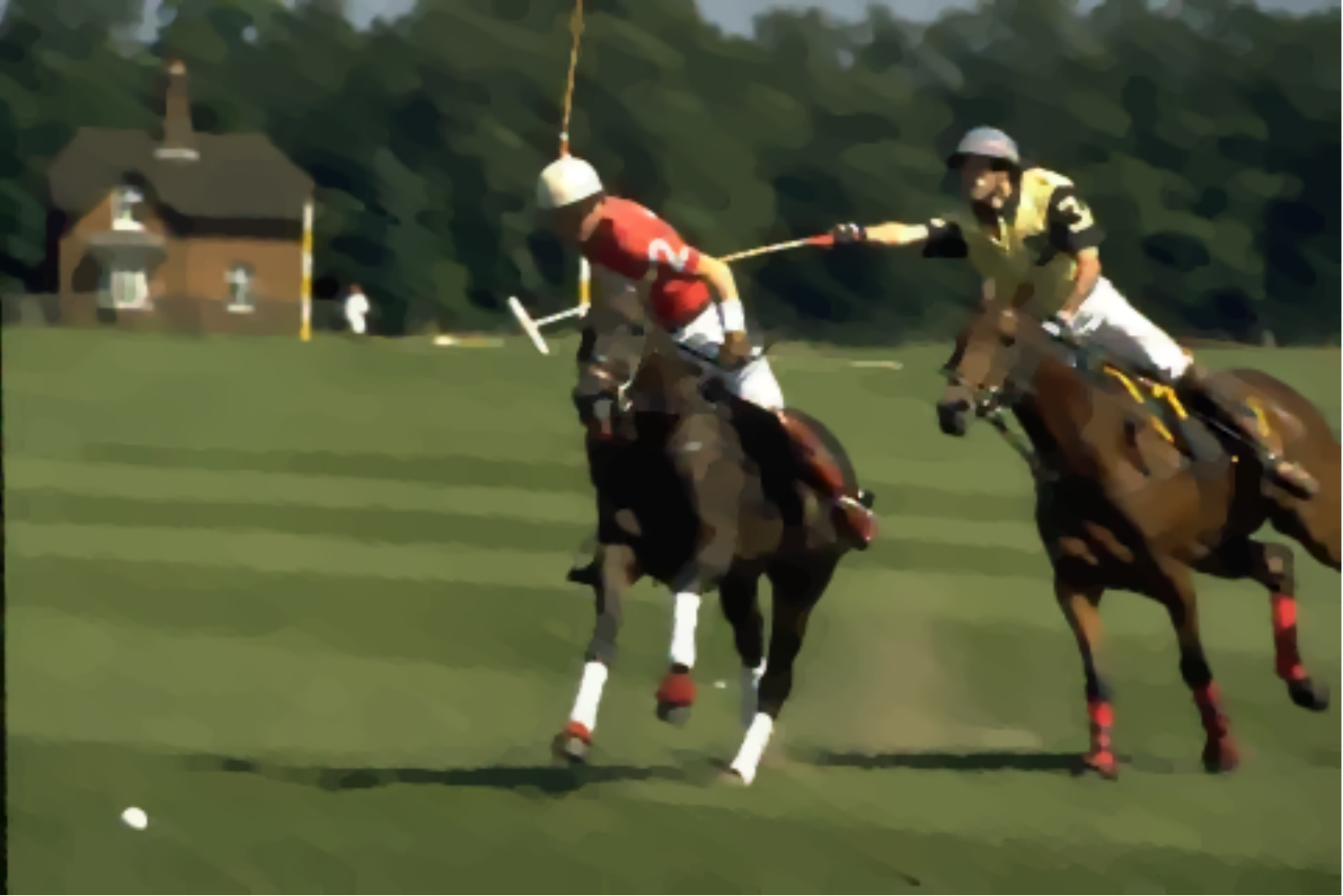}\\[.05em]
      \includegraphics[width=.495\linewidth,viewport=20 200 135 255,clip]{jockey_blur_subimg5}\hfill%
      \includegraphics[width=.495\linewidth,viewport=270 160 385 215,clip]{jockey_blur_subimg5}\\[-.7em]
    \end{minipage}%
  }\\
  \subfigure[QD {\scriptsize (29.31)}]{%
    \begin{minipage}{.322\linewidth}
      \includegraphics[width=1\linewidth]{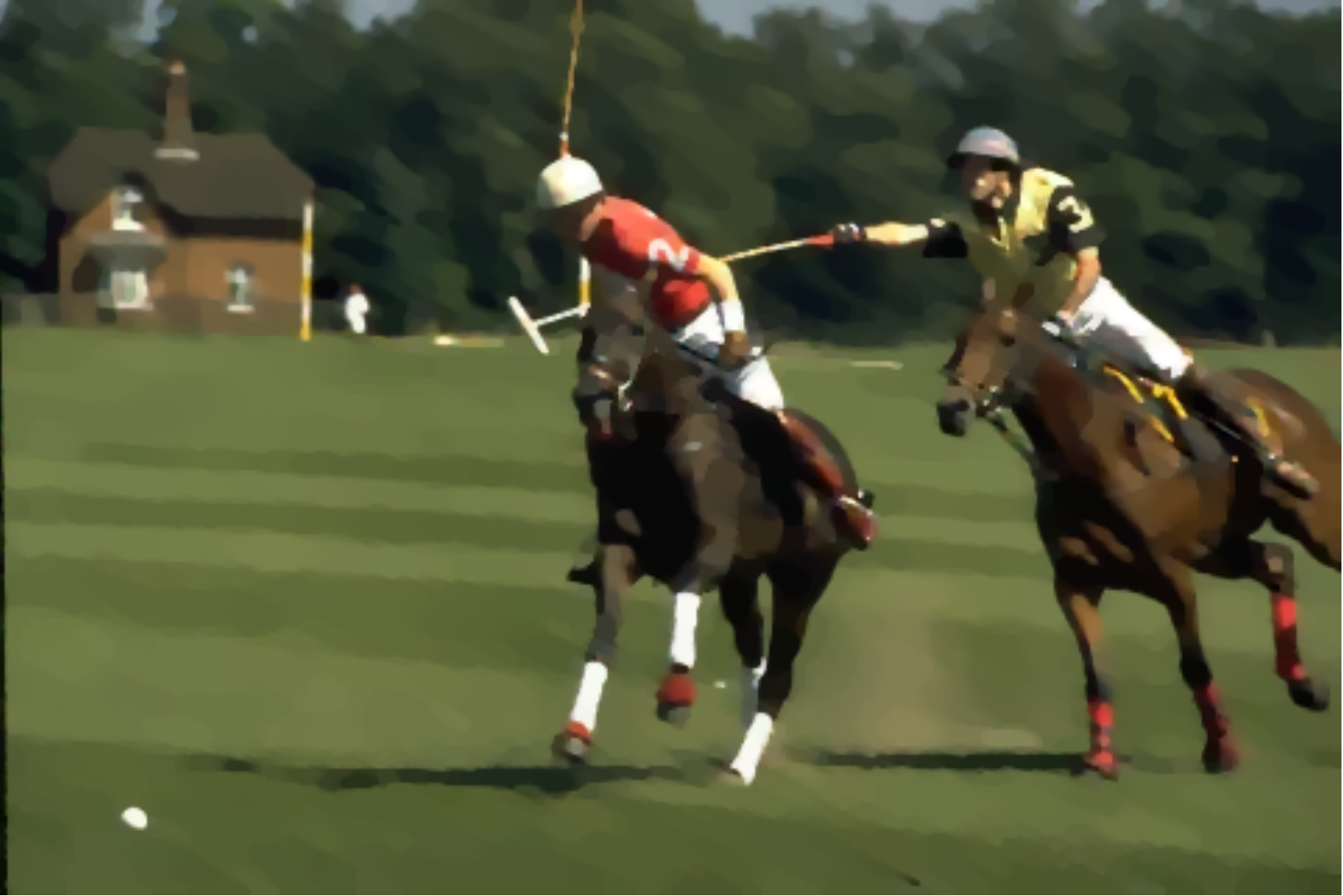}\\[.05em]
      \includegraphics[width=.495\linewidth,viewport=20 200 135 255,clip]{jockey_blur_subimg6}\hfill%
      \includegraphics[width=.495\linewidth,viewport=270 160 385 215,clip]{jockey_blur_subimg6}\\[-.7em]
    \end{minipage}%
  }\hfill
  \subfigure[SO {\scriptsize (29.34)}]{%
    \begin{minipage}{.322\linewidth}
      \includegraphics[width=1\linewidth]{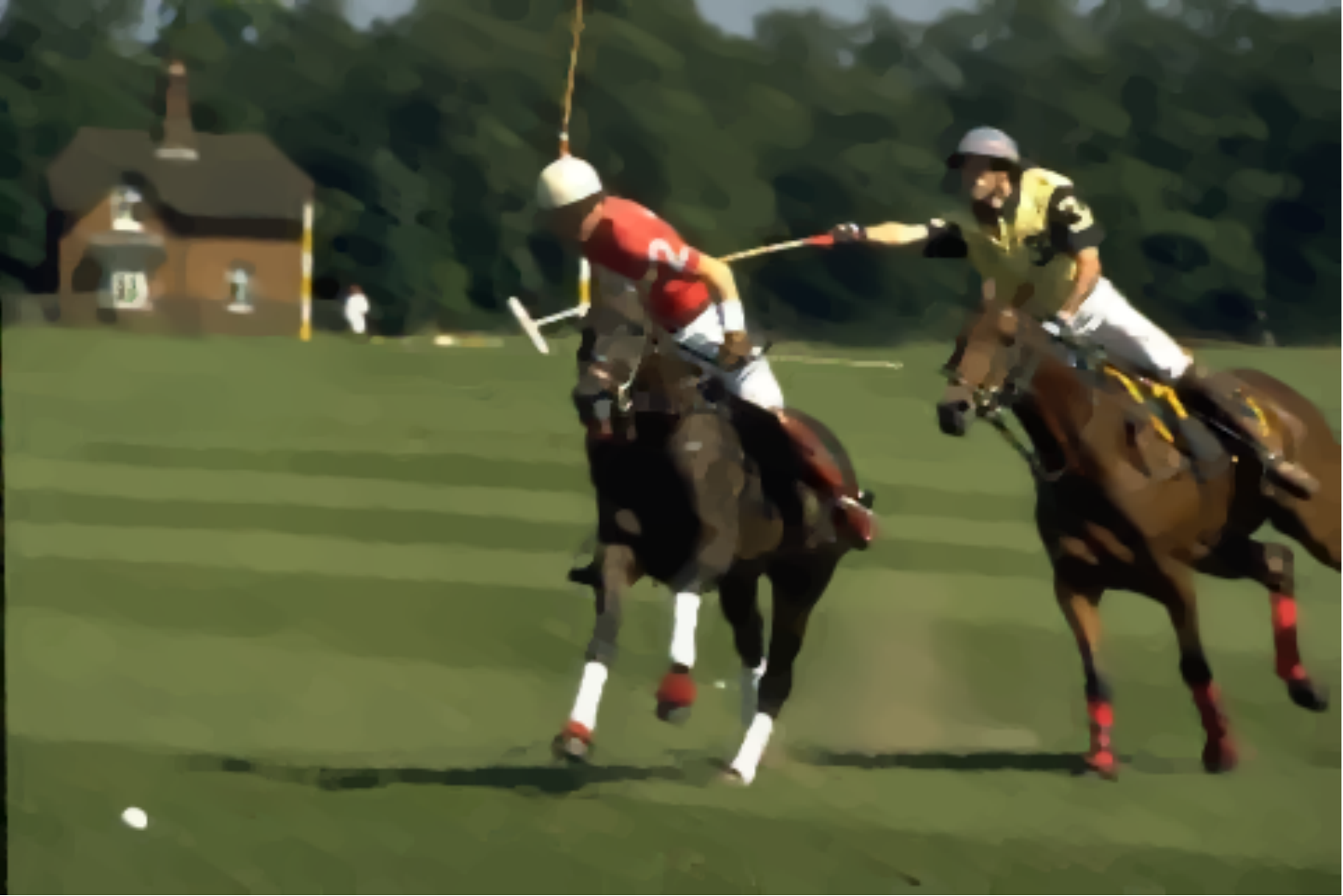}\\[.05em]
      \includegraphics[width=.495\linewidth,viewport=20 200 135 255,clip]{jockey_blur_subimg7}\hfill%
      \includegraphics[width=.495\linewidth,viewport=270 160 385 215,clip]{jockey_blur_subimg7}\\[-.7em]
    \end{minipage}%
  }\hfill
  \subfigure[SD {\scriptsize (30.18)}]{%
    \begin{minipage}{.322\linewidth}
      \includegraphics[width=1\linewidth]{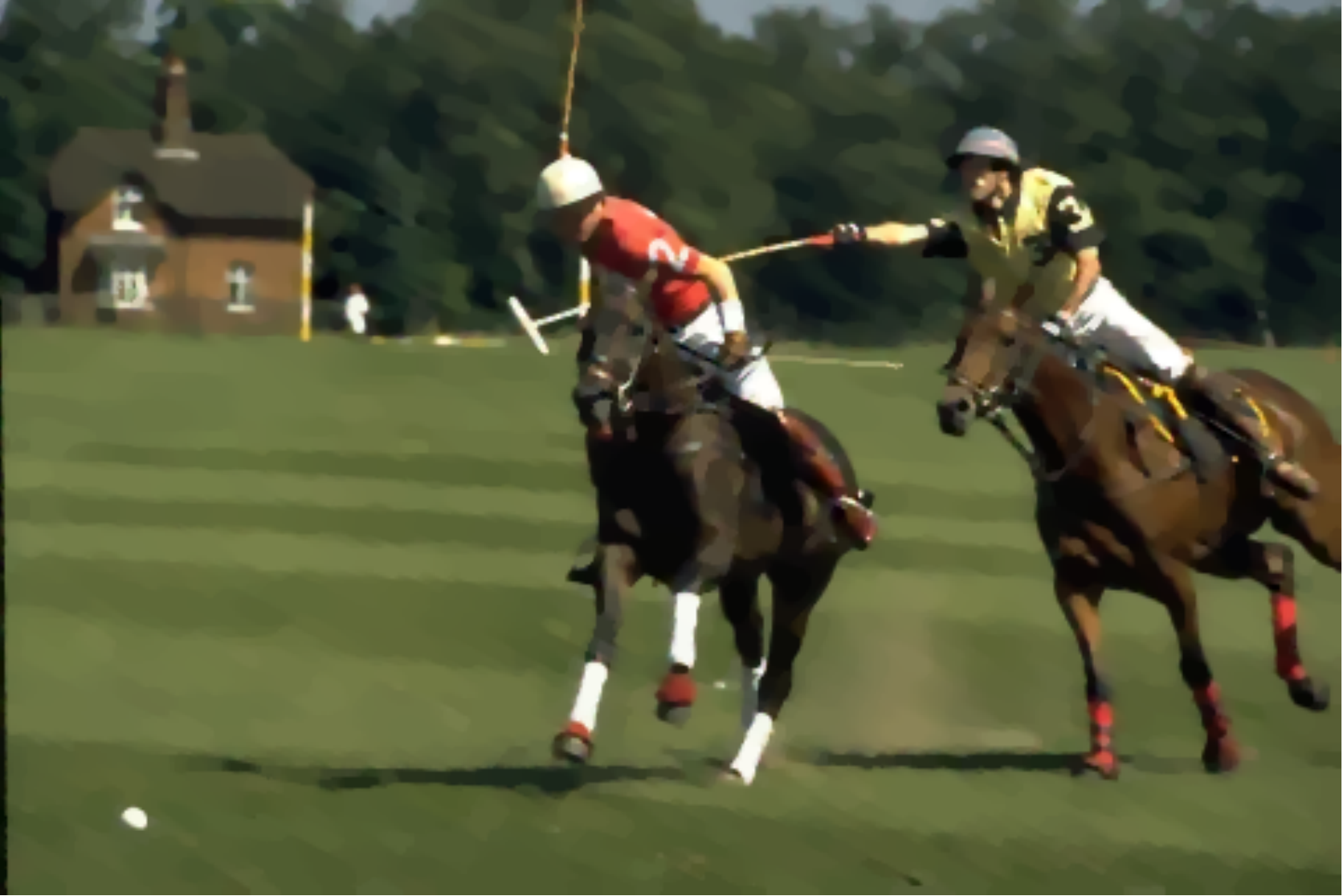}\\[.05em]
      \includegraphics[width=.495\linewidth,viewport=20 200 135 255,clip]{jockey_blur_subimg8}\hfill%
      \includegraphics[width=.495\linewidth,viewport=270 160 385 215,clip]{jockey_blur_subimg8}\\[-.7em]
    \end{minipage}%
  }\\
  \caption{
    (a) A color image. (b) A corrupted version by a directional blur and
    Gaussian noise with standard deviation $\sigma=2$.
    (c) Solution of TViso.
    Debiased solution with (d) HD, (e) QO and (f) SD.
    The PSNR is indicated in brackets bellow each image.}
  \label{fig:deblurring}
\end{figure*}

\subsection{Experiments with second order TGV}\label{sec:tgv}

We finaly consider the second order Total Generalized Variation (TGV)
regularization \cite{bredies2010total} of grayscale images
that can be expressed, for an image $x \in \RR^n$ and regularization
parameters $\lambda > 0$ and $\zeta \geq 0$, as
\begin{align}\label{eq:tgv_iso}
  &\hat{x} \in \uargmin{x \in \RR^n}
  \frac12 \norm{\Phi x - y}^2_2
  \nonumber
  \\
  & \quad + \lambda \min_{z \in \RR^{n \times 2}}
  \norm{\nabla x - \zeta z}_{1,2}
  + \norm{\Ee(z)}_{1,F}
\end{align}
where $\nabla = (\nabla^1, \nabla^2) : \RR^n \to \RR^{n \times 2}$,
$\nabla^d$ denotes the forward discrete gradient in the direction $d \in \{ 1, 2 \}$,
and $\Ee : \RR^{n \times 2} \to \RR^{n \times 2 \times 2}$ is a 
symmetric tensor field
operator defined for a 2d vector field $z = (z^1, z^2)$ as
\begin{align}
  \Ee(z)
  =
  \begin{pmatrix}
    \bar \nabla^1 z^1 & \frac12(\bar \nabla^1 z^2 + \bar \nabla^2 z^1)\\
    \frac12(\bar \nabla^1 z^2 + \bar \nabla^2 z^1) & \bar \nabla^2 z^2 \\
  \end{pmatrix}
\end{align}
where $\bar \nabla^d$ denotes the backward discrete gradient in the direction $d \in \{ 1, 2 \}$,
and $\norm{\cdot}_{1,F}$ is the pointwise sum of the Frobenius norm of all matrices of a field.
One can observe that for $\zeta = 0$ this model is equivalent to TViso.
Interestingly the solution of the second order TGV  can be obtained by
solving a regularized least square problem with an $\ell_{12}$
sparse analysis term as
\begin{align}\label{eq:tgv_iso_block}
  \hat{X} \in \uargmin{X \in \RR^{n \times 3}}
  \frac12 \norm{\Xi X - y}^2_2
  +
  \lambda
  \norm{\Gamma X}_{1,2}
\end{align}
where for an image $x \in \RR^n$ and a vector field
$z \in \RR^{n \times 2}$, we consider
$  X =
  \begin{pmatrix}
    x, z
  \end{pmatrix} \in \RR^{n \times 3}$ and
 $ \Xi : (x, z) \to \Phi x
  $
and we define $\Gamma : \RR^{n \times 3} \mapsto \RR^{2 n \times 3}$ as
\begin{align}
  \Gamma(x, z)\hspace{-0.05cm}
  =\hspace{-0.1cm}
  \begin{pmatrix}
    \nabla^1 x - \zeta z^1 && \nabla^2 x - \zeta z^2 && 0\\
    \bar \nabla^1 z^1 && \bar \nabla^2 z^2 &&
    \tfrac1{\sqrt{2}} (\bar \nabla^1 z^2 + \bar \nabla^2 z^1)\\
  \end{pmatrix}\hspace{-0.08cm}.
\end{align}
After solving \eqref{eq:tgv_iso_block}, the solution  of \eqref{eq:tgv_iso} is obtained
as the first component of $\hat{X}=(\hat{x},\hat{z})$.
We use Chambolle-Pock algorithm \cite{CP} for which we also need to implement the
adjoint $\Gamma^t : \RR^{2 n \times 3} \to \RR^{n \times 3}$
of $\Gamma$ given for fields $z \in \RR^{n \times 2}$ and $e \in \RR^{n \times 3}$ by
\begin{align}
  \Gamma^t
  \begin{pmatrix}
    z^1, & z^2, & \cdot\\
    e^2, & e^2, & e^3
  \end{pmatrix}
  &=
  -
  \begin{pmatrix}
    (\bar \nabla^1 z^1 + \bar \nabla^2 z^2)^t\\
    (\zeta z^1 + \nabla^1 e^1 +\tfrac1{\sqrt{2}} \nabla^2 e^3)^t\\
    (\zeta z^2 + \nabla^1 e^2 +\tfrac1{\sqrt{2}} \nabla^2 e^3)^t\\
  \end{pmatrix}^t ~.
\end{align}

We focused on a denoising problem $y = x + w$ where
$x$ is a simulated elevation profile, ranging in $[0, 255]$, and $w$ is an additive white Gaussian
noise with standard deviation $\sigma=2$.
We chose $\lambda=15$ and $\zeta = 0.45$.
We applied the iterative primal-dual algorithm with our joint-refitting
(Algorithm \eqref{algo_final}) for $4,000$ iterations,
with $\tau = 1/\norm{\Gamma}_2$, $\kappa = 1/\norm{\Gamma}_2$ and $\theta = 1$,
and where we estimated $\norm{\Gamma}_2 \approx 3.05$ by power iteration
(note that the quantity $\norm{\Gamma}_2$ depends on the value of $\zeta$).

Refitting comparisons between  our proposed SD block penalty,
and  the five other alternatives are provided on Fig.~\ref{fig:tgv}.
Using the SD model offers the best refitting
performances in terms of both visual and quantitative measures.
The bias of TGV is well-corrected,
elevations (see the chimneys) and slopes (see the sidewalks)
are enhanced while smoothness and sharpness of TGV is
preserved. Meanwhile, the approach does not  create artifacts,
invert contrasts, or reintroduce information
that was not recovered by TGV.

\begin{figure*}
  \subfigure[Original]{\includegraphics[width=.32\linewidth,viewport=50 30 360 230,clip]{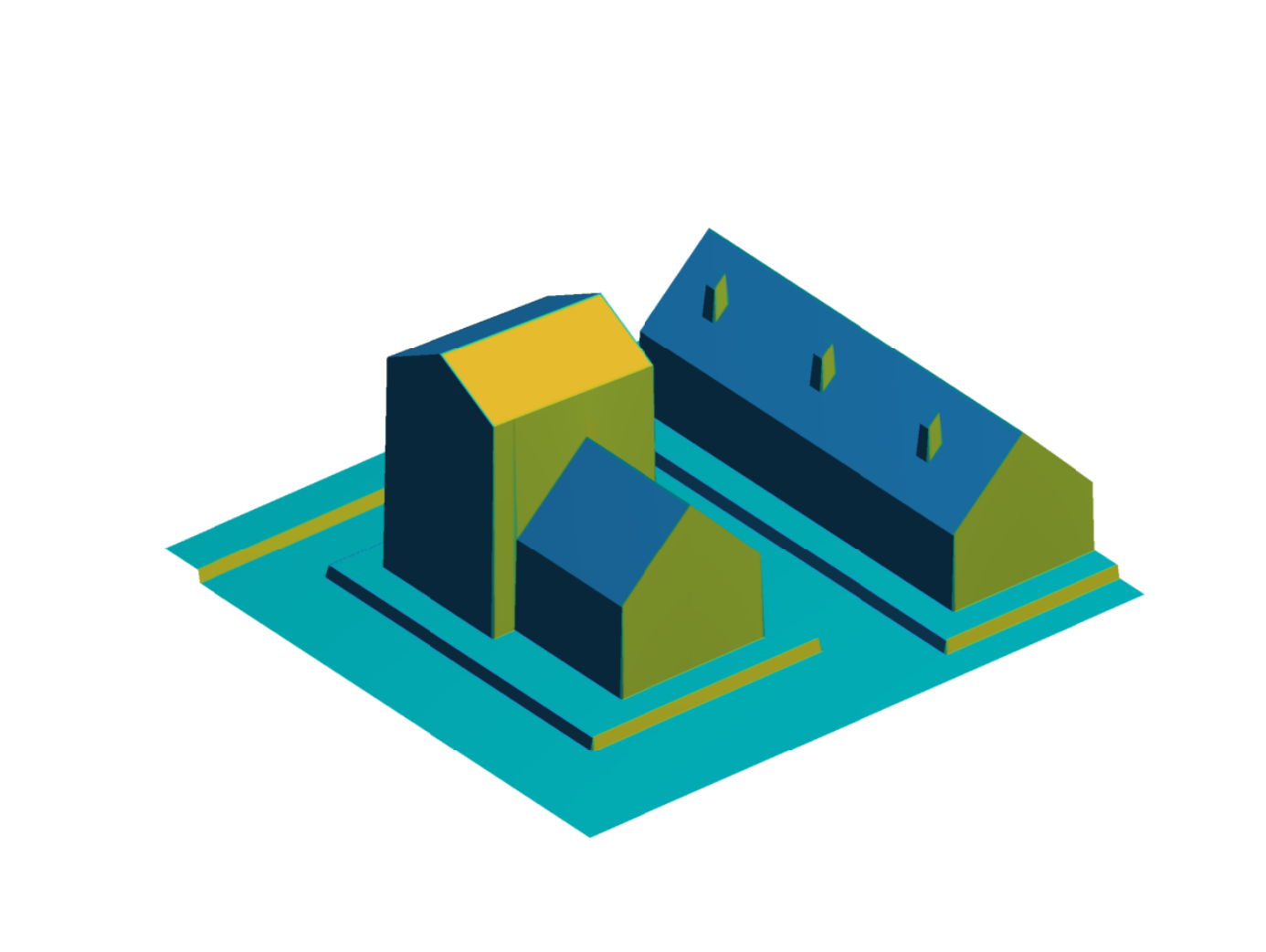}}\hfill%
  \subfigure[Noisy (42.14)]{\includegraphics[width=.32\linewidth,viewport=50 30 360 230,clip]{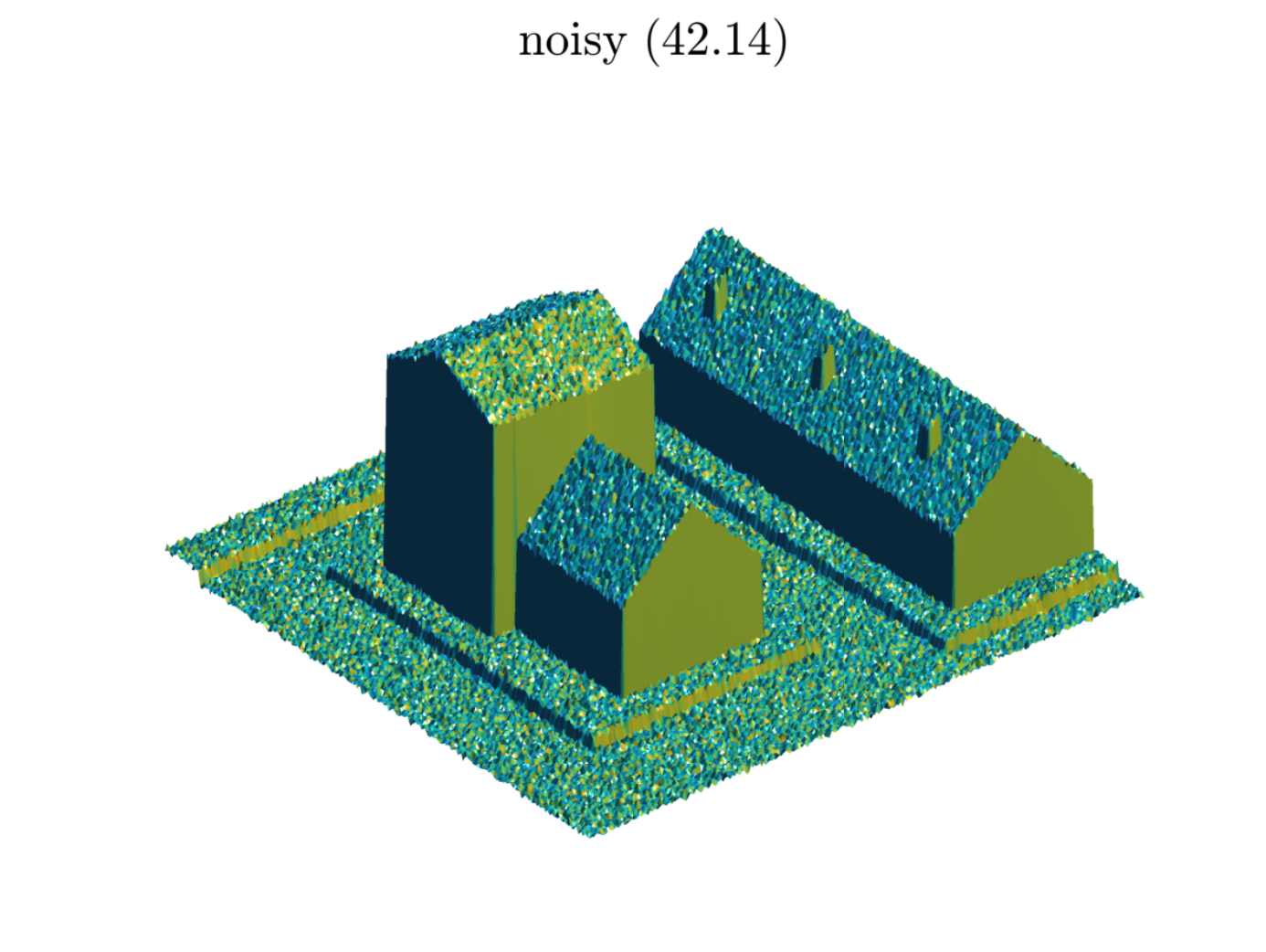}}\hfill%
  \subfigure[TGV (43.16)]{\includegraphics[width=.32\linewidth,viewport=50 30 360 230,clip]{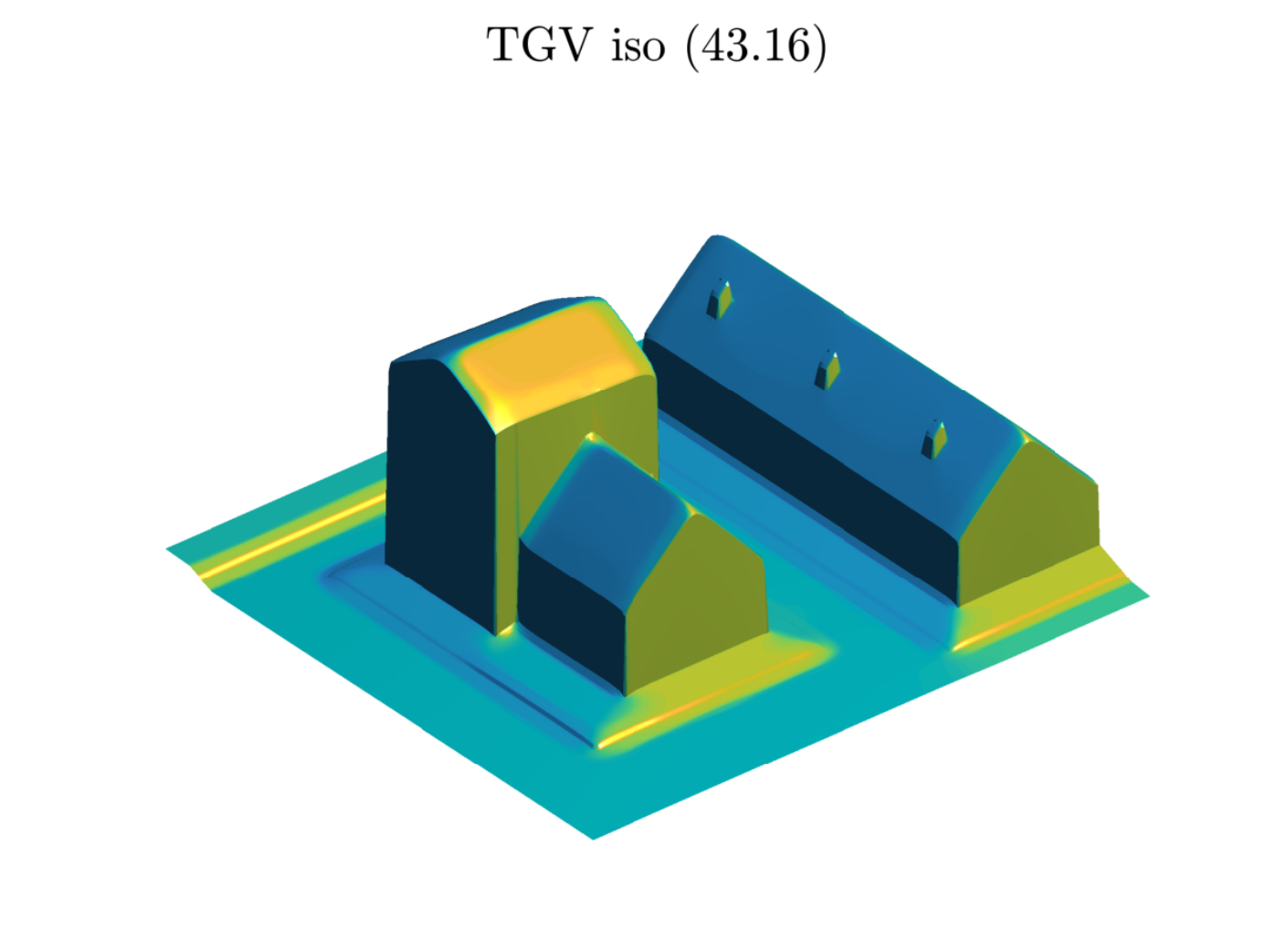}}\\
  \subfigure[HO (43.60)]{\includegraphics[width=.32\linewidth,viewport=50 30 360 230,clip]{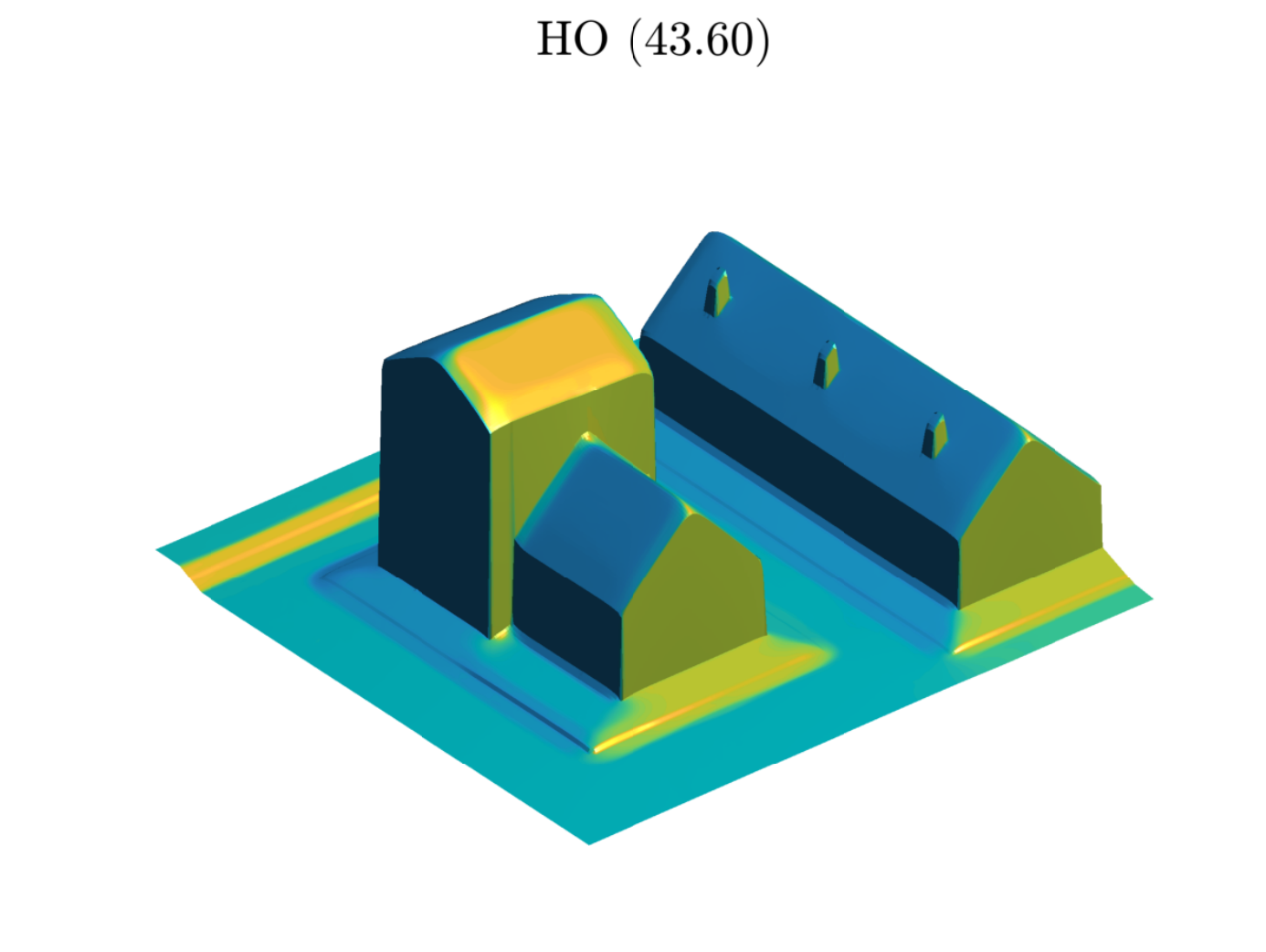}}
  \subfigure[HD (43.55)]{\includegraphics[width=.32\linewidth,viewport=50 30 360 230,clip]{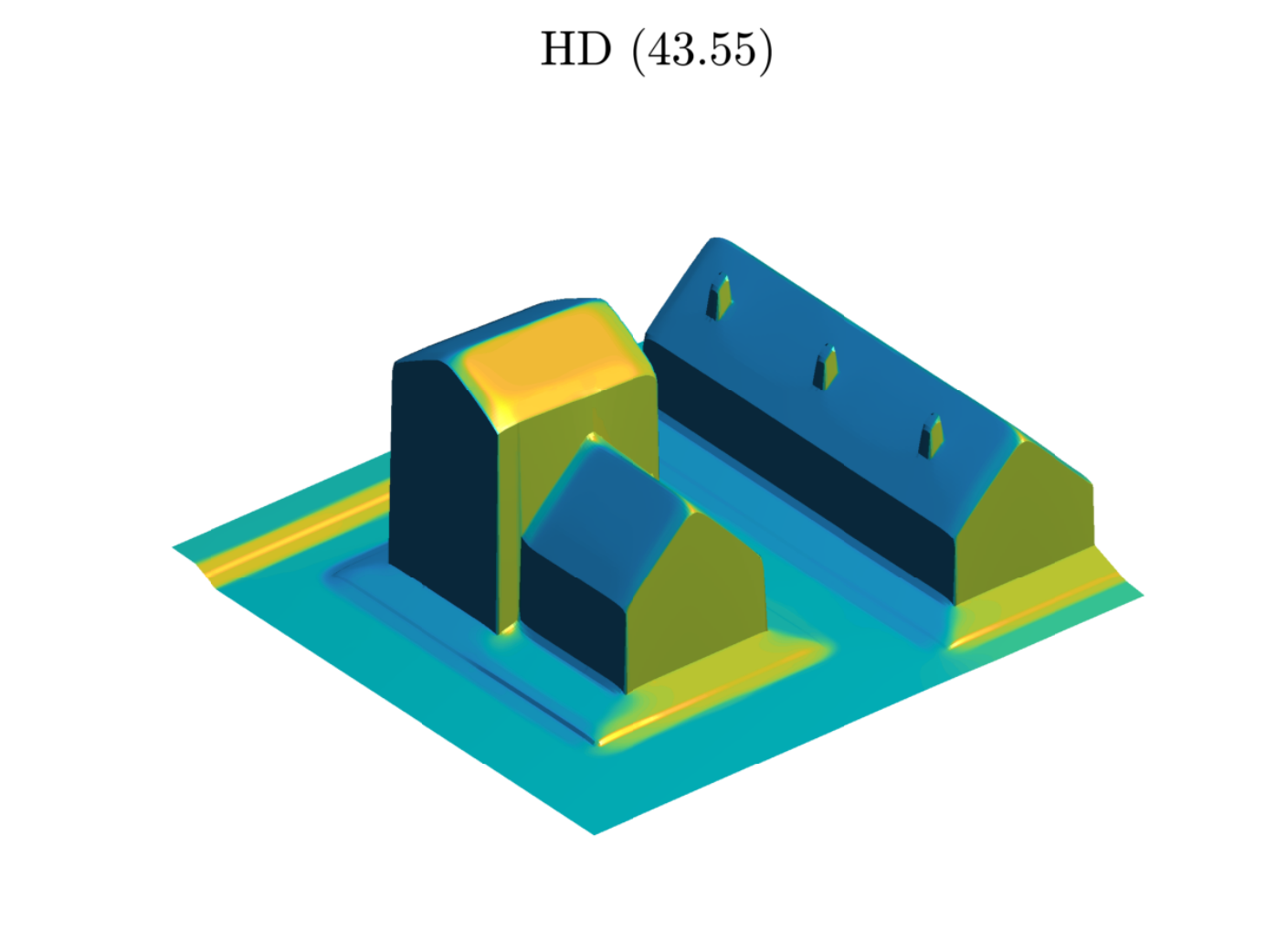}}\hfill%
  \subfigure[QO (58.31)]{\includegraphics[width=.32\linewidth,viewport=50 30 360 230,clip]{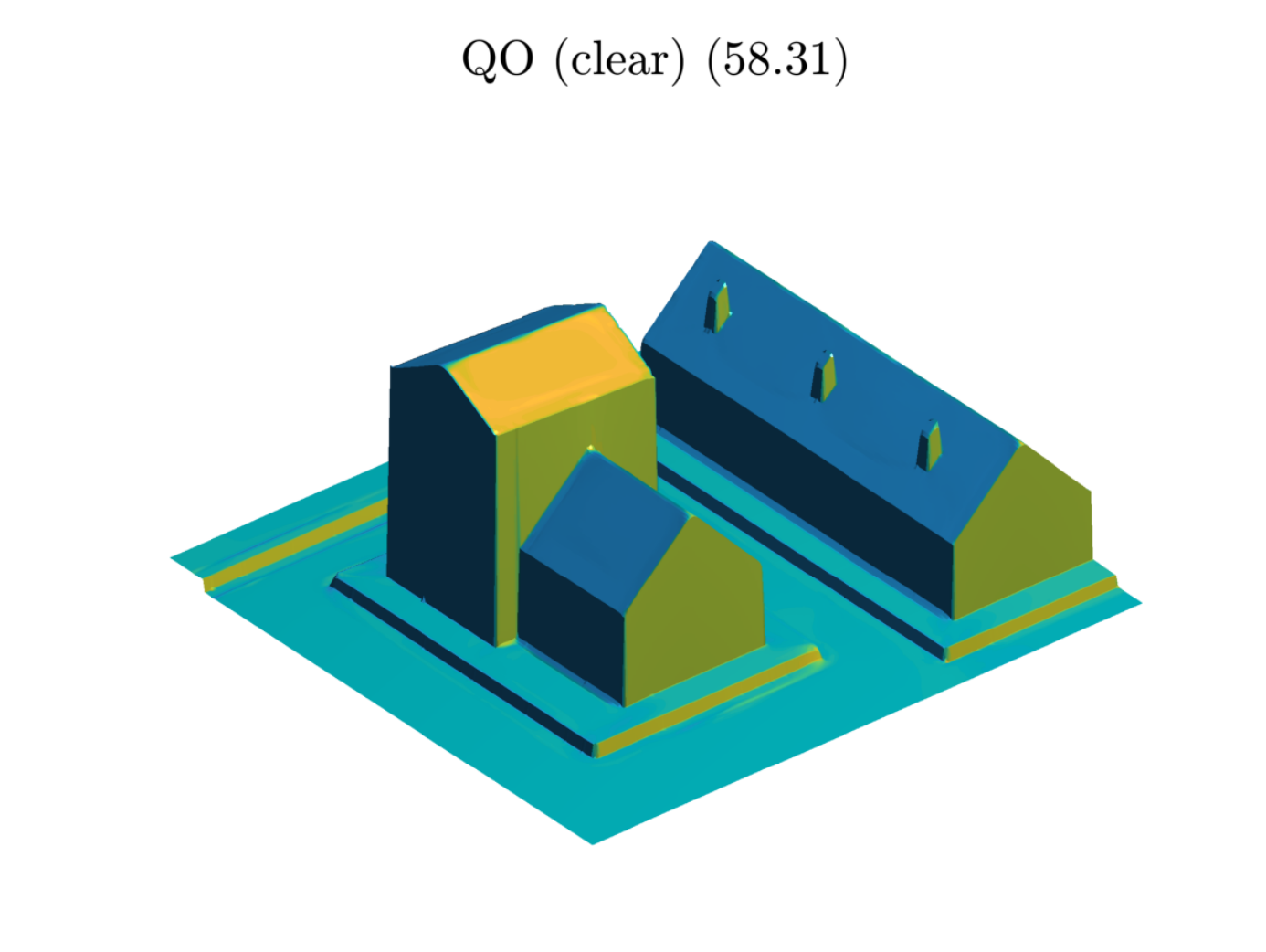}}\hfill%
  \subfigure[QD (58.00)]{\includegraphics[width=.32\linewidth,viewport=50 30 360 230,clip]{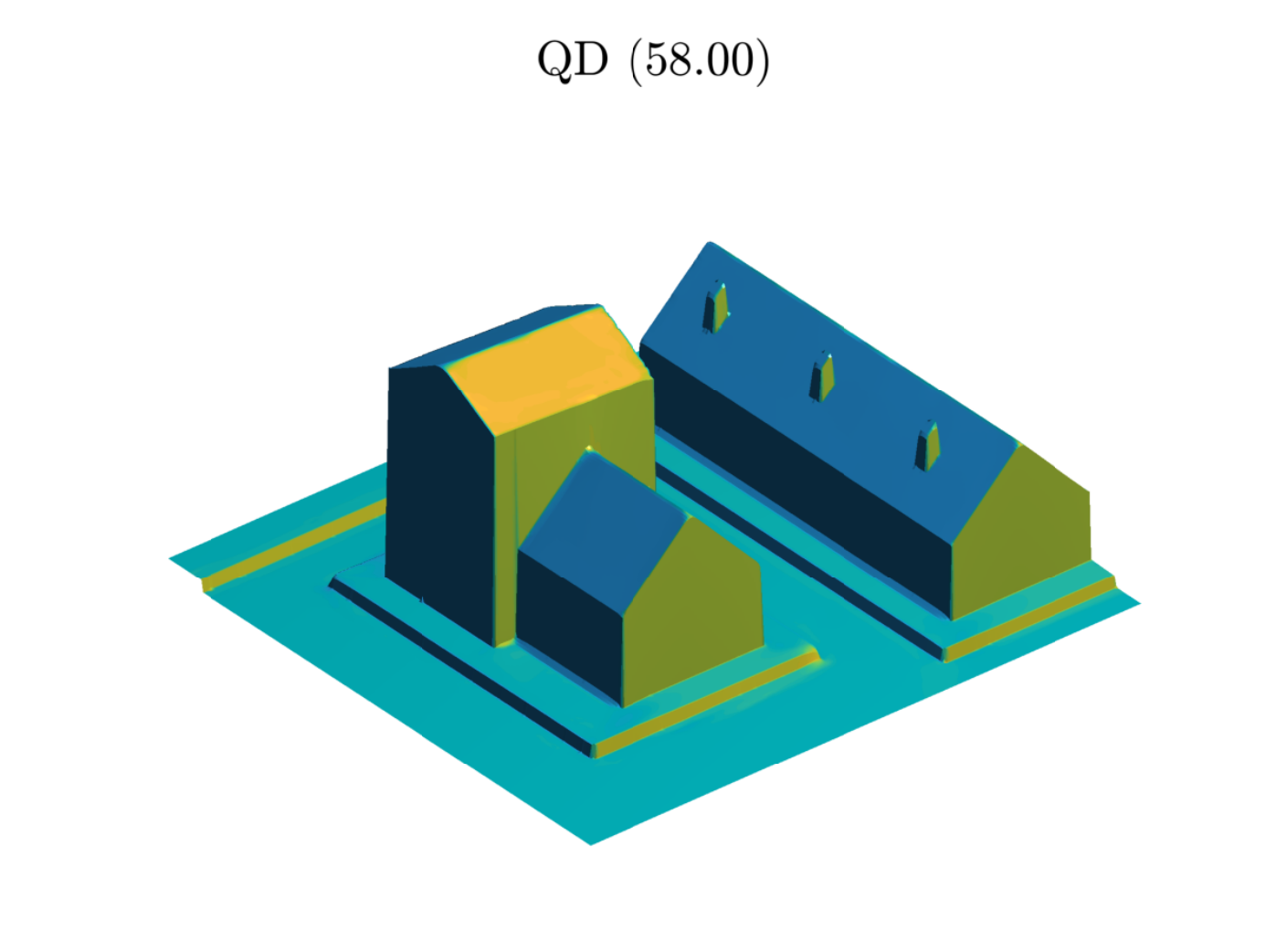}}\hfill%
  \subfigure[SO (58.31)]{\includegraphics[width=.32\linewidth,viewport=50 30 360 230,clip]{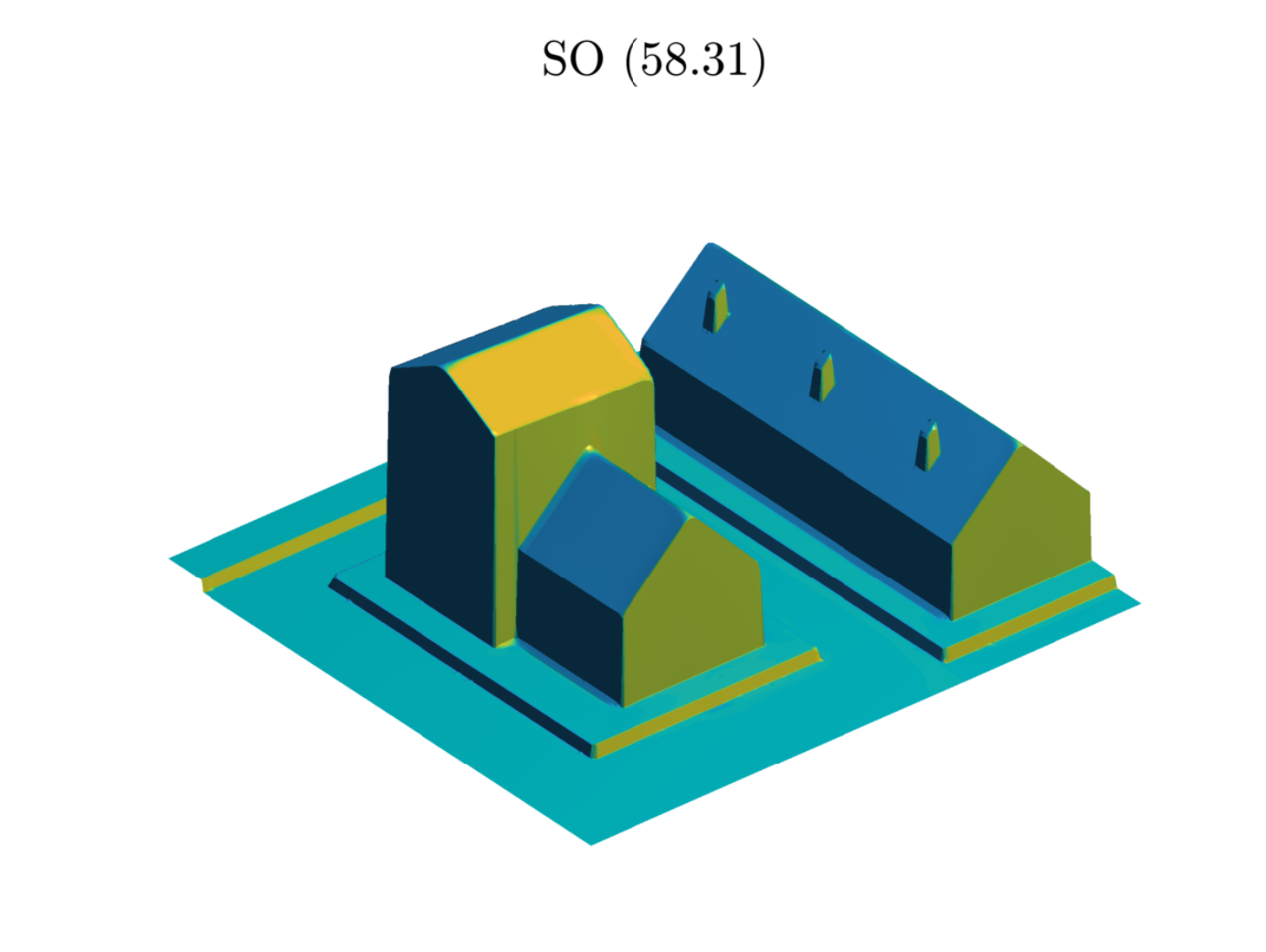}}\hfill%
  \subfigure[SD (60.33)]{\includegraphics[width=.32\linewidth,viewport=50 30 360 230,clip]{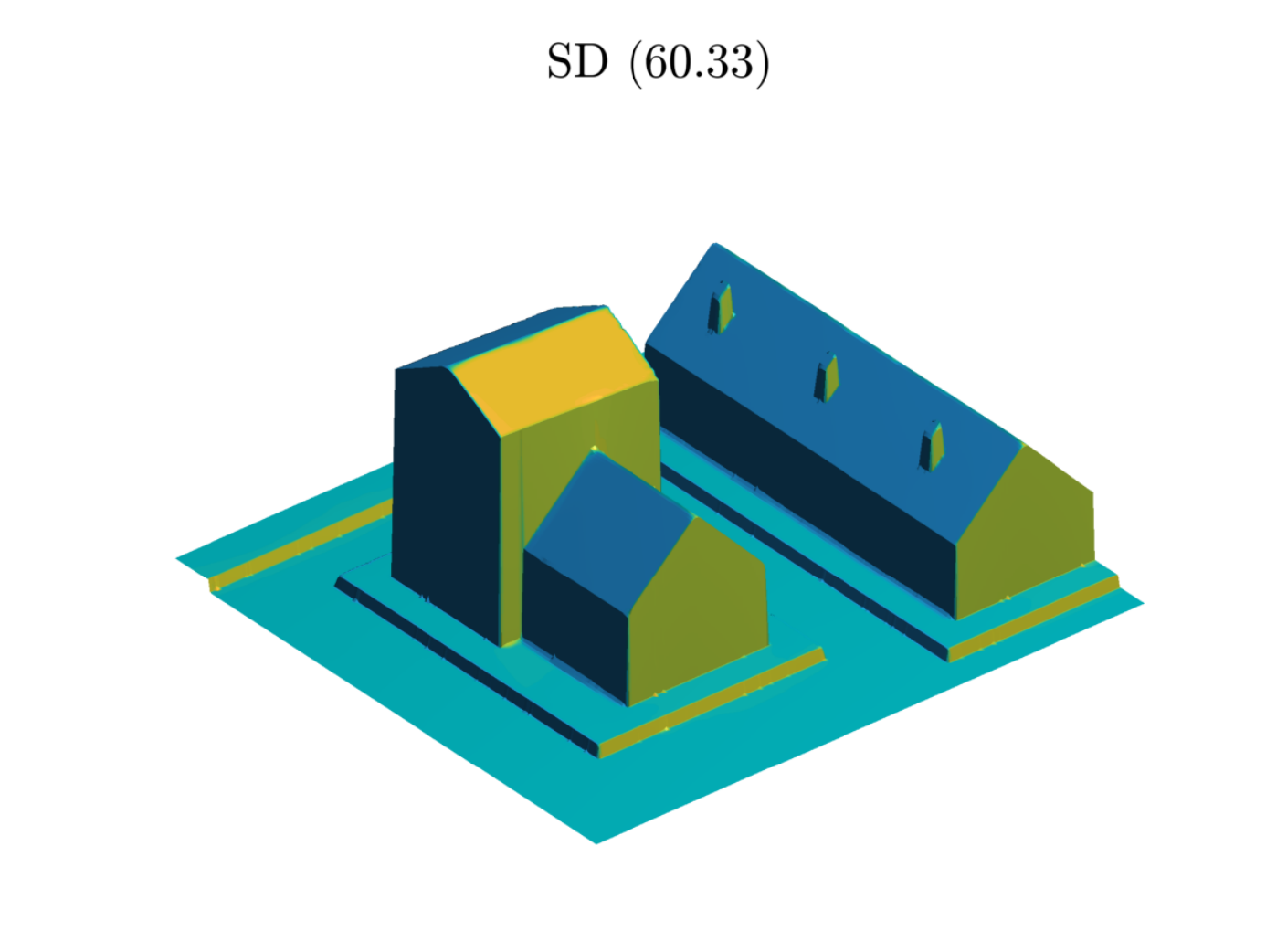}}\\
  \caption{
    (a) An elevation profile (range 0 to 255).
    (b) A corrupted version by Gaussian noise
    with standard deviation $\sigma=0.5$.
    (c) Solution of TGV.
    Debiased solution with (d) HO, (e) HD, (f) QO, (g) QD, (h) SO and (i) SD.
    The PSNR is indicated in brackets bellow each profile.
  }
  \label{fig:tgv}
\end{figure*}

\section{Conclusion}
We have presented a block penalty formulation for the refitting of solutions obtained with $\ell_{12}$ sparse regularization.
Through this framework, desirable properties of refitted solutions can  easily be promoted.
Different  block penalty functions  have been proposed and their properties discussed.  We have namely introduced the SD block penalty that interpolates between Bregman iterations \cite{Osher} and direction preservation \cite{brinkmann2016bias}.

Based on standard optimization schemes, we have also defined stable numerical strategies to jointly estimate a biased solution and its refitting.
In order to take advantage of our efficient joint-refitting algorithm,  we underline that it is important to consider  {\em simple} block penalty functions, which proximal operator  can be computed explicitly. This is the case for all the  presented block penalties,  more  complex ones having been discarded from this paper for this reason.

Initially designed for  TViso regularization, the approach  has finally been extended to a generalized TGV model. Experiments show how the block penalties are able to preserve different structures of the biased solutions, while recovering the correct  signal amplitude.

The refitting strategy is also robust to the numerical artifacts due to classical TViso discretization, and it allows to recover sharp discontinuities in any direction (see  the circle in Fig. \ref{ex:TViso-cc}), without considering  advanced schemes \cite{condat2017discrete,chambolle2020crouzeix,Caillaud}.

For the future, a challenging problem would be to define similar simple numerical schemes able to consider global properties on the refitting, such as the  preservation of the tree of shapes \cite{weiss2019contrast} of the biased solution.

\paragraph{Acknowledgements}
This work was supported by the European Union’s Horizon 2020 research and innovation programme under the Marie Skłodowska-Curie grant agreement No 777826.

\appendix

\section{Proximity operators of block penalties}
\label{sec:prox_block_proof}

\subsection{Convex conjugates $\phi^*$}
We  here compute the convex conjugate $\phi^*$ of the different  block penalties $\phi(z,\hat z)$ that only depends on $z\in \RR^b$ and where $\hat z\in \RR^b$ is a given fixed non null vector.
We consider the following representation of the vectors $z$ with respect to the $\hat z$ axis: $z=\xi \frac{\hat z}{\norm{\hat z}_2}+\zeta \frac{\hat z^\perp}{\norm{\hat z}_2}$.
This expression is valid for the case $b=2$. If $b=1$ then $z$ is only parameterized by $\xi$. When $b>2$, $\hat z^\perp$ must be understood as a subspace $S$ of dimension $b-1$ and $\zeta$ as a vector of $b-1$ components corresponding to each dimension of $S$.

With this change of variables, we have $\norm{z}^2_2=\xi^2+\norm{\zeta}^2_2$ (since $\zeta$ is of dimension $b-1$),  $\cos(z,\hat z)=\xi/\sqrt{\xi^2+\norm{\zeta}^2_2}$, $P_{\hat z}(z)=\xi \hat z/\norm{\hat z}_2$ and $z-P_{\hat z}(z)=\zeta \hat z^\perp/\norm{\hat z}_2$.
We also observe for instance that
$\abs{\cos(z,\hat z)}=1\Leftrightarrow \zeta=0_{b-1}$.
All the block penalties $\phi(z, \hat z)$ can thus be expressed as   $\phi(\xi,\zeta)$.
The convex conjugate reads
\begin{equation}\label{convex_conj}
  \phi^*(\xi_0,\zeta_0)=\sup_{\xi,\zeta} \;
  \xi_0 \xi+\dotp{\zeta_0}{\zeta}-\phi(\xi,\zeta)
  \enspace.
\end{equation}
\noindent
$\bullet$ [HO] It results that the convex conjugate $\phi_{\rm HO}^*(\xi_0,\zeta_0)$ is
\begin{align}\label{convex_conjHO}
  \sup_{\xi,\zeta} \;
  \xi_0 \xi+\dotp{\zeta_0}{\zeta}-
  \iota_{\{ 0_{b-1} \}}(\zeta)
  =
  \iota_{\{ 0 \}}(\xi_0)~.
\end{align}
$\bullet$ [HD] The convex conjugate $\phi_{\rm HD}^*(\xi_0,\zeta_0)$ is
\begin{align}\label{convex_conjHD}
  \sup_{\xi,\zeta} \;
  \xi_0 \xi+\dotp{\zeta_0}{\zeta}-
  \iota_{\RR_+^* \times \{ 0_{b-1} \}}(\xi, \zeta)
  =
  \iota_{\mathbb{R}_{-}}(\xi_0)~.
\end{align}
$\bullet$ [QO] The convex conjugate $\phi_{\rm QO}^*(\xi_0,\zeta_0)$ is
\begin{align}\label{convex_conjQO}
  & \sup_{\xi,\zeta} \;
  \xi_0 \xi+\dotp{\zeta_0}{\zeta}- \frac{\lambda}{2}
  \frac{\xi^2+\norm{\zeta}^2_2}{\norm{\hat{z}}_2}\left(1-
  \frac{\xi^2}{\xi^2+\norm{\zeta}^2_2}\right)\\
  &=
  \sup_{\xi,\zeta} \;
  \xi_0 \xi+\dotp{\zeta_0}{\zeta}- \frac{\lambda}{2}
  \frac{\norm{\zeta}^2_2}{\norm{\hat{z}}_2}\\
  &=
  \choice{
    \frac{\norm{\hat z}_2\norm{\zeta_0}^2_2}{2\lambda} & \ifq \xi_0 = 0 ~,\\
    + \infty & \otherwise~,
  }
\end{align}
since the optimality condition on $\zeta$ gives us
$\zeta^\star=\frac{\zeta_0\norm{\hat z}_2}{\lambda}$.
\medskip

\noindent
$\bullet$ [QD]
The convex conjugate $\phi_{\rm QD}^*(\xi_0,\zeta_0)$ is
\begin{align}\label{convex_conjQD}
  &
  \sup_{\xi,\zeta} \;
  \xi_0 \xi+\dotp{\zeta_0}{\zeta}-  \frac{\lambda}{2}
  \frac{\norm{\zeta}^2_2}{\norm{\hat{z}}_2}-\choice{
    \frac{\lambda}{2} \frac{\xi^2}{\norm{\hat{z}}_2} & \ifq \xi\leq 0\\
    0&\otherwise .
  }
  \\
  &=
  \choice{
    \frac{\norm{\hat z}_2(\xi_0^2+\norm{\zeta_0}^2_2)}{2\lambda} & \ifq \xi_0\leq 0 ~,\\
    + \infty & \otherwise~,
  }
\end{align}
where we used that if $\xi_0>0$, taking $\zeta=0$ and
letting $\xi\to\infty$ leads to $\phi^*_{\rm QD}(\xi_0,\zeta_0)=+\infty$,
otherwise, we used the optimality conditions on $\zeta$ and $\xi$ giving us
$\zeta^\star=\frac{\zeta_0\norm{\hat z}_2}{\lambda}$ and
$\xi^\star=\frac{\xi_0\norm{\hat z}_2}{\lambda}$.
\medskip

\noindent
$\bullet$ [SO]
The convex conjugate $\phi_{\rm SO}^*(\xi_0,\zeta_0)$ is
\begin{align}\label{convex_conjSO}
  & \sup_{\xi,\zeta} \;
  \xi_0 \xi+\dotp{\zeta_0}{\zeta} -
  \lambda \norm{\zeta}_2
  =
  \iota_{\{ 0 \} \times B_2^\lambda}(\xi_0, \zeta_0)~,
\end{align}
since the optimality condition on $\zeta$ gives us
$\norm{\zeta_0}_2 \leq \lambda$ if $\zeta^\star = 0$,
and $\zeta_0 = \lambda \frac{\zeta^\star}{\norm{\zeta^\star}_2}$
otherwise.
\medskip
\noindent
$\bullet$ [SD]
The convex conjugate $\phi_{\rm SD}^*(\xi_0,\zeta_0)$ is
\begin{align}\label{convex_conjSD}
  & \sup_{\xi,\zeta} \;
  \xi_0 \xi+\dotp{\zeta_0}{\zeta}-  \frac\lambda2
  \left(  \sqrt{\xi^2+\norm{\zeta}_2^2 }-\xi\right)
  \\
  &=\sup_{\xi,\zeta} \;
  (\xi_0+\lambda/2) \xi+\dotp{\zeta_0}{\zeta}-  \frac\lambda2
  \sqrt{\xi^2+\norm{\zeta}_2^2 }
  \\
  &=
  \choice{
    0& \ifq  \sqrt{\norm{\zeta_0}^2_2+(\xi_0+\lambda/2)^2}\leq\lambda/2~,\\
    + \infty & \otherwise~,
  }
\end{align}
since if $\sqrt{\norm{\zeta_0}^2_2+(\xi_0+\lambda/2)^2}>\lambda/2$,
then  letting  $\xi\to\sign{(\xi_0+\lambda/2)}\times \infty$ and
$\zeta\to\sign{\zeta_0}\times \infty$ leads to $\phi_{\rm SD}^*(\xi_0,\zeta_0)=+\infty$.

\subsection{Computing $\prox_{\kappa \phi^*}$}
We here give the computation of the proximal operator
$\prox_{\kappa\phi^*}(\xi_0,\zeta_0)$
of the different $\phi^*$ that is given at point $(\xi_0,\zeta_0)$ by
\begin{equation}
  \uargmin{\xi,\zeta}
  \frac1{2\kappa} \normb{
    \begin{pmatrix}\xi\\ \zeta \end{pmatrix}
    -
    \begin{pmatrix}\xi_0\\ \zeta_0 \end{pmatrix}
  }^2_2+
  \phi^*(\xi,\zeta).
\end{equation}

\noindent
$\bullet$ [HO]
The proximal operator $\prox_{\kappa\phi_{HO}^*}(\xi_0,\zeta_0)$ is given by
\begin{align}
  \uargmin{\xi,\zeta}
    \frac1{2\kappa} \normb{
    \begin{pmatrix}\xi\\ \zeta \end{pmatrix}
    -
    \begin{pmatrix}\xi_0\\ \zeta_0 \end{pmatrix}
  }^2_2+
  \iota_{\{ 0 \}}(\xi)
  =(0,\zeta_0).
\end{align}
\medskip

\noindent
$\bullet$ [HD]
The proximal operator $\prox_{\kappa\phi_{HD}^*}(\xi_0,\zeta_0)$ is given by
\begin{align}
  \uargmin{\xi,\zeta}
    \frac1{2\kappa} \normb{
    \begin{pmatrix}\xi\\ \zeta \end{pmatrix}
    -
    \begin{pmatrix}\xi_0\\ \zeta_0 \end{pmatrix}
  }^2_2+
  \iota_{\RR_{-}}(\xi)
  =(\min(0,\xi_0),\zeta_0).
\end{align}

\noindent
$\bullet$ [QO]
The proximal operator $\prox_{\kappa\phi_{QO}^*}(\xi_0,\zeta_0)$ is given by
\begin{align}
  &\uargmin{\xi,\zeta}
  \frac1{2\kappa} \normb{
    \begin{pmatrix}\xi\\ \zeta \end{pmatrix}
    -
    \begin{pmatrix}\xi_0\\ \zeta_0 \end{pmatrix}
  }^2_2
  +\choice{
    \frac{\norm{\hat z}_2\norm{\zeta}^2_2}{2\lambda} & \ifq \xi = 0 ~,\\
    + \infty & \otherwise~.
  }\\
  &=\frac{\lambda}{\lambda+\kappa \norm{\hat z}_2}\left(0,\zeta_0\right),
\end{align}
since the optimality condition gives
$\lambda(\zeta^\star-\zeta_0)+\kappa \norm{\hat z}_2 \zeta^\star=0$.
\medskip

\noindent
$\bullet$ [QD]
The proximal operator $\prox_{\kappa\phi_{QD}^*}(\xi_0,\zeta_0)$ is given by
\begin{align}
  & \uargmin{\xi,\zeta}
  \frac1{2\kappa} \normb{
    \begin{pmatrix}\xi\\ \zeta \end{pmatrix}
    -
    \begin{pmatrix}\xi_0\\ \zeta_0 \end{pmatrix}
  }^2_2
  +\choice{
    \frac{\norm{\hat z}_2(\xi^2+\norm{\zeta}^2_2)}{2\lambda} & \ifq \xi\leq 0 ~,\\
    + \infty & \otherwise~.
  } \nonumber\\
  &=\frac{\lambda}{\lambda+\kappa \norm{\hat z}_2}\left(\min(0,\xi_0),\zeta_0\right).
\end{align}
\medskip
\noindent
$\bullet$ [SO]
The proximal operator $\prox_{\kappa\phi_{QD}^*}(\xi_0,\zeta_0)$ is given by
\begin{align}
  & \uargmin{\xi,\zeta}
  \frac1{2\kappa} \normb{
    \begin{pmatrix}\xi\\ \zeta \end{pmatrix}
    -
    \begin{pmatrix}\xi_0\\ \zeta_0 \end{pmatrix}
  }^2_2
  +
  \iota_{\{ 0 \} \times B_2^\lambda}(\xi_0, \zeta_0)
  \nonumber
  \\
  &=
  \frac{\lambda}{\max(\lambda, \norm{\zeta_0}_2)} (0, \zeta_0)~.
\end{align}
\medskip

\noindent
$\bullet$ [SD]
The proximal operator $\prox_{\kappa\phi_{QD}^*}(\xi_0,\zeta_0)$ is given by
\begin{align}
  &
  \uargmin{\xi,\zeta}
  \frac1{2\kappa} \normb{
    \begin{pmatrix}\xi\\ \zeta \end{pmatrix}
    -
    \begin{pmatrix}\xi_0\\ \zeta_0 \end{pmatrix}
  }^2_2
  \nonumber
  \\
  & \hspace{1.5cm}
  +\choice{
    0& \ifq  \sqrt{\norm{\zeta}^2_2+(\xi+\lambda/2)^2}\leq\lambda/2~,\\
    + \infty & \otherwise~.
  }\\
  & = \frac\lambda{2}\frac{(\xi_0+\lambda/2,\zeta_0) }{\max(\lambda/2,\sqrt{\norm{\zeta_0}^2_2+(\xi_0+\lambda/2)^2})} -\left(\lambda/{2},0\right),
\end{align}
which just corresponds to the projection of the $\ell_2$ ball of $\RR^b$ of radius $\lambda/2$ and center $(-\lambda/2,0)$.

%
%

\bibliographystyle{spmpsci}      

\begin{thebibliography}{10}
\providecommand{\url}[1]{{#1}}
\providecommand{\urlprefix}{URL }
\expandafter\ifx\csname urlstyle\endcsname\relax
  \providecommand{\doi}[1]{DOI~\discretionary{}{}{}#1}\else
  \providecommand{\doi}{DOI~\discretionary{}{}{}\begingroup
  \urlstyle{rm}\Url}\fi

\bibitem{bach2012optimization}
Bach, F., Jenatton, R., Mairal, J., Obozinski, G., et~al.: Optimization with
  sparsity-inducing penalties.
\newblock Foundations and Trends{\textregistered} in Machine Learning
  \textbf{4}(1), 1--106 (2012)

\bibitem{ballester2006variational}
Ballester, C., Caselles, V., Igual, L., Verdera, J., Roug{\'e}, B.: A
  variational model for p+ xs image fusion.
\newblock International Journal of Computer Vision \textbf{69}(1), 43--58
  (2006)

\bibitem{Belloni_Chernozhukov13}
Belloni, A., Chernozhukov, V.: Least squares after model selection in
  high-dimensional sparse models.
\newblock Bernoulli \textbf{19}(2), 521--547 (2013)

\bibitem{boyer2019representer}
Boyer, C., Chambolle, A., Castro, Y.D., Duval, V., De~Gournay, F., Weiss, P.:
  On representer theorems and convex regularization.
\newblock SIAM Journal on Optimization \textbf{29}(2), 1260--1281 (2019)

\bibitem{bredies2020sparsity}
Bredies, K., Carioni, M.: Sparsity of solutions for variational inverse
  problems with finite-dimensional data.
\newblock Calculus of Variations and Partial Differential Equations
  \textbf{59}(1), 14 (2020)

\bibitem{bredies2010total}
Bredies, K., Kunisch, K., Pock, T.: Total generalized variation.
\newblock SIAM Journal on Imaging Sciences \textbf{3}(3), 492--526 (2010)

\bibitem{brinkmann2016bias}
Brinkmann, E.M., Burger, M., Rasch, J., Sutour, C.: Bias-reduction in
  variational regularization.
\newblock J. of Math. Imaging and Vision  (2017)

\bibitem{Buhlmann_Yu03}
B{\"u}hlmann, P., Yu, B.: Boosting with the {L2} loss: regression and
  classification.
\newblock J. Am. Statist. Assoc. \textbf{98}(462), 324--339 (2003)

\bibitem{BGMEC16}
Burger, M., Gilboa, G., Moeller, M., Eckardt, L., Cremers, D.: Spectral
  decompositions using one-homogeneous functionals.
\newblock SIAM J. on Imaging Sciences \textbf{9}(3), 1374--1408 (2016)

\bibitem{burger2019convergence}
Burger, M., Korolev, Y., Rasch, J.: Convergence rates and structure of
  solutions of inverse problems with imperfect forward models.
\newblock Inverse Problems \textbf{35}(2), 024006 (2019)

\bibitem{Caillaud}
Caillaud, C.: Asymptotical estimates for some algorithms for data and image
  processing: a study of the sinkhorn algorithm and a numerical analysis of
  total variation minimization.
\newblock Ph.D. thesis, Institut Polytechnique de Paris (2020)

\bibitem{chambolle2010introduction}
Chambolle, A., Caselles, V., Cremers, D., Novaga, M., Pock, T.: An introduction
  to total variation for image analysis.
\newblock Theoretical foundations and numerical methods for sparse recovery
  \textbf{9}(263-340), 227 (2010)

\bibitem{CP}
Chambolle, A., Pock, T.: A first-order primal-dual algorithm for convex
  problems with applications to imaging.
\newblock J. of Math. Imaging and Vision \textbf{40}, 120--145 (2011)

\bibitem{chambolle2020crouzeix}
Chambolle, A., Pock, T.: Crouzeix--raviart approximation of the total variation
  on simplicial meshes.
\newblock Journal of Mathematical Imaging and Vision pp. 1--28 (2020)

\bibitem{chan2000high}
Chan, T., Marquina, A., Mulet, P.: High-order total variation-based image
  restoration.
\newblock SIAM Journal on Scientific Computing \textbf{22}(2), 503--516 (2000)

\bibitem{charest2008iterative}
Charest, M.R., Milanfar, P.: On iterative regularization and its application.
\newblock IEEE Transactions on Circuits And Systems For Video Technology
  \textbf{18}(3), 406--411 (2008)

\bibitem{chzhen2019}
Chzhen, E., Hebiri, M., Salmon, J.: On lasso refitting strategies.
\newblock Bernoulli \textbf{25}(4A), 3175--3200 (2019)

\bibitem{combettes2007douglas}
Combettes, P.L., Pesquet, J.C.: A douglas--rachford splitting approach to
  nonsmooth convex variational signal recovery.
\newblock Selected Topics in Signal Processing, IEEE Journal of \textbf{1}(4),
  564--574 (2007)

\bibitem{condat2017discrete}
Condat, L.: Discrete total variation: New definition and minimization.
\newblock SIAM Journal on Imaging Sciences \textbf{10}(3), 1258--1290 (2017)

\bibitem{deledalle2015contrast}
Deledalle, C.A., Papadakis, N., Salmon, J.: Contrast re-enhancement of
  total-variation regularization jointly with the douglas-rachford iterations.
\newblock In: Signal Processing with Adaptive Sparse Structured Representations
  (2015)

\bibitem{Deledalle_Papadakis_Salmon15}
Deledalle, C.A., Papadakis, N., Salmon, J.: {On debiasing restoration
  algorithms: applications to total-variation and nonlocal-means}.
\newblock In: Int. Conf. on Scale Space and Variational Methods in Computer
  Vision, pp. 129--141. Springer (2015)

\bibitem{deledalle2016clear}
Deledalle, C.A., Papadakis, N., Salmon, J., Vaiter, S.: {CLEAR}: Covariant
  least-square re-fitting with applications to image restoration.
\newblock SIAM J. on Imaging Sciences \textbf{10}(1), 243--284 (2017)

\bibitem{deledalle2019refitting}
Deledalle, C.A., Papadakis, N., Salmon, J., Vaiter, S.: Refitting solutions
  promoted by $\ell_{12}$ sparse analysis regularizations with block penalties.
\newblock In: International Conference on Scale Space and Variational Methods
  in Computer Vision, pp. 131--143. Springer (2019)

\bibitem{dobson1996recovery}
Dobson, D.C., Santosa, F.: Recovery of blocky images from noisy and blurred
  data.
\newblock SIAM Journal on Applied Mathematics \textbf{56}(4), 1181--1198 (1996)

\bibitem{douglas1956numerical}
Douglas, J., Rachford, H.: On the numerical solution of heat conduction
  problems in two and three space variables.
\newblock Transactions of the American mathematical Society pp. 421--439 (1956)

\bibitem{efron2004least}
Efron, B., Hastie, T., Johnstone, I., Tibshirani, R.: Least angle regression.
\newblock Ann. Statist. \textbf{32}(2), 407--499 (2004)

\bibitem{elad2007analysis}
Elad, M., Milanfar, P., Rubinstein, R.: Analysis versus synthesis in signal
  priors.
\newblock Inverse problems \textbf{23}(3), 947--968 (2007)

\bibitem{esedoglu2004decomposition}
Esedoḡlu, S., Osher, S.J.: Decomposition of images by the anisotropic
  rudin-osher-fatemi model.
\newblock Commun. Pure Appl. Math. \textbf{57}(12), 1609--1626 (2004)

\bibitem{gilboa2014total}
Gilboa, G.: A total variation spectral framework for scale and texture
  analysis.
\newblock SIAM J. Imaging Sci. \textbf{7}(4), 1937--1961 (2014)

\bibitem{Lederer13}
Lederer, J.: Trust, but verify: benefits and pitfalls of least-squares
  refitting in high dimensions.
\newblock arXiv preprint arXiv:1306.0113  (2013)

\bibitem{Lin_Zhang06}
Lin, Y., Zhang, H.H.: Component selection and smoothing in multivariate
  nonparametric regression.
\newblock Ann. Statist. \textbf{34}(5), 2272--2297 (2006)

\bibitem{lions1979splitting}
Lions, P.L., Mercier, B.: Splitting algorithms for the sum of two nonlinear
  operators.
\newblock SIAM Journal on Numerical Analysis \textbf{16}(6), 964--979 (1979)

\bibitem{liu2015image}
Liu, J., Huang, T.Z., Selesnick, I.W., Lv, X.G., Chen, P.Y.: Image restoration
  using total variation with overlapping group sparsity.
\newblock Information Sciences \textbf{295}, 232--246 (2015)

\bibitem{milanfar2013tour}
Milanfar, P.: A tour of modern image filtering: New insights and methods, both
  practical and theoretical.
\newblock IEEE Signal Processing Magazine \textbf{30}(1), 106--128 (2013)

\bibitem{nam2013cosparse}
Nam, S., Davies, M., Elad, M., Gribonval, R.: The cosparse analysis model and
  algorithms.
\newblock Appl. Comput. Harmon. Anal. \textbf{34}, 30--56 (2013)

\bibitem{Osher}
Osher, S., Burger, M., Goldfarb, D., Xu, J., Yin, W.: An iterative
  regularization method for total variation-based image restoration.
\newblock Multiscale Model. \& Simul. \textbf{4}(2), 460--489 (2005)

\bibitem{peyre2011group}
Peyr{\'e}, G., Fadili, J.: Group sparsity with overlapping partition functions.
\newblock In: 2011 19th European Signal Processing Conference, pp. 303--307.
  IEEE (2011)

\bibitem{peyre2011dyadic}
Peyr{\'e}, G., Fadili, J., Chesneau, C.: Adaptive structured block sparsity via
  dyadic partitioning.
\newblock In: Proc. EUSIPCO 2011, pp. 1455--1459 (2011)

\bibitem{pierre2017luminance}
Pierre, F., Aujol, J.F., Deledalle, C.A., Papadakis, N.: Luminance-guided
  chrominance denoising with debiased coupled total variation.
\newblock In: International Workshop on Energy Minimization Methods in Computer
  Vision and Pattern Recognition, pp. 235--248. Springer (2017)

\bibitem{Rigollet_Tsybakov11}
Rigollet, P., Tsybakov, A.B.: Exponential screening and optimal rates of sparse
  estimation.
\newblock Ann. Statist. \textbf{39}(2), 731--471 (2011)

\bibitem{romano2015boosting}
Romano, Y., Elad, M.: Boosting of image denoising algorithms.
\newblock SIAM J. Imaging Sci. \textbf{8}(2), 1187--1219 (2015)

\bibitem{rudin1992nonlinear}
Rudin, L.I., Osher, S., Fatemi, E.: Nonlinear total variation based noise
  removal algorithms.
\newblock Physica D: nonlinear phenomena \textbf{60}(1-4), 259--268 (1992)

\bibitem{Schetzer01}
Scherzer, O., Groetsch, C.: Inverse scale space theory for inverse problems.
\newblock In: M.~Kerckhove (ed.) Scale-Space and Morphology in Computer Vision,
  pp. 317--325. Springer Berlin Heidelberg (2001)

\bibitem{strong2003edge}
Strong, D., Chan, T.: Edge-preserving and scale-dependent properties of total
  variation regularization.
\newblock Inverse problems \textbf{19}(6), 165--187 (2003)

\bibitem{Tadmor04}
Tadmor, E., Nezzar, S., Vese, L.: A multiscale image representation using
  hierarchical ({BV},{L}2) decompositions.
\newblock Multiscale Model. \& Simul. \textbf{2}(4), 554--579 (2004)

\bibitem{talebi2013saif}
Talebi, H., Zhu, X., Milanfar, P.: How to saif-ly boost denoising performance.
\newblock IEEE Transactions on Image Processing \textbf{22}(4), 1470--1485
  (2013)

\bibitem{tibshirani1996regression}
Tibshirani, R.: Regression shrinkage and selection via the lasso.
\newblock J. R. Stat. Soc. Ser. B Stat. Methodol. pp. 267--288 (1996)

\bibitem{Tukey77}
Tukey, J.W.: Exploratory data analysis.
\newblock Reading, Mass. (1977)

\bibitem{vaiter2017degrees}
Vaiter, S., Deledalle, C., Fadili, J., Peyr{\'e}, G., Dossal, C.: The degrees
  of freedom of partly smooth regularizers.
\newblock Annals of the Institute of Statistical Mathematics \textbf{69}(4),
  791--832 (2017)

\bibitem{vaiter2013local}
Vaiter, S., Deledalle, C.A., Peyr{\'e}, G., Dossal, C., Fadili, J.: Local
  behavior of sparse analysis regularization: Applications to risk estimation.
\newblock Applied and Computational Harmonic Analysis \textbf{35}(3), 433--451
  (2013)

\bibitem{vaiter2012robust}
Vaiter, S., Peyr{\'e}, G., Dossal, C., Fadili, J.: Robust sparse analysis
  regularization.
\newblock IEEE Transactions on Information Theory \textbf{59}(4), 2001--2016
  (2012)

\bibitem{vaiter2017model}
Vaiter, S., Peyr{\'e}, G., Fadili, J.: Model consistency of partly smooth
  regularizers.
\newblock IEEE Transactions on Information Theory \textbf{64}(3), 1725--1737
  (2017)

\bibitem{weiss2019contrast}
Weiss, P., Escande, P., Bathie, G., Dong, Y.: Contrast invariant {SNR} and
  isotonic regressions.
\newblock International Journal of Computer Vision pp. 1--18 (2019)

\bibitem{yu2008audio}
Yu, G., Mallat, S., Bacry, E.: Audio denoising by time-frequency block
  thresholding.
\newblock IEEE Transactions on Signal processing \textbf{56}(5), 1830--1839
  (2008)

\bibitem{Yuan_Lin06}
Yuan, M., Lin, Y.: Model selection and estimation in regression with grouped
  variables.
\newblock J. R. Stat. Soc. Ser. B Stat. Methodol. \textbf{68}(1), 49--67 (2006)

\end{thebibliography}

\end{document}